\UseRawInputEncoding

\documentclass{amsart}
\usepackage{mathrsfs}
\usepackage{graphicx}
\usepackage{amsmath}
\usepackage{amscd}
\usepackage{cite}
\usepackage{hyperref}
\usepackage{epsfig}
\usepackage{subfigure}
\usepackage{caption}
\usepackage{tikz}
  \usepackage{float}
  \usepackage{amssymb}
\usepackage{amsfonts}
\usepackage[all,cmtip]{xy}
\usepackage{appendix}

\newtheorem{Example}{Example}[section]
\newtheorem{Definition}{Definition}[section]
\newtheorem{Theorem}{Theorem}[section]
\newtheorem{Theorem/Definition}{Theorem/Definition}[section]
\newtheorem{Proposition}{Proposition}[section]
\newtheorem{Lemma}{Lemma}[section]

\newtheorem{Conjecture}{Conjecture}[section]

\newcommand{\pd}{\partial}
\newcommand{\bC}{{\mathbb C}}

\newcommand{\bQ}{{\mathbb Q}}

\newcommand{\bZ}{{\mathbb Z}}

\newcommand{\cC}{{\mathcal C}}

\newcommand{\cF}{{\mathcal F}}

\newcommand{\cH}{{\mathcal H}}
\newcommand{\cI}{{\mathcal I}}
\newcommand{\cJ}{{\mathcal J}}
\newcommand{\cL}{{\mathcal L}}
\newcommand{\cM}{{\mathcal M}}

\newcommand{\cP}{{\mathcal P}}

\newcommand{\cS}{{\mathcal S}}
\newcommand{\cT}{{\mathcal T}}
\newcommand{\half}{\frac{1}{2}}
\newcommand{\cV}{{\mathcal V}}
\newcommand{\cW}{{\mathcal W}}

\newcommand{\cZ}{{\mathcal Z}}

\newcommand{\Mbar}{\overline{\cM}}

\newcommand{\tK}{{\widetilde K}}

\newcommand{\wcL}{{\widehat{\mathcal{L}}}}

\newcommand{\tcW}{{\widetilde{\mathcal{W}}}}

\newcommand{\tD}{{\widetilde{D}}}

\newcommand{\tE}{{\widetilde{E}}}
\newcommand{\fs}{{\mathfrak{s}}}
\newcommand{\fj}{{\mathfrak{j}}}
\newcommand{\ff}{{\mathfrak{f}}}
\newcommand{\cU}{{\mathcal{U}}}

\newcommand{\be}{\begin{equation}}
\newcommand{\ee}{\end{equation}}
\newcommand{\bea}{\begin{eqnarray}}
\newcommand{\ben}{\begin{eqnarray*}}
\newcommand{\een}{\end{eqnarray*}}
\newcommand{\eea}{\end{eqnarray}}

\DeclareMathOperator{\Aut}{Aut}
\DeclareMathOperator{\Id}{id}

\DeclareMathOperator{\val}{val}
\DeclareMathOperator{\tr}{tr}
\DeclareMathOperator{\diag}{diag}
\DeclareMathOperator{\perm}{perm}
\DeclareMathOperator{\Res}{Res}
\DeclareMathOperator{\Sym}{Sym}

\usepackage{color}

\definecolor{yellow}{rgb}{1,1,0}
\definecolor{orange}{rgb}{1,.7,0}
\definecolor{red}{rgb}{1,0,0}
\definecolor{green}{rgb}{0,1,1}
\definecolor{white}{rgb}{1,1,1}

\definecolor{A}{rgb}{.75,1,.75}

\theoremstyle{remark}
\newtheorem{Remark}{Remark}[section]

\begin{document}

\newtheorem{myDef}{Definition}
\newtheorem{thm}{Theorem}
\newtheorem{eqn}{equation}

\title[A formalism of abstract QFT of summation of fat graphs]
{A formalism of abstract quantum field theory of summation of fat graphs}

\author{Zhiyuan Wang}
\address{School of Mathematical Sciences\\
Peking University\\Beijing, 100871, China}
\email{zhiyuan19@math.pku.edu.cn}

\author{Jian Zhou}
\address{Department of Mathematical Sciences\\
Tsinghua University\\Beijing, 100084, China}
\email{jianzhou@mail.tsinghua.edu.cn}

\begin{abstract}

In this work we present a formalism of abstract quantum field theory
for fat graphs and its realizations.
This is a generalization of an earlier work for stable graphs.
We define the abstract correlators $\mathcal F_g^\mu$,
abstract free energy $\mathcal F_g$,
abstract partition function $\mathcal Z$,
and abstract $n$-point functions $\mathcal W_{g,n}$ to be formal summations of fat graphs,
and derive quadratic recursions using edge-contraction/vertex-splitting operators,
including the abstract Virasoro constraints,
an abstract cut-and-join type representation
for $\mathcal Z$,
and a quadratic recursion for $\mathcal W_{g,n}$ which resembles
the Eynard-Orantin topological recursion.
When considering the realization by the Hermitian one-matrix models,
we obtain the Virasoro constraints,
a cut-and-join representation for
the partition function $Z_N^{\text{Herm}}$ which proves that $Z_N^{\text{Herm}}$ is
a tau-function of KP hierarchy,
a recursion for $n$-point functions which is known to be
equivalent to the E-O recursion,
and a Schr\"odinger type-equation which is equivalent to the quantum spectral curve.
We conjecture that in general cases the realization of the
quadratic recursion for $\mathcal W_{g,n}$
is the E-O recursion,
where the spectral curve and Bergmann kernel are constructed from
realizations of $\mathcal W_{0,1}$ and $\mathcal W_{0,2}$ respectively
using the framework of emergent geometry.

\end{abstract}

\maketitle

\tableofcontents

\section{Introduction}

\subsection{Backgrounds}

Various types of quadratic recursion relations have
played an important role in Gromov-Witten theory and mirror symmetry,
for example,
the Virasoro constraints \cite{dvv1,dvv2,ehx,get},
the BCOV holomorphic anomaly equations \cite{bcov1,bcov2},
and the Eynard-Orantin topological recursion \cite{eo},
etc.
The Virasoro constraints are a sequence of quadratic recursions
for the free energies on the big phase space,
while the BCOV holomorphic anomaly equations and Eynard-Orantin topological recursions
are quadratic recursions on the small phase space.
Finding relations between different types of quadratic recursion relations
is an interesting and important problem.

Techniques in quantum field theory such like summation over graphs
have now proven to be very powerful in the study of various mathematical and  physical problems.
In literatures,
there are many results about relating  quadratic relations
to various types of Feynman graphs.
For example,
Bershadsky et al. have solved the holomorphic anomaly equation
 in \cite{bcov2} and found that the solutions can be
represented as a summation over stable graphs.
In \cite{eo},
a graph sum formulation of the E-O topological recursion
has been introduced by Eynard and Orantin,
where the graphs are certain trivalent graphs;
and then this graph sum was reformulated as a summation over stable graphs
by Eynard \cite{ey2} and Dunin-Barkowski et al. \cite{doss}.
Inspired by such results,
a general formalism called the `abstract quantum field theory and its realizations'
was introduced by the authors in a previous work \cite{wz1}
to deal with all problems concerning summations over stable graphs.
The purpose of this work is to develop a similar theory for fat graphs and present some of its applications.

Fat graphs,
also known as `ribbon graphs',
are another type of Feynman graphs.
They are actually graphs on surfaces, called maps by some authors.
See eg. \cite{lz} for an introduction.
The study of the combinatorics of such graphs
was initiated by Tutte \cite{tu1, tu2, tu3, tu4}.
In physics literatures,
fat graphs play an important role in the study of gauge theory and matrix models,
beginning with a fascinating connection proposed in the work of 't Hooft \cite{tH}.
For introductions to matrix models and some related works,
see eg. \cite{biz, bipz, dgz, me, iz2}.

Fat graphs have also entered the study of mathematics,
especially modern geometry,
in the last few decades.
Stable fat graphs (i.e., fat graphs whose vertices are all of valences $\geq 3$)
with metrics on edges
lead to the construction of a combinatorial model of the
moduli space $\cM_{g,n}^{\text{comb}}$ ($2g-2+n>0$) of curves \cite{ha, mum, pe2, thu}.
Such a construction of combinatorial moduli spaces inspires powerful techniques
in some geometric problems,
for example,
the computation of orbifold Euler characteristics of the moduli space $\cM_{g,n}$
(i.e., the Harer-Zagier formula),
see Harer-Zagier \cite{hz} and Penner \cite{pe};
and the first proof of the famous Witten Conjecture \cite{wit}
by Kontsevich \cite{kon}.
These authors all introduced new techniques inspired by QFT such as
summations over fat graphs and matrix integrals into mathematics.
See also \cite{biz}.

It is also well-known that fat graphs have entered the study of number theory.
In fact, fat graphs are special cases (the `clean dessins')
of Grothendieck's dessins d'enfants \cite{gr},
see \cite{be, sc} for some arithmetic aspects of the dessins d'enfants,
and see also \cite{lz} for an introduction to these problems.

While computing the partition functions of some Hermitian matrix models,
one may take the Virasoro constraints as a powerful tool.
See \cite{am, da, im, mm} for the derivation of the Virasoro constraints of the Hermitian matrix model,
and \cite{dvv1, fkn} for the Virasoro constraints of the Kontsevich model.
Eynard and Orantin \cite{eo} introduce a formalism of topological recursions inspired by
examples in matrix model theory.
In the present work,
our main goal is to develop the formalism of abstract quantum field theory
for fat graphs and derive  two types of recursion relations in this formalism.
In concrete examples that are given by realizations of such abstract theory,
these recursion relations are realized by
the Virasoro constraints
and the Eynard-Orantin topological recursion respectively.

\subsection{The formalism of abstract quantum field theories and their realizations}

Now let us briefly recall the formalism of
the abstract QFT and its realizations for stable graphs
developed in \cite{wz1}.
This formalism is inspired by the physical literatures on
the BCOV holomorphic anomaly equations \cite{bcov1, bcov2}
and quantum spectral curves \cite{gs}.
The idea is as follows:
\begin{itemize}
\item[1)]
Define `abstract free energy' of genus $g$
to be a linear combination of stable graphs of genus $g$,
such that the coefficient of a stable graph $\Gamma$ is $\frac{1}{|\Aut(\Gamma)|}$
where $\Aut(\Gamma)$ is the group of automorphisms of $\Gamma$.

\item[2)]
We construct the edge-cutting operators and edge-adding operators
on stable graphs,
and derive quadratic recursion relations for the abstract free energies
using these operators.

\item[3)]
A realization of the abstract QFT  is achieved by
specifying suitable  Feynman rules to stable graphs.
This turns the abstract free energies into  some functions
(or formal power series, etc.),
and the quadratic recursion relations are realized by recursions for these functions.
By choosing different Feynman rules
we can obtain different specific theories in literatures and various types of recursions.

\item[4)]
The realizations can be represented as formal Gaussian integrals.

\end{itemize}

In \cite{wz1} and this work we use the word `abstract' to describe
the formalism that we work with. It means that we work at the level of Fenman graphs
without worrying about how they come from a concrete theory.
By a `realization' we mean a specific set of Feynman rules from a concrete quantum field
theory is assigned.
The goal of our formalism is to show that quadratic recursion relations for concrete
theories in the physics literature are often the consequence of some
operations on the Feynman graphs that they involve,
or in our terminology, `at the abstract level'.

It is worthwhile pointing out that
the constructions of the edge-cutting and edge-adding operators on stable graphs
are inspired by the geometric structures of the moduli spaces of curves.
In fact, in mathematical literatures the stable graphs can be used to describe the stratification of
the Deligne-Mumford moduli spaces $\Mbar_{g,n}$ of stable curves \cite{dm,kn},
and the edge-cutting and edge-adding operators are defined as `inverses'
of the gluing maps and forgetful maps on these moduli spaces.
As an application of this formalism,
in \cite{wz2} we
started with the formulas in \cite{hz, bh} and solved the problem of computations
of the orbifold Euler characteristics of $\Mbar_{g,n}$.
This is another typical example that ideas inspired by QFT provide new tools
in the study of algebraic geometry.

Now in the present work,
we we will adopt the idea of \cite{wz1}
and construct the formalsim of the `abstract QFT and its realizations'
for fat graphs (not necessarily stable).
We will construct the abstract correlators,
abstract partition function, abstract free energies, and abstract $n$-point functions
as certain summations of fat graphs,
and derive quadratic recursion relations for them
using the `edge-contraction operator'.
Similar to the case of stable graphs discussed in \cite{wz1},
here the construction of the edge-contraction operator is also
inspired by the geometry of moduli spaces.
In fact,
it comes from the `Whitehead collapse' on fat graphs,
which plays an important role in the construction of the combinatorial moduli space
$\cM_{g,n}^{\text{comb}}$ (see eg. \cite{lz}).
We will also `inverse' the edge-contraction procedures
to construct a family of `vertex-splitting operators',
which gives another equivalent description of the recursions.

We first derive a quadratic recursion relation for the abstract correlators.
Then we reformulate this recursion as quadratic recursions
for two different types of `generating series' of
the abstract correlators--the abstract partition function and abstract $n$-point functions.
For the abstract partition function $\cZ$,
we obtain a family of `abstract Virasoro constraints',
i.e.,
we can construct a family of operators $\{\cL_n\}_{n\geq -1}$
called the abstract Virasoro operators,
such that they annihilate $\cZ$ and form a basis of a Virasoro algebra
(without central extension).
And for the abstract $n$-point functions,
we obtain a quadratic recursion which resembles
the Eynard-Orantin topological recursion.

We will also consider the realizations of the abstract QFT for fat graphs.
A typical example is the Hermitian one-matrix models.
In this example,
one may obtain the Virasoro constraints for the partition function,
and the Eynard-Orantin topological recursion for the $n$-point functions,
as the realizations of the above two abstract recursions.

In summary, we have developed the formalisms of abstract QFTs and their realizations
for stable graphs in \cite{wz1} and for fat graphs in the present work.
As applications we give unified derivations of
the holomorphic anomaly equations (HAE),
the Virasoro constraints, and the Eynard-Orantin topological recursions
in such formalism.
These different quadratic recursions are different realizations of
some abstract recursion relations
derived using different types of operators on Feynman graphs:
for the HAE,
the operators are the edge-cutting/edge-adding operators on stable graphs;
and for the Virasoro constraints and E-O topological recursions,
the operators are the edge-contraction/vertex-splitting operators on fat graphs.
All these operators are geometrically interpreted by
natural geometric structures of moduli spaces of curves.

\subsection{Description of main results}

Now let us describe some results of this work.

First we construct the abstract QFT for fat graphs.
Given $g\in \bZ_{\geq 0}$ and $\mu=(\mu_1,\cdots,\mu_n)\in \bZ_{>0}^n$,
define the abstract correlator of genus $g$ and type $\mu$
to be the following linear combination of fat graphs:
\ben
\cF_g^\mu:=\sum_{\Gamma\in\mathfrak{Fat}_g^{\mu, c}}
\frac{1}{|\Aut(\Gamma)|}\Gamma,
\een
where $\mathfrak{Fat}_g^{\mu, c}$ is the set of all connected fat graphs
whose vertices are labelled by $v_1,\cdots,v_n$ respectively,
and the valence of $v_j$ is $\mu_j$ for every $j$.
We construct an `edge-contraction operator' $K_1$ on fat graphs,
and show that the abstract correlators $\cF_g^\mu$
satisfies the following quadratic recursion
(Theorem \ref{thm-abstract-rec}):
\be
\label{eq-intro-absqrec-correlator}
\begin{split}
&K_1 \cF_g^{\mu}=
\delta_{g,0}\delta_{n,1}\delta_{\mu_1,2}\cF_{0,\{1\}}^{(0)}\cF_{0,\{2\}}^{(0)}+
\sum_{j=2}^n(\mu_1+\mu_j-2)\cF_g^{(\mu_1+\mu_j-2,\mu_{[n]\backslash\{1,j\}})}\\
&\quad
+\sum_{\substack{\alpha+\beta=\mu_1-2\\\alpha\geq 1,\beta\geq 1}}
\alpha\beta\bigg(
\cF_{g-1}^{(\alpha,\beta,\mu_{[n]\backslash\{1\}})}
+\sum_{\substack{g_1+g_2=g\\I\sqcup J=[n]\backslash\{1\}}}
\cF_{g_1,\{1\}\sqcup (I+1)}^{(\alpha,\mu_I)}\cF_{g_2,\{2\}\sqcup (J+1)}^{(\beta,\mu_J)}
\bigg)\\
&\quad
+ (\mu_1-2)\cdot
\cF_{0,\{1\}}^{(0)}\cF_{g,[n+1]\backslash\{1\}}^{(\mu_1-2,\mu_{[n]\backslash\{1\}})}
+ (\mu_1-2)\cdot
\cF_{0,\{2\}}^{(0)}\cF_{g,[n+1]\backslash\{2\}}^{(\mu_1-2,\mu_{[n]\backslash\{1\}})},
\end{split}
\ee
where $[n]:=\{1,2,\cdots,n\}$, and given a finite set of indices
$I=\{i_1,\cdots,i_n\}\subset \bZ_{>0}$ with $i_1<i_2<\cdots<i_n$
we denote by $\cF_{g,I}^\mu$ the summation of graphs obtained from $\cF_{g}^\mu$
by replacing the labels $v_1,\cdots,v_n$ by $v_{i_1},\cdots,v_{i_n}$ respectively,
and $I+1:=\{i_1+1,\cdots,i_n+1\}$
(see \S \ref{sec-absqft} for details).
Here the `multiplication' of two graphs stands for the disjoint union.

The above quadratic recursion for the abstract correlators
can be reformulated in two different ways.
The first way is to reformulate it as a recursion
for the following abstract partition function:
\begin{equation*}
\cZ:=
1+\sum_{g\geq 0}g_s^{2g-2}
\sum_{n\geq 1}\frac{1}{n!}\sum_{\mu\in \bZ_{> 0}^n}
\sum_{\Gamma\in\mathfrak{Fat}_g^\mu}\frac{1}{|\Aut(\Gamma)|}\Gamma,
\end{equation*}
where $g_s$ is a formal variable,
$\mathfrak{Fat}_g^\mu$ is the set of all fat graphs (not necessarily connected)
of genus $g$ and type $\mu\in \bZ_{>0}^n$,
and the term `1' is the empty graph in the sense that
$1\cdot\Gamma=1\sqcup \Gamma:=\Gamma$ for every graph $\Gamma$.
The above quadratic recursion can be reformulated as the following
abstract Virasoro constraints
(Theorem \ref{thm-abstract-vir1}):
\be
\label{eq-intro-abs-Virasoro-1}
\cL_{n} (\cZ) =0,
\qquad
\forall n\geq -1,
\ee
where the abstract Virasoro operators $\{\cL_n\}_{n\geq -1}$ are
of the following form:
\begin{equation*}
\begin{split}
&
\cL_{-1}:=-\pd_1+\sum_{n\geq 1}\cS_{n,n}
+g_s^{-2}\cdot\gamma_{-1},\\
&
\cL_{0}:=-2\pd_2+\sum_{n\geq 1}\cS_{n,n-1}
+g_s^{-2}\cdot\gamma_0,\\
&
\cL_{1}:=-3\pd_3+\sum_{n\geq 1}\cS_{n+1,n-1}
+2\cJ_{1,0}(-\sqcup\Gamma_{dot}),\\
\end{split}
\end{equation*}
and for $m\geq 2$,
\begin{equation*}
\cL_{m}:=-(m+2)\pd_{m+2}+\sum_{n\geq 1}\cS_{n+m,n-1}
+g_s^2\cdot \sum_{n=1}^{m-1}\cJ_{n,m-n}
+2\cJ_{m,0}(-\sqcup\Gamma_{dot}),
\end{equation*}
see \S \ref{sec-fat-abs-virasoro} for the details of this construction.
Roughly speaking,
the operators $\cS_{k,l}$, $\cJ_{k,l}$
are defined to be the `inverses' of the edge-contraction.
Moreover,
one can construct a Lie bracket $[-,-]$ on the linear space spanned by
$\{\cL_n\}_{n\geq -1}$
which is a composition of the commutator and an additional edge-contraction
(see \S \ref{sec-fatabs-virasoro} for details),
such that we have
(Theorem \ref{thm-fat-abs-virasorocomm}):
\be
\label{eq-intro-abs-Virasoro-2}
[\cL_m,\cL_n]=
(m-n)\cL_{m+n},
\qquad
\forall m,n\geq -1.
\ee
As a consequence of the abstract Virasoro constraints,
we obtain a formula of the following form for the abstract partition function $\cZ$:
\be
\label{eq-intro-cut&join}
\cZ = e^{\cM} (1),
\ee
see \S \ref{sec-abs-cur&join} for details.

Another way to reformulate the quadratic recursion \eqref{eq-intro-absqrec-correlator}
is to represent it as a quadratic recursion
for the abstract $n$-point functions.
The abstract $n$-point function $\cW_{g,n}$ of genus $g$ is defined to be
the following formal summation of all connected fat graphs
of genus $g$ with $n$ vertices:
\begin{equation*}
\cW_{g,n}:=
\delta_{g,0}\delta_{n,1}\cF_0^{(0)}
+\sum_{\mu\in \bZ_{>0}^n}\mu_1\mu_2\cdots\mu_n\cF_g^{\mu},
\end{equation*}
where $\cF_0^{(0)}$ is the graph consisting of a single vertex of valence zero.
These abstract $n$-point functions satisfy the following quadratic recursion
(Theorem \ref{thm-abstract-eotype-3}):
\be
\label{eq-intro-abs-EO}
\begin{split}
&(1-2\cT\circ\sigma_1^{-1})\cW_{g,n}
=
\sum_{j=2}^{n}\sigma_j\circ\cS_{\{1;j\}}\circ\sigma_1^{-1}\cW_{g,n-1}\\
&\quad+
\cJ_{\{1,2\}}\circ\sigma_1^{-1}\sigma_2^{-1}\bigg(
\cW_{g-1,n+1}
+\sum_{\substack{g_1+g_2=g\\I\sqcup J=[n+1]\backslash\{1,2\}}}^s
\cW_{g_1,\{1\}\sqcup I}\cdot\cW_{g_2,\{2\}\sqcup J}\bigg),
\end{split}
\ee
where $\sum\limits^{s}$ means that we exclude all the terms with
$(g_1,I)=(0,\emptyset)$ or $(g_2,J)=(0,\emptyset)$.
Taking averaging of this recursion by forgetting the labels on vertices,
one obtains (Theorem \ref{thm-abstract-qsc}):
\be
\label{eq-intro-qsc}
\sigma(\cF_m)
=\cS (\cF_{m-1}) + \cJ (\cF_{m-1}) + \sum_{i+j=m} \cJ(\cF_i,\cF_j),
\qquad
m\geq 1.
\ee
See \S \ref{sec-abstract-eorec} for the detailed constructions of
the operators in these formulas.

It is worthwhile noticing that the recursion \eqref{eq-intro-abs-EO} looks similar
to the Eynard-Orantin topological recursion recursion.
For the models considered in \cite{zhou2, zhou10, zhou11},
the spectral curve and Eynard-Orantin topological recursion both emerge from the Virasoro constraints.
In the case of Witten-Kontsevich tau-function,
it is shown in \cite{Zhou-WK} that the Virasoro constraints are equivalent to the E-O topological recursion
on the Airy curve.
Here again we see a similar phenomenon:
The abstract Virasoro constraints for the abstract partition function and the quadratic recursion \eqref{eq-intro-abs-EO}
for the abstract  $n$-point function for fat graphs are equivalent to each other
because they  both are equivalent to \eqref{eq-intro-absqrec-correlator}.

Next we consider the realizations of the above abstract QFT.
Our main example is the Hermitian one-matrix models,
whose partition function is:
\begin{equation*}
Z_N^{\text{Herm}}:=\frac{
\int_{\cH_N}dM \exp\bigg(tr\sum\limits_{n=1}^\infty\frac{g_n-\delta_{n,2}}{ng_s}M^n\bigg)
}{
\int_{\cH_N}dM \exp\bigg(-\frac{1}{2g_s}tr(M^2)\bigg)
}.
\end{equation*}
There are three types of Feynman rules that we will discuss.
The first Feynman rule is:
\begin{equation*}
\Gamma\quad \mapsto \quad w_\Gamma:=t^{|F(\Gamma)|},
\end{equation*}
for every fat graph $\Gamma$,
where $F(\Gamma)$ is the set of faces of $\Gamma$
and $t:=Ng_s$ is the 't Hooft coupling constant.
Under this Feynman rule,
the abstract correlator $\cF_g^\mu$ is realized by the correlator
$\langle\frac{p_{\mu_1}}{\mu_1}\cdots\frac{p_{\mu_n}}{\mu_n}\rangle_g^c$
of the Hermitian one-matrix models,
and the recursion \eqref{eq-intro-absqrec-correlator} is realized by
the fat Virasoro constraints for the correlators
(see \eqref{eq-1mm-Vir-cor-1} and \eqref{eq-1mm-Vir-cor-2}).

The second Feynman rule is:
\begin{equation*}
\Gamma\quad \mapsto \quad w_\Gamma:=t^{|F(\Gamma)|}\cdot
\prod_{v\in V(\Gamma)} g_{\val(v)},
\end{equation*}
where $V(\Gamma)$ is the set of vertices of $\Gamma$,
and $\val(v)$ is the valence of $v\in V(\Gamma)$.
Under this Feynman rule,
the abstract partition function $\cZ$ is realized by the
partition function $Z_N^{\text{Herm}}$,
and the abstract Virasoro constraints \eqref{eq-intro-abs-Virasoro-1}
are realized by the fat Virasoro constraints of
the Hermitian one-matrix models:
\begin{equation*}
L_{n,t}^{\text{Herm}} (Z_N^{\text{Herm}})=0,\qquad
\forall n\geq -1,
\end{equation*}
where $\{L_{n,t}^{\text{Herm}}\}_{n\geq -1}$ are the fat Virasoro operators
for the Hermitian one-matrix models,
see \eqref{eq-1mm-fatvir-opr-1} and \eqref{eq-1mm-fatvir-opr-2}.
These operators are the realizations of the
abstract Virasoro operators $\{\cL_n\}_{n\geq -1}$.
Moreover,
the relations \eqref{eq-intro-abs-Virasoro-2} are realized by
the Virasoro commutation relations for $\{L_{n,t}^{\text{Herm}}\}_{n\geq -1}$:
\begin{equation*}
[L_{m,t}^{\text{Herm}},L_{m,t}^{\text{Herm}}]=(m-n)L_{m+n,t}^{\text{Herm}},
\qquad \forall m,n\geq -1.
\end{equation*}
Furthermore,
the realization of \eqref{eq-intro-cut&join} gives us:
\be
Z_N^{\text{Herm}}|_{g_s=1} = e^M (1),
\ee
where $M$ is the following element in $\widehat{\mathfrak{gl}(\infty)}$:
\begin{equation*}
M= \half\sum_{i+j+k=-2} :\alpha_i \alpha_j \alpha_k:
+\frac{t}{2} \sum_{i+j=-2} :\alpha_i \alpha_j: + \frac{t^2}{2} \alpha_{-2}.
\end{equation*}
This tells that $Z_N^{\text{Herm}}|_{g_s=1}$ is a tau-function of the KP hierarchy.

The third Feynman rule we consider is:
\begin{equation*}
\Gamma\quad \mapsto\quad
w_\Gamma:=x_1^{-(\val(v_1)+1)}x_2^{-(\val(v_2)+1)}\cdots x_n^{-(\val(v_n)+1)}
\cdot t^{|F(\Gamma)|},
\end{equation*}
where $\Gamma$ is a fat graph with $n$ vertices,
and $x=(x_1,x_2,\cdots)$ is a family of formal variables.
Then the abstract $n$-point functions $\cW_{g,n}$ are realized by the
$n$-point functions of the Hermitian one-matrix models:
\begin{equation*}
W_{g,n}^{\text{Herm}}(x_1,\cdots,x_n):=
\delta_{g,0}\delta_{n,1}\cdot t x_1^{-1}+
\sum_{\mu\in \bZ_{>0}^n}
\langle p_{\mu_1}\cdots p_{\mu_n}
\rangle_g^c
\cdot x_1^{-(\mu_1+1)}\cdots x_n^{-(\mu_n+1)},
\end{equation*}
and the recursion \eqref{eq-intro-abs-EO} is realized by the following
quadratic recursion
(Theorem \ref{thm-EOtype-rec-1mm}):
\begin{equation*}
\begin{split}
&W_{g,n}^{\text{Herm}}(x_1,x_2,\cdots,x_n)\\
=&\sum_{j=2}^n \tD_{x_1,x_j}^{\text{Herm}} W_{g,n-1}^{\text{Herm}}(x_1,\cdots,\hat{x}_j,\cdots,x_n)
+\tE_{x_1,u,v}^{\text{Herm}} W_{g-1,n+1}^{\text{Herm}}(u,v,x_2,\cdots,x_n)\\
&+
\sum_{\substack{g_1+g_2=g\\I\sqcup J=[n]\backslash\{1\}}}^s
\tE_{x_1,u,v}^{\text{Herm}}\bigg(
W_{g_1,|I|+1}^{\text{Herm}}(u,x_I)\cdot W_{g_2,|J|+1}^{\text{Herm}}(v,x_J)\bigg),
\end{split}
\end{equation*}
see \S \ref{sec-realizationrec-npt} for notations.
This recursion has already been derived in a previous work \cite{zhou11} of the second author
using a different method.
In that work it is proved that
this recursion is actually equivalent to Eynard-Orantin topological recursion
on the fat spectral curve
\be
\label{eq-intro-1mm-curve}
x^2-4y^2=4t
\ee
of Hermitian one-matrix models,
where the Bergmann kernel is chosen to be the $2$-point function of genus zero.
Using this result in \cite{zhou11},
we see that the realization of the equation \eqref{eq-intro-qsc},
which gives the following Schr\"odinger type equation:
\begin{equation*}
\bigg( \hbar^2 \frac{d^2}{dx^2} - \frac{1}{4}x^2 +t \bigg)
\exp\bigg( \sum_{m\geq 0} \hbar^{m-1} \widetilde S_m^{\text{Herm}}(x) \bigg) =0,
\end{equation*}
is equivalent to the quantum spectral curve of \eqref{eq-intro-1mm-curve}
(see \S \ref{sec-realization-qsc} for details).
Similar results are also known for the special case $t=1$,
due to \cite{dmss, ms}.

We also consider the emergent geometry of the abstract QFT.
In the formalism of emergent geometry developed by the second author \cite{zhou2},
the spectral curve of a Gromov-Witten type theory is supposed to emerge from the one-point functions
of genus zero.
In this work we show that the following spectral curve
\be
\label{eq-intro-abs-curve}
y=-\frac{z}{\sqrt{2}}
+\sqrt{2} \pd_0(\Gamma_{dot}) z^{-1}
+\sqrt{2}\cdot \sum_{k\geq 1}
\bigg(
\sum_{\Gamma\in \mathfrak{Ho}_k}\frac{2k}{|\Aut(\Gamma)|}\Gamma
\bigg) z^{-2k-1}
\ee
emerges naturally from the abstract Virasoro constraints
\eqref{eq-intro-abs-Virasoro-1}.
This is a curve on the $(y,z)$-plane,
with coefficients some vectors in the linear space spanned by some fat graphs,
see \S \ref{sec-QDT} for notations and details.
We construct a special deformation of the spectral curve \eqref{eq-intro-abs-curve},
and explain how this special deformation and its quantization
encodes the abstract Virasoro constraints
(see Theorem \ref{thm-abstract-vira-genus0} and Theorem \ref{thm-quantumdefo-abs}).

As a realization of these results,
we recover the quantum deformation theory for the fat spectral curve
of the Hermitian one matrix models developed in \cite{zhou10}.
In particular,
the fat spectral curve \eqref{eq-intro-1mm-curve} of the Hermitian one-matrix models
is a realization of the spectral curve \eqref{eq-intro-abs-curve}
under the Feynman rule $w_\Gamma:=t^{|F(\Gamma)|}$.

Now it is natural to expect that in more general cases
the realization of the recursion \eqref{eq-intro-abs-EO}
is also equivalent to the Eynard-Orantin topological recursion,
where the spectral curve is equivalent to the realization of \eqref{eq-intro-abs-curve},
and the Bergmann kernel is the realization of the abstract $2$-point function of genus zero.
We formulate this as a conjecture in \S \ref{sec-conj}.
Moreover,
we conjecture that the realization of \eqref{eq-intro-qsc}
is equivalent to the quantum spectral curve of this spectral curve.
See \S \ref{sec-conj-mainconj} for details.

\subsection{Plan of the paper}

The rest of this paper is arranged as follows.
In \S \ref{sec-absqft},
we introduce the abstract quantum field theory
for fat graphs,
and derive the quadratic recursion relation \eqref{eq-intro-absqrec-correlator}
for abstract correlators using the edge-contraction operator.
In \S \ref{sec-fat-realization},
we discuss the realizations of this abstract QFT by Hermitian matrix models,
and show that for the Hermitian one-matrix models the realization of the
quadratic recursion is equivalent to the fat Virasoro constraints.
Inspired by the results in \S \ref{sec-fat-realization},
we reformulate the quadratic recursion relation \eqref{eq-intro-absqrec-correlator}
in terms of the abstract Virasoro operators $\{\cL_n\}_{n\geq -1}$
in \S \ref{sec-fat-abs-virasoro},
and develop the abstract Virasoro constraints
for the abstract partition function.
In \S \ref{sec-realization-Virasoro} we consider the realization of the
abstract Virasoro constraints.
In \S \ref{sec-QDT},
we discuss the quantum deformation theory of this abstract QFT,
and as an application we show that the quantum deformation theory
for Hermitian one-matrix models is a realization.
In particular,
we can obtain the spectral curves \eqref{eq-intro-1mm-curve} and \eqref{eq-intro-abs-curve}.
We derive the quadratic recursions \eqref{eq-intro-abs-EO} and \eqref{eq-intro-qsc}
for the abstract $n$-point functions in \S \ref{sec-abstract-eorec},
and consider its realization by the Hermitian one-matrix models in \S \ref{sec-1mm-npt}.
Finally in \S \ref{sec-conj},
we present our conjectures towards the E-O topological recursion and quantum spectral curves.

\section{Abstract Quantum Field Theory for Fat Graphs}
\label{sec-absqft}

In this section we formulate the `abstract quantum field theory' for fat graphs.
We define the correlators of this field theory
to be certain linear combinations of fat graphs (not necessarily stable),
called the `abstract correlators',
and derive a quadratic recursion relation for the abstract correlators
using edge-contracting operators.
We also define the abstract free energy and the abstract partition function
for this abstract QFT as certain infinite summations of fat graphs.

\subsection{Graphs on oriented surfaces}
\label{sec-cellgraph}

In this subsection we recall the definition of fat graphs.
We will understand them as `graphs on oriented surfaces'
(see eg. \cite{lz}).
Notice that here we do not impose the stability condition on these graphs
(see Remark \ref{rmk-stable}).

Let $\Sigma_g$ be a connected closed oriented surface of genus $g$.
A graph on the surface $\Sigma_g$ is the $1$-skeleton of a cell decomposition of $\Sigma_g$,
such that the $0$-cells (called `vertices') are labelled by
$\{v_1,v_2,\cdots,v_n\}$,
where $n$ is the number of $0$-cells in this cell decomposition.

Two graphs on $\Sigma_g$ are equivalent if there is
an orientation-preserving homeomorphism of this surface $\Sigma_g$ to itself,
such that one of these graphs is mapped to the other,
and the labels $\{v_1,v_2,\cdots,v_n\}$ on vertices are all preserved.

In what follows,
we will always consider equivalent classes of such graphs,
and call them `fat graphs'.
Fat graphs are drawn as some vertices
(labelled by $v_1,\cdots,v_n$)
and some internal edges connecting them,
together with a cyclic order of the half-edges incident at each vertex,
induced by the the orientation of the surface $\Sigma_g$.
In what follows,
by abuse of notations
we will simply call these vertices $v_1,\cdots,v_n$ respectively.
The genus $g(\Gamma)$ of a fat graph $\Gamma$ is defined to be
the genus of the corresponding surface $\Sigma_g$.

\begin{Example}

The following two graphs are actually the same one of genus zero:
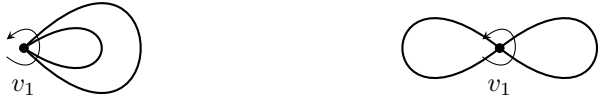
\ben
\begin{tikzpicture}[scale=1.15]
\draw [fill] (0,0) circle [radius=0.05];
\node [align=center,align=center] at (0,-0.45) {$v_1$};
\draw [thick] (0,0) .. controls (1.8,-1.8) and (1.8,1.8) ..  (0,0);
\draw [thick] (0,0) .. controls (1.2,-0.8) and (1.2,0.8) ..  (0,0);
\draw [->,>=stealth] (-0.2,-0.1) .. controls (0.3,-0.5) and (0.3,0.5) ..  (-0.2,0.1);
\draw [fill] (5+0.5,0) circle [radius=0.05];
\node [align=center,align=center] at (5+0.5,-0.45) {$v_1$};
\draw [->,>=stealth] (-0.2+5+0.5,-0.1) .. controls (0.3+5+0.5,-0.5) and (0.3+5+0.5,0.5) ..  (-0.2+5+0.5,0.1);
\draw [thick] (0+5+0.5,0) .. controls (1.5+5+0.5,-1.2) and (1.5+5+0.5,1.2) ..  (0+5+0.5,0);
\draw [thick] (0+5+0.5,0) .. controls (-1.5+5+0.5,-1.2) and (-1.5+5+0.5,1.2) ..  (0+5+0.5,0);
\end{tikzpicture}
\een
This can be seen by placing them on a $2$-sphere.
\end{Example}

Given a fat graph $\Gamma$,
denote by $V(\Gamma)$, $E(\Gamma)$, $H(\Gamma)$
the set of vertices, edges, and half-edges of $\Gamma$ respectively.
An automorphism of the fat graph $\Gamma$ is specified by
a bijection $\phi$ from $H(\Gamma)$ to itself which preserve
all the incidence relations and cyclic orders,
i.e.,
\begin{itemize}
\item[1)]
If $h_1$ and $h_2$ are two half-edges of the same edge,
then so are the half-edges $\phi(h_1)$ and $\phi(h_2)$;

\item[2)]
If $h\in H(\Gamma)$ is incident at the vertex $v_j$,
then $\phi(h)$ is also incident at $v_j$;

\item[3)]
If $(h_{i_1}, h_{i_2},\cdots,h_{i_l})$ are the half-edges
incident at the vertex $v_j$ arranged according to the cyclic order at $v_j$,
then so do $(\phi(h_{i_1}), \phi(h_{i_2}),\cdots,\phi(h_{i_l}))$.

\end{itemize}
Clearly such an automorphism of a fat graph $\Gamma$ is induced by
an orientation-preserving homeomorphism of the corresponding surface $\Sigma_g$ to itself
which preserves each vertex of the graph.
We will denote by $\Aut(\Gamma)$ the automorphism group of $\Gamma$.
Assume the valences
(i.e., number of half-edges incident at a vertex)
of vertices $v_1,\cdots,v_n$ are $\mu_1,\cdots,\mu_n$ respectively,
then it is clear that $\Aut(\Gamma)$ is a subgroup of $\prod_{i=1}^n \bZ/\mu_i\bZ$,
where $\bZ/\mu_i\bZ$ acts by rotating the half edges incident at the vertex $v_i$
(preserving the cyclic order).

\begin{Remark}
\label{rmk-stable}
A fat graph $\Gamma$ is called stable
if the valences of vertices of $\Gamma$ are all greater than two.
In this paper we do not need this stability condition.
\end{Remark}

Given a fat graph $\Gamma$ of genus $g$,
we can recover the topology of the surface $\Sigma_g$ in the following way.
First,
we `fatten' $\Gamma$ to get a surface with boundaries.
There are natural orientations on these boundary components,
induced by the cyclic order of the half-edges on each vertex of $\Gamma$.
Then we attach an oriented $2$-dimensional disc to each boundary component,
and the resulting surface is homeomorphic to $\Sigma_g$.

These $2$-dimensional discs above are called `faces' of $\Gamma$,
since they are $2$-cells in the cell-decomposition of $\Sigma_g$.
Denote by $F(\Gamma)$ the set of all faces of $\Gamma$,
then Euler's formula tells us:
\be\label{eq-euler}
2-2g(\Gamma)=|V(\Gamma)|-|E(\Gamma)|+|F(\Gamma)|.
\ee

\begin{Example}
The following is a fat graph of genus $0$:
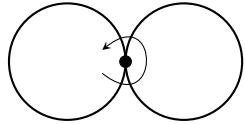
\ben
\begin{tikzpicture}[scale=1.55]
\draw [fill] (0,0) circle [radius=0.05];
\draw [thick](-0.5,0) circle [radius=0.5];
\draw [thick](0.5,0) circle [radius=0.5];
\draw [->,>=stealth] (-0.2,-0.1) .. controls (0.3,-0.5) and (0.3,0.5) ..  (-0.2,0.1);
\end{tikzpicture}
\een
To see this,
we first fatten this graph according to the cyclic order and obtain:
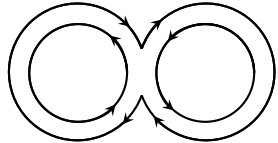
\ben
\begin{tikzpicture}[scale=1.3]
\draw [thick] (0.65,0) arc (0:360:0.5);
\draw [thick] (-0.65,0) arc (0:360:0.5);
\draw [thick] (-0.5,0.25) arc (20:340:0.7);
\draw [thick] (-0.5,0.25) arc (160:-160:0.7);
\draw [->,>=stealth,thick] (-0.65,0) arc (0:50:0.5);
\draw [->,>=stealth,thick] (-0.65,0) arc (0:320:0.5);
\draw [->,>=stealth,thick] (0.65,0) arc (0:140:0.5);
\draw [->,>=stealth,thick] (0.65,0) arc (0:230:0.5);
\draw [->,>=stealth,thick] (-0.5,0.25) arc (160:-140:0.7);
\draw [->,>=stealth,thick] (-0.5,0.25) arc (160:130:0.7);
\draw [-<,>=stealth,thick] (-0.5,0.25) arc (20:50:0.7);
\draw [-<,>=stealth,thick] (-0.5,0.25) arc (20:320:0.7);
\end{tikzpicture}
\een
Now one can see that this graph has one vertex,
two edges, and three faces.
Thus by Euler's formula we have $g=0$.
Or more explicitly,
one may directly glue three $2$-discs along the boundary components,
and in this way we will obtain a $2$-dimensional sphere with a fixed orientation.
It is easy to see that the automorphism of this fat graph
is $\bZ/2\bZ$.

\end{Example}

\begin{Example}
The following is a fat graph of genus $0$:
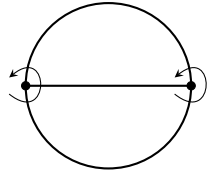
\ben
\begin{tikzpicture}[scale=1.1]
\draw [fill] (-1,0) circle [radius=0.05];
\draw [fill] (1,0) circle [radius=0.05];
\draw [thick](0,0) circle [radius=1];
\draw [thick](-1,0)--(1,0);
\draw [->,>=stealth] (-1.2,-0.1) .. controls (-0.7,-0.5) and (-0.7,0.5) ..  (-1.2,0.1);
\draw [->,>=stealth] (0.8,-0.1) .. controls (1.3,-0.5) and (1.3,0.5) ..  (0.8,0.1);
\end{tikzpicture}
\een
To see this,
we fatten this graph according to the cyclic orders and obtain:
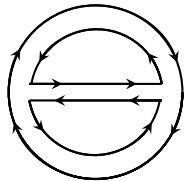
\ben
\begin{tikzpicture}[scale=1.1]
\draw [thick] (0.8,0.1) arc (10:170:0.8);
\draw [->,>=stealth,thick] (0.8,0.1) arc (10:40:0.8);
\draw [->,>=stealth,thick] (0.8,0.1) arc (10:150:0.8);
\draw [thick] (-0.8,-0.1) arc (190:350:0.8);
\draw [->,>=stealth,thick] (-0.8,-0.1) arc (190:220:0.8);
\draw [->,>=stealth,thick] (-0.8,-0.1) arc (190:330:0.8);
\draw [thick](-0.8,0.1)--(0.8,0.1);
\draw [->,>=stealth,thick](-0.8,0.1)--(-0.4,0.1);
\draw [->,>=stealth,thick](-0.8,0.1)--(0.5,0.1);
\draw [thick](-0.8,-0.1)--(0.8,-0.1);
\draw [->,>=stealth,thick](0.8,-0.1)--(-0.5,-0.1);
\draw [->,>=stealth,thick](0.8,-0.1)--(0.4,-0.1);
\draw [thick] (1.05,0) arc (0:360:1.05);
\draw [->,>=stealth,thick] (1.05,0) arc (360:20:1.05);
\draw [->,>=stealth,thick] (1.05,0) arc (360:150:1.05);
\draw [->,>=stealth,thick] (1.05,0) arc (360:200:1.05);
\draw [->,>=stealth,thick] (1.05,0) arc (360:330:1.05);
\end{tikzpicture}
\een
After gluing three $2$-discs to the boundaries,
we obtain a $2$-dimensional sphere with a fixed orientation.
The automorphism of this graph
is $\bZ/3\bZ$.
\end{Example}

\subsection{Abstract QFT for fat graphs and abstract correlators}

In this subsection we formulate an abstract quantum field theory
by defining the correlators of this theory using diagrammatics of fat graphs.
This is inspired by the previous work \cite{wz1} on stable graphs.

Given a sequence of positive integers  $\mu=(\mu_1,\mu_2,\cdots,\mu_n)$,
we say that a fat graph $\Gamma$ is of type $\mu$ if $|V(\Gamma)|=n$ and
the valence of vertex the $v_i$ is $\mu_i$ for $i=1,\cdots,n$.
Denote by $\mathfrak{Fat}_g^{\mu, c}$
the set of all connected fat graphs of genus $g$ and type $\mu$,
then we define the abstract correlators for this abstract QFT
to be the following formal summations of fat graphs:

\begin{Definition}
\label{def-fat-abs-cor}
Let $g\in \bZ_{\geq 0}$ and $\mu\in \bZ_{> 0}^n$.
Define the abstract correlator of genus $g$ and type $\mu=(\mu_1,\mu_2,\cdots,\mu_n)$
to be the following linear combination of fat graphs:
\be
\cF_g^\mu:=\sum_{\Gamma\in\mathfrak{Fat}_g^{\mu, c}}
\frac{1}{|\Aut(\Gamma)|}\Gamma.
\ee
It is an element in the vector space
\be
\cV_g^{\mu,c}:=\bigoplus_{\Gamma\in\mathfrak{Fat}_g^{\mu, c}} \bQ \Gamma.
\ee
We also formally denote:
\begin{equation*}
\begin{tikzpicture}
\draw [fill] (-0.1,0) circle [radius=0.06];
\node [align=center,align=center] at (-1.1,0) {$\cF_0^{(0)}:=$};
\draw [thick] (0,0) .. controls (0,0) and (0,0) ..  (0,0);
\node [align=center,align=center] at (0.3,0) {$v_1$};
\node [align=center,align=center] at (0.75,-0.2) {$.$};
\end{tikzpicture}
\end{equation*}
If $\mathfrak{Fat}_g^{\mu,c}$ is empty for a pair $(g,\mu)$,
then we require $\cF_g^\mu:=0$.

\end{Definition}

\begin{Remark}
It is clear that $\cF_g^\mu=0$ whenever $|\mu|:=\mu_1+\cdots+\mu_n$ is odd,
since we have $|\mu|=2|E(\Gamma)|$ for every $\Gamma\in\mathfrak{Fat}_g^{\mu,c}$.
\end{Remark}

\begin{Example}
\label{eg-abstract-fe}
Let us give some examples of $\cF_g^\mu$:
\begin{flalign*}
\begin{tikzpicture}
\draw [fill] (0,0) circle [radius=0.06];
\node [align=center,align=center] at (0,-0.4) {$v_1$};
\node [align=center,align=center] at (-1.1,0) {$\cF_0^{(2)}=\half$};
\draw [thick] (0,0) .. controls (1.5,-1.2) and (1.5,1.2) ..  (0,0);
\draw [->,>=stealth] (-0.2,-0.1) .. controls (0.3,-0.5) and (0.3,0.5) ..  (-0.2,0.1);
\node [align=center,align=center] at (1.7,-0.25) {$,$};
\end{tikzpicture}&&
\end{flalign*}
\begin{flalign*}
\begin{tikzpicture}[scale=0.95]
\draw [fill] (0,0) circle [radius=0.06];
\node [align=center,align=center] at (0,-0.4) {$v_1$};
\node [align=center,align=center] at (-2.1,0) {$\cF_0^{(4)}=\half$};
\draw [thick] (0,0) .. controls (1.5,-1.2) and (1.5,1.2) ..  (0,0);
\draw [->,>=stealth] (-0.2,-0.1) .. controls (0.3,-0.5) and (0.3,0.5) ..  (-0.2,0.1);
\draw [thick] (0,0) .. controls (-1.5,-1.2) and (-1.5,1.2) ..  (0,0);
\node [align=center,align=center] at (1.7,-0.25) {$,$};
\end{tikzpicture}&&
\end{flalign*}
\begin{flalign*}
\begin{tikzpicture}[scale=0.95]
\draw [fill] (0,0) circle [radius=0.06];
\node [align=center,align=center] at (0,-0.4) {$v_1$};
\node [align=center,align=center] at (-2.1,0) {$\cF_0^{(6)}=\half$};
\draw [thick] (0,0) .. controls (1.8,-1.8) and (1.8,1.8) ..  (0,0);
\draw [thick] (0,0) .. controls (1.2,-0.8) and (1.2,0.8) ..  (0,0);
\draw [->,>=stealth] (-0.2,-0.1) .. controls (0.3,-0.5) and (0.3,0.5) ..  (-0.2,0.1);
\draw [thick] (0,0) .. controls (-1.5,-1.2) and (-1.5,1.2) ..  (0,0);
\node [align=center,align=center] at (2,0) {$+\frac{1}{3}$};
\draw [thick] (3.7,0) .. controls (2.2,-1.5) and (2.2,0.8) ..  (3.7,0);
\draw [thick] (3.7,0) .. controls (5.2,-1.5) and (5.2,0.8) ..  (3.7,0);
\draw [thick] (3.7,0) .. controls (2.5,1.5) and (4.9,1.5) ..  (3.7,0);
\draw [->,>=stealth] (3.5,-0.1) .. controls (4,-0.5) and (4,0.5) ..  (3.5,0.1);
\draw [fill] (3.7,0) circle [radius=0.06];
\node [align=center,align=center] at (3.7,-0.4) {$v_1$};
\node [align=center,align=center] at (5.4,-0.25) {$,$};
\end{tikzpicture}&&
\end{flalign*}
\begin{flalign*}
\begin{tikzpicture}[scale=0.9]
\draw [fill] (0,0) circle [radius=0.06];
\node [align=center,align=center] at (0,-0.4) {$v_1$};
\node [align=center,align=center] at (-2.45,0) {$\cF_0^{(8)}=\half$};
\draw [thick] (0,0) .. controls (1.8,-1.8) and (1.8,1.8) ..  (0,0);
\draw [thick] (0,0) .. controls (1.2,-0.8) and (1.2,0.8) ..  (0,0);
\draw [->,>=stealth] (-0.2,-0.1) .. controls (0.3,-0.5) and (0.3,0.5) ..  (-0.2,0.1);
\draw [thick] (0,0) .. controls (-1.8,-1.8) and (-1.8,1.8) ..  (0,0);
\draw [thick] (0,0) .. controls (-1.2,-0.8) and (-1.2,0.8) ..  (0,0);
\node [align=center,align=center] at (1.9,0) {$+$};
\draw [thick] (3.5,0) .. controls (2,-1.5) and (2,0.8) ..  (3.5,0);
\draw [thick] (3.5,0) .. controls (5,-1.5) and (5,0.8) ..  (3.5,0);
\draw [thick] (3.5,0) .. controls (1.7,1.8) and (5.3,1.8) ..  (3.5,0);
\draw [thick] (3.5,0) .. controls (2.7,1.2) and (4.3,1.2) ..  (3.5,0);
\draw [->,>=stealth] (3.3,-0.1) .. controls (3.8,-0.5) and (3.8,0.5) ..  (3.3,0.1);
\draw [fill] (3.5,0) circle [radius=0.06];
\node [align=center,align=center] at (3.5,-0.4) {$v_1$};
\node [align=center,align=center] at (5.2,0) {$+\frac{1}{4}$};
\draw [thick] (6.9,0) .. controls (5.4,-1.2) and (5.4,1.2) ..  (6.9,0);
\draw [thick] (6.9,0) .. controls (8.4,-1.2) and (8.4,1.2) ..  (6.9,0);
\draw [thick] (6.9,0) .. controls (5.7,1.5) and (8.1,1.5) ..  (6.9,0);
\draw [thick] (6.9,0) .. controls (5.7,-1.5) and (8.1,-1.5) ..  (6.9,0);
\draw [fill] (6.9,0) circle [radius=0.06];
\node [align=center,align=center] at (6.9,-0.4) {$v_1$};
\draw [->,>=stealth] (6.7,-0.1) .. controls (7.2,-0.5) and (7.2,0.5) ..  (6.7,0.1);
\node [align=center,align=center] at (8.6,-0.25) {$,$};
\end{tikzpicture}&&
\end{flalign*}
\begin{flalign*}
\begin{tikzpicture}
\draw [fill] (0,0) circle [radius=0.06];
\node [align=center,align=center] at (0,-0.35) {$v_1$};
\draw [fill] (1.5,0) circle [radius=0.06];
\node [align=center,align=center] at (1.5,-0.35) {$v_2$};
\node [align=center,align=center] at (-1.5,0) {$\cF_0^{(1,1)}=$};
\draw [thick] (0,0) -- (1.5,0);
\node [align=center,align=center] at (2.1,-0.25) {$,$};
\end{tikzpicture}&&
\end{flalign*}
\begin{flalign*}
\begin{tikzpicture}
\draw [fill] (0,0) circle [radius=0.06];
\draw [fill] (1.5,0) circle [radius=0.06];
\node [align=center,align=center] at (-1.5,0) {$\cF_0^{(3,1)}=$};
\draw [thick] (0,0) -- (1.5,0);
\draw [thick] (1.5,0) .. controls (3,1.2) and (3,-1.2) ..  (1.5,0);
\draw [->,>=stealth] (1.3,-0.1) .. controls (1.8,-0.5) and (1.8,0.5) ..  (1.3,0.1);
\node [align=center,align=center] at (0,-0.35) {$v_2$};
\node [align=center,align=center] at (1.5,-0.35) {$v_1$};
\node [align=center,align=center] at (3.2,-0.25) {$,$};
\end{tikzpicture}&&
\end{flalign*}
\begin{flalign*}
\begin{tikzpicture}
\draw [fill] (0,0) circle [radius=0.06];
\node [align=center,align=center] at (0,-0.35) {$v_1$};
\draw [fill] (1.5,0) circle [radius=0.06];
\node [align=center,align=center] at (1.5,-0.35) {$v_2$};
\node [align=center,align=center] at (-1.5,0) {$\cF_0^{(2,2)}=\half$};
\draw [thick] (0,0) .. controls (0.3,0.6) and (1.2,0.6) ..  (1.5,0);
\draw [thick] (0,0) .. controls (0.3,-0.6) and (1.2,-0.6) ..  (1.5,0);
\draw [->,>=stealth] (1.3,-0.1) .. controls (1.8,-0.5) and (1.8,0.5) ..  (1.3,0.1);
\draw [->,>=stealth] (-0.2,-0.1) .. controls (0.3,-0.5) and (0.3,0.5) ..  (-0.2,0.1);
\node [align=center,align=center] at (2.35,-0.25) {$,$};
\end{tikzpicture}&&
\end{flalign*}
\begin{flalign*}
\begin{tikzpicture}[scale=0.9]
\node [align=center,align=center] at (-2.4,0) {$\cF_0^{(3,3)}=\frac{1}{3}$};
\draw [fill] (-1,0) circle [radius=0.06];
\draw [fill] (1,0) circle [radius=0.06];
\node [align=center,align=center] at (-1.15,-0.4) {$v_1$};
\node [align=center,align=center] at (1.15,-0.4) {$v_2$};
\draw [thick](0,0) circle [radius=1];
\draw [thick](-1,0)--(1,0);
\draw [->,>=stealth] (-1.2,-0.1) .. controls (-0.7,-0.5) and (-0.7,0.5) ..  (-1.2,0.1);
\draw [->,>=stealth] (0.8,-0.1) .. controls (1.3,-0.5) and (1.3,0.5) ..  (0.8,0.1);
\node [align=center,align=center] at (1.8,0) {$+$};
\draw [thick](3.8,0)--(5,0);
\draw [thick](3.2,0) circle [radius=0.6];
\draw [thick](5.6,0) circle [radius=0.6];
\draw [fill] (3.8,0) circle [radius=0.06];
\draw [fill] (5,0) circle [radius=0.06];
\node [align=center,align=center] at (3.95,-0.4) {$v_1$};
\node [align=center,align=center] at (4.85,-0.4) {$v_2$};
\draw [->,>=stealth] (-1.2+4.8,-0.1) .. controls (-0.7+4.8,-0.5) and (-0.7+4.8,0.5) ..  (-1.2+4.8,0.1);
\draw [->,>=stealth] (-1.2+6,-0.1) .. controls (-0.7+6,-0.5) and (-0.7+6,0.5) ..  (-1.2+6,0.1);
\node [align=center,align=center] at (6.75,-0.25) {$,$};
\end{tikzpicture}&&
\end{flalign*}
\begin{flalign*}
\begin{tikzpicture}
\node [align=center,align=center] at (-1.5,0) {$\cF_1^{(2)}=0,$};
\end{tikzpicture}&&
\end{flalign*}
\begin{flalign*}
\begin{tikzpicture}[scale=0.9]
\draw [fill] (0,0) circle [radius=0.06];
\node [align=center,align=center] at (-0.1,-0.4) {$v_1$};
\draw [->,>=stealth] (-0.2,-0.1) .. controls (0.3,-0.5) and (0.3,0.5) ..  (-0.2,0.1);
\node [align=center,align=center] at (-2,0) {$\cF_1^{(4)}=\frac{1}{4}$};
\draw [thick] (0,0) arc (180:540:0.8);
\draw [thick] (0,0) arc (90:10:0.8);
\draw [thick] (0,0) arc (90:350:0.8);
\node [align=center,align=center] at (2.25,-0.3) {$,$};
\end{tikzpicture}&&
\end{flalign*}
\begin{flalign*}
\begin{tikzpicture}[scale=0.9]
\node [align=center,align=center] at (-2.65,0) {$\cF_1^{(6)}=\half$};
\draw [fill] (0,0) circle [radius=0.06];
\node [align=center,align=center] at (0,-0.6) {$v_1$};
\draw [->,>=stealth] (-0.2,-0.1) .. controls (0.3,-0.5) and (0.3,0.5) ..  (-0.2,0.1);
\draw [thick] (0,0) arc (0:260:0.8);
\draw [thick] (0,0) arc (360:280:0.8);
\draw [thick] (0,0) arc (180:260:0.8);
\draw [thick] (0,0) arc (540:280:0.8);
\draw [thick] (0,0) arc (90:450:0.8);
\node [align=center,align=center] at (2,0) {$+$};
\draw [fill] (0+3.3,0) circle [radius=0.06];
\node [align=center,align=center] at (3.2,-0.4) {$v_1$};
\draw [->,>=stealth] (-0.2+3.3,-0.1) .. controls (0.3+3.3,-0.5) and (0.3+3.3,0.5) ..  (-0.2+3.3,0.1);
\draw [thick] (0+3.3,0) arc (180:540:0.8);
\draw [thick] (0+3.3,0) arc (90:10:0.8);
\draw [thick] (0+3.3,0) arc (90:350:0.8);
\draw [thick] (0+3.3,0) .. controls (1.5,0.1) and (3.2,1.8) ..  (0+3.3,0);
\node [align=center,align=center] at (5.4,0) {$+\frac{1}{6}$};
\draw [->,>=stealth] (6.1+0.15,-0.1) .. controls (6.6+0.15,-0.5) and (6.6+0.15,0.5) ..  (6.1+0.15,0.1);
\draw [fill] (6.3+0.15,0) circle [radius=0.06];
\node [align=center,align=center] at (6.05,0) {$v_1$};
\draw [thick] (6.3+0.15,0) arc (180:540:0.8);
\draw [thick] (6.3+0.15,0) arc (230:350:0.8);
\draw [thick] (6.3+0.15,0) arc (590:370:0.8);
\draw [thick] (6.3+0.15,0) arc (130:350:0.8);
\draw [thick] (6.3+0.15,0) arc (490:420:0.8);
\draw [thick] (7.42+0.15,-0.1) arc (40:10:0.8);
\node [align=center,align=center] at (8.45+0.15,-0.35) {$.$};
\end{tikzpicture}&&
\end{flalign*}

\end{Example}

\subsection{Abstract free energy and abstract partition function}
\label{sec-abstractpartition}

In the study of field theories,
sometimes it is more convenient to work with the generating series
of correlators.
Now let us define the abstract free energy and abstract partition function
of the abstract QFT for fat graphs.

When considering the abstract free energy and abstract partition function
for fat graphs,
we need to make a slight adjustment on some notations on graphs first.
In this subsection we will abandon the labels $v_1,v_2,\cdots,v_n$ on vertices of fat graphs
(although we still need each vertex to be fixed under automorphisms of graphs).
In other words,
if two graphs differs only by a permutation of the labels on vertices,
then we regard them as the same one
whenever we are dealing with the abstract free energy and abstract partition function.

\begin{Definition}
\label{def-fat-abs-fe}
Define the abstract free energy of genus $g$ to be
the following formal infinite summation of fat graphs:
\be
\cF_g:=
\sum_{n\geq 1}\frac{1}{n!}
\sum_{\mu\in\bZ_{>0}^n}\cF_g^\mu=
\sum_{(m)\in\cP}
\frac{1}{\prod_{j\geq 1} m_j!}\cF_{g}^{\lambda_{(m)}}
\ee
(after forgetting the labels on vertices on the right-hand side),
where $\mathcal P$ is the set of all sequences of non-negative integers
$(m)=(m_1,m_2,\cdots)$ such that
$m_j=0$ for all but a finite number of $j\in\bZ_{>0}$,
and $\lambda_{(m)}=(\lambda_1\geq\lambda_2\geq\cdots\geq\lambda_l)$
is the partition corresponding to $(m)$,
i.e., $m_j$ is the number of $j$ appearing in the sequence
$(\lambda_1,\lambda_2,\cdots,\lambda_l)$.
The infinite sum $\cF_g$ is an element in
\be
\cV_g^c:=
\prod_{\Gamma\in\mathfrak{Fat}_g^{c}}
\bQ\Gamma,
\ee
where
$\mathfrak{Fat}_g^c$
is the set of all connected fat graphs of genus $g$.

Define the abstract free energy $\cF$ to be:
\be
\cF:=\sum_{g\geq 0} g_s^{2g-2} \cF_g,
\ee
where $g_s$ is a formal variable.
\end{Definition}

\begin{Example}
Consider the following fat graph $\Gamma$
(without labels on vertices) of genus zero
(here we omit the cyclic orders on the vertices since
the valences$\leq 2$ and hence there is no confusion):
\ben
\begin{tikzpicture}
\node [align=center,align=center] at (-1,0) {$\Gamma=$};
\draw [fill] (0,0) circle [radius=0.06];
\draw [fill] (1,0) circle [radius=0.06];
\draw [fill] (2,0) circle [radius=0.06];
\draw [thick](0,0)--(2,0);
\node [align=center,align=center] at (2.5,-0.25) {$.$};
\end{tikzpicture}
\een
Let us find the coefficient of $\Gamma$ in $\cF_0$.
Notice that $\Gamma$ can be obtained from three different fat graphs
$\Gamma_1\in\mathfrak{Fat}_0^{(2,1,1),c}$,
$\Gamma_2\in\mathfrak{Fat}_0^{(1,2,1),c}$,
$\Gamma_3\in\mathfrak{Fat}_0^{(1,1,2),c}$
with labels on vertices by forgetting the labels:
\begin{equation*}
\begin{tikzpicture}[scale=0.72]
\node [align=center,align=center] at (-1,0) {$\Gamma_1=$};
\draw [fill] (0,0) circle [radius=0.08];
\draw [fill] (1,0) circle [radius=0.08];
\draw [fill] (2,0) circle [radius=0.08];
\draw [thick](0,0)--(2,0);
\node [align=center,align=center] at (2.5,-0.25) {$,$};
\node [align=center,align=center] at (0,-0.35) {$v_2$};
\node [align=center,align=center] at (1,-0.35) {$v_1$};
\node [align=center,align=center] at (2,-0.35) {$v_3$};
\node [align=center,align=center] at (-1+6,0) {$\Gamma_2=$};
\draw [fill] (0+6,0) circle [radius=0.08];
\draw [fill] (1+6,0) circle [radius=0.08];
\draw [fill] (2+6,0) circle [radius=0.08];
\draw [thick](0+6,0)--(2+6,0);
\node [align=center,align=center] at (2.5+6,-0.25) {$,$};
\node [align=center,align=center] at (0+6,-0.35) {$v_1$};
\node [align=center,align=center] at (1+6,-0.35) {$v_2$};
\node [align=center,align=center] at (2+6,-0.35) {$v_3$};
\node [align=center,align=center] at (-1+6+6,0) {$\Gamma_3=$};
\draw [fill] (0+6+6,0) circle [radius=0.08];
\draw [fill] (1+6+6,0) circle [radius=0.08];
\draw [fill] (2+6+6,0) circle [radius=0.08];
\draw [thick](0+6+6,0)--(2+6+6,0);
\node [align=center,align=center] at (2.5+6+6,-0.25) {$.$};
\node [align=center,align=center] at (0+6+6,-0.35) {$v_1$};
\node [align=center,align=center] at (1+6+6,-0.35) {$v_3$};
\node [align=center,align=center] at (2+6+6,-0.35) {$v_2$};
\end{tikzpicture}
\end{equation*}
Since these three graphs all have trivial automorphism groups,
thus the coefficient of $\Gamma$ in the abstract free energy $\cF_0$ equals:
\ben
\frac{1}{3!}(1+1+1)=\half.
\een

\end{Example}

Now let us define the abstract partition function.
In order to do that,
we need to take disconnected graphs into consideration.
Denote by $\mathfrak{Fat}_g^\mu$ is the set of fat graphs (not necessarily connected)
of genus $g$ and type $\mu$.
Here if $\Gamma=\Gamma_1\sqcup\cdots\sqcup \Gamma_k$ is a disconnected fat graph
where $\Gamma_1,\cdots,\Gamma_k$ are the connected components,
then the automorphism group of $\Gamma$ is defined to be:
\ben
\Aut(\Gamma):=\Aut(\Gamma_1)\times\cdots\times\Aut(\Gamma_k),
\een
and the genus of $\Gamma$ is defined to be:
\ben
g(\Gamma):=g(\Gamma_1)+\cdots+g(\Gamma_k)-k+1.
\een
We define the abstract partition function to be
$\cZ:=\exp(\cF)$ in the following sense:

\begin{Definition}
\label{def-fat-abs-par}
Define the abstract partition function $\cZ$ for fat graphs to be:
\be
\begin{split}
\cZ:=&\sum_{k=0}^\infty \frac{1}{k!}\bigg(
\sum_{g\geq 0} g_s^{2g-2} \cF_g
\bigg)^{k}\\
=&1+\sum_{g\in \bZ}g_s^{2g-2}
\sum_{n\geq 1}\frac{1}{n!}\sum_{\mu\in \bZ_{> 0}^n}
\sum_{\Gamma\in\mathfrak{Fat}_g^\mu}\frac{1}{|\Aut(\Gamma)|}\Gamma,
\end{split}
\ee
where the additional term `1' can be understood as an `empty graph'
(i.e., we formally require $1\sqcup \Gamma:=\Gamma$ for every graph $\Gamma$.
The abstract partition function $\cZ$ is an element in the following space:
\ben
\bQ\cdot 1\times
\prod_{g\in \bZ}\bigg(
g_s^{2g-2}\cdot\prod_{\Gamma\in \mathfrak{Fat}_{g}}
\bQ\Gamma\bigg),
\een
where $\mathfrak{Fat}_g$ is the set of all fat graphs (not necessarily connected)
of genus $g$.

\end{Definition}

\begin{Remark}
Throughout this paper,
by the `product' of two fat graphs $\Gamma_1$ and $\Gamma_2$
we always mean the disconnected graph $\Gamma_1\sqcup \Gamma_2$ represented by
the disjoint union of $\Gamma_1$ and $\Gamma_2$.
\end{Remark}

\subsection{Edge-contracting operator}
\label{sec-fatedge-contr}

In \cite{dmss, wl1, tu4},
those authors dealed with the enumeration problem of fat graphs
using the technique of `contracting an edge'.
Inspired by these works,
in this subsection
let us formulate a linear operator (called the `edge-contracting operator'),
acting on the following linear space:
\be
\cV:=\bigoplus_{g,\mu}\cV_g^{\mu},
\ee
where $\cV_g^{\mu}$ is defined to be
\be
\cV_g^{\mu}:=\bigoplus_{\Gamma\in\mathfrak{Fat}_g^\mu}\bQ\Gamma.
\ee

\begin{Definition}
\label{def-abs-opr}
Define the edge-contracting operator $K_1$ to be the following linear map
on $\cV$:
\be\label{eq-def-opr}
K_1:\quad \cV\to\cV,\quad \Gamma\mapsto\sum_{h\in H(v_1)} \Gamma^h,
\ee
where $H(v_1)$ is the set of all half-edges incident at the vertex $v_1$,
and the graph $\Gamma^h$ is defined as follows:
\begin{itemize}
\item[1)]
If the edge $e$ containing $h$ is not a loop,
then $\Gamma^h$ is obtained from $\Gamma$ by simply contracting $e$
such that the two endpoints of $e$ merges into one new vertex.
The cyclic order at this new vertex is induced by the
cyclic orders at the two endpoints of $e$ in an obvious way,
for example:
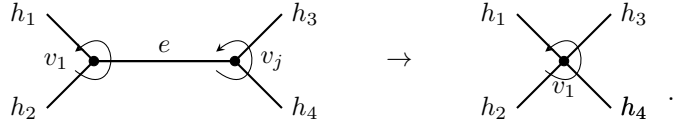
\ben
\begin{tikzpicture}[scale=1.25]
\draw [fill] (0,0) circle [radius=0.05];
\draw [fill] (1.5,0) circle [radius=0.05];
\node [align=center,align=center] at (-0.4,0) {$v_1$};
\node [align=center,align=center] at (1.9,0) {$v_j$};
\draw [thick] (0,0) -- (1.5,0);
\draw [thick] (0,0) -- (-0.5,0.5);
\draw [thick] (0,0) -- (-0.5,-0.5);
\draw [thick] (1.5,0) -- (2,0.5);
\draw [thick] (1.5,0) -- (2,-0.5);
\draw [->,>=stealth] (-0.2,-0.1) .. controls (0.3,-0.5) and (0.3,0.5) ..  (-0.2,0.1);
\node [align=center,left] at (-0.5,0.5) {$h_1$};
\node [align=center,left] at (-0.5,-0.5) {$h_2$};
\node [align=center,right] at (2,0.5) {$h_3$};
\node [align=center,right] at (2,-0.5) {$h_4$};
\node [above,align=center] at (0.75,0) {$e$};
\draw [->,>=stealth] (1.3,-0.1) .. controls (1.8,-0.5) and (1.8,0.5) ..  (1.3,0.1);
\draw [thick] (0,0) -- (1.5,0);
\node [align=center,align=center] at (3.25,0) {$\to$};
\draw [fill] (5,0) circle [radius=0.05];
\node [align=center,align=center] at (5,-0.3) {$v_1$};
\draw [thick] (0+5,0) -- (-0.5+5,0.5);
\draw [thick] (0+5,0) -- (-0.5+5,-0.5);
\draw [thick] (5,0) -- (5.5,0.5);
\draw [thick] (5,0) -- (5.5,-0.5);
\draw [->,>=stealth] (-0.2+5,-0.1) .. controls (0.3+5,-0.5) and (0.3+5,0.5) ..  (-0.2+5,0.1);
\node [align=center,left] at (-0.5+5,0.5) {$h_1$};
\node [align=center,left] at (-0.5+5,-0.5) {$h_2$};
\node [align=center,right] at (5.5,0.5) {$h_3$};
\node [align=center,right] at (5.5,-0.5) {$h_4$};
\node [align=center,right] at (5.5,-0.5) {$h_4$};
\node [align=center,right] at (6,-0.4) {$.$};
\end{tikzpicture}
\een
(Such an operation is called the `Whitehead collapse',
see \cite[\S 4.4]{lz}.)
We label $v_1$ on this new vertex,
and relabel $v_2,v_3,\cdots,v_{n-1}$ on the original vertices
$v_2,,\cdots,\hat{v}_j,\cdots,v_{n}$,
where the notation $\hat{v}_j$ means deleting the term $v_j$.

\item[2)]
If the edge $e$ containing $h$ is a loop at the vertex $v_1$,
then $\Gamma^h$ is obtained from $\Gamma$ by
removing the edge $e$,
splitting $v_1$ into two new vertices,
and partitioning the remaining half-edges incident at $v_1$ into two subsets
according to the cyclic order,
for example:
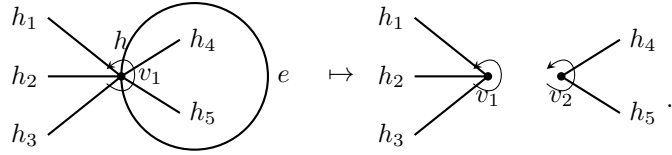
\ben
\begin{tikzpicture}[scale=0.975]
\draw [fill] (0,0) circle [radius=0.05];
\draw [thick] (1,0) circle [radius=1];
\draw [thick] (-1,0) -- (0,0);
\draw [thick] (-1,0.8) -- (0,0);
\draw [thick] (-1,-0.8) -- (0,0);
\draw [thick] (0,0) -- (0.8,0.5);
\draw [thick] (0,0) -- (0.8,-0.5);
\node [align=center,left] at (-1,0) {$h_2$};
\node [align=center,left] at (-1,0.8) {$h_1$};
\node [align=center,left] at (-1,-0.8) {$h_3$};
\node [align=center,right] at (0.8,0.5) {$h_4$};
\node [align=center,right] at (0.8,-0.5) {$h_5$};
\node [align=center,right] at (2,0) {$e$};
\node [align=center,align=center] at (0.4,0) {$v_1$};
\draw [->,>=stealth] (-0.2,-0.1) .. controls (0.3,-0.5) and (0.3,0.5) ..  (-0.2,0.1);
\node [align=center,align=center] at (3,0) {$\mapsto$};
\draw [fill] (5,0) circle [radius=0.05];
\draw [fill] (6,0) circle [radius=0.05];
\node [align=center,align=center] at (5,-0.3) {$v_1$};
\node [align=center,align=center] at (6,-0.3) {$v_2$};
\draw [thick] (-1+5,0) -- (0+5,0);
\draw [thick] (-1+5,0.8) -- (0+5,0);
\draw [thick] (-1+5,-0.8) -- (0+5,0);
\draw [thick] (0+6,0) -- (0.8+6,0.5);
\draw [thick] (0+6,0) -- (0.8+6,-0.5);
\node [align=center,left] at (-1+5,0) {$h_2$};
\node [align=center,left] at (-1+5,0.8) {$h_1$};
\node [align=center,left] at (-1+5,-0.8) {$h_3$};
\node [align=center,right] at (0.8+6,0.5) {$h_4$};
\node [align=center,right] at (0.8+6,-0.5) {$h_5$};
\draw [->,>=stealth] (-0.2+5,-0.1) .. controls (0.3+5,-0.5) and (0.3+5,0.5) ..  (-0.2+5,0.1);
\draw [->,>=stealth] (-0.2+6,-0.1) .. controls (0.3+6,-0.5) and (0.3+6,0.5) ..  (-0.2+6,0.1);
\node [align=center,right] at (0.8+6+0.5,-0.4) {$.$};
\node [align=center,align=center] at (0,0.5) {$h$};
\end{tikzpicture}
\een
In this case,
we label $v_1$ on one of the two new vertices,
such that the half-edges on $v_1$ are those
right after the half-edge $h$
according to the original cyclic order;
and label $v_2$ on the other new vertex.
Finally,
we relabel $v_3,v_4,\cdots,v_{n+1}$ on the
original vertices $v_2,v_3,\cdots,v_n$.

\end{itemize}

\end{Definition}

\begin{Remark}
From the point of view of cellular graphs on surface $\Sigma_g$,
the procedure $2)$ means contracting the $1$-cycle $e$ on $\Sigma_g$
and separating the resulting nodal point.
Notice that in the special case that $e$ bounds a $2$-dimensional disc on $\Sigma_g$,
the new fat graph has two connected components,
and one of them consists of one single vertex,
for example:
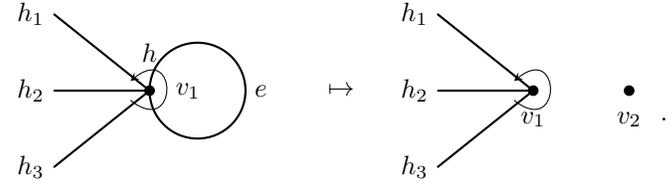
\ben
\begin{tikzpicture}[scale=1.275]
\draw [fill] (0,0) circle [radius=0.05];
\draw [thick] (0.5,0) circle [radius=0.5];
\draw [thick] (-1,0) -- (0,0);
\draw [thick] (-1,0.8) -- (0,0);
\draw [thick] (-1,-0.8) -- (0,0);
\node [align=center,left] at (-1,0) {$h_2$};
\node [align=center,left] at (-1,0.8) {$h_1$};
\node [align=center,left] at (-1,-0.8) {$h_3$};
\node [align=center,right] at (1,0) {$e$};
\node [align=center,align=center] at (0.4,0) {$v_1$};
\draw [->,>=stealth] (-0.2,-0.1) .. controls (0.3,-0.5) and (0.3,0.5) ..  (-0.2,0.1);
\node [align=center,align=center] at (2,0) {$\mapsto$};
\draw [fill] (4,0) circle [radius=0.05];
\draw [fill] (5,0) circle [radius=0.05];
\node [align=center,align=center] at (4,-0.3) {$v_1$};
\node [align=center,align=center] at (5,-0.3) {$v_2$};
\draw [thick] (-1+4,0) -- (0+4,0);
\draw [thick] (-1+4,0.8) -- (0+4,0);
\draw [thick] (-1+4,-0.8) -- (0+4,0);
\node [align=center,left] at (-1+4,0) {$h_2$};
\node [align=center,left] at (-1+4,0.8) {$h_1$};
\node [align=center,left] at (-1+4,-0.8) {$h_3$};
\draw [->,>=stealth] (-0.2+4,-0.1) .. controls (0.3+4,-0.5) and (0.3+4,0.5) ..  (-0.2+4,0.1);
\node [align=center,left] at (5.5,-0.3) {$.$};
\node [align=center,align=center] at (0,0.4) {$h$};
\end{tikzpicture}
\een
\end{Remark}

\begin{Example}
We give some examples of the operator $K_1$:
\begin{flalign*}
\begin{tikzpicture}[scale=1.1]
\draw [fill] (0,0) circle [radius=0.05];
\node [align=center,align=center] at (0,-0.4) {$v_1$};
\node [align=center,align=center] at (-0.8,0) {$K_1\bigg($};
\draw [thick] (0,0) .. controls (1.5,-1.2) and (1.5,1.2) ..  (0,0);
\draw [->,>=stealth] (-0.2,-0.1) .. controls (0.3,-0.5) and (0.3,0.5) ..  (-0.2,0.1);
\node [align=center,align=center] at (1.9,0) {$\bigg)=2\bigg($};
\draw [fill] (2.7,0) circle [radius=0.05];
\draw [fill] (3.2,0) circle [radius=0.05];
\node [align=center,align=center] at (2.7,-0.35) {$v_1$};
\node [align=center,align=center] at (3.2,-0.35) {$v_2$};
\node [align=center,align=center] at (3.6,0) {$\bigg)$};
\node [align=center,align=center] at (4,-0.25) {$,$};
\end{tikzpicture}&&
\end{flalign*}
\begin{flalign*}
\begin{split}
&\begin{tikzpicture}[scale=0.925]
\draw [fill] (0-0.1,0) circle [radius=0.06];
\node [align=center,align=center] at (-0.1,-0.4) {$v_1$};
\node [align=center,align=center] at (-1.9,0) {$K_1\bigg($};
\draw [thick] (0-0.1,0) .. controls (1.8-0.1,-1.8) and (1.8-0.1,1.8) ..  (0-0.1,0);
\draw [thick] (0-0.1,0) .. controls (1.2-0.1,-0.8) and (1.2-0.1,0.8) ..  (0-0.1,0);
\draw [->,>=stealth] (-0.2-0.1,-0.1) .. controls (0.3-0.1,-0.5) and (0.3-0.1,0.5) ..  (-0.2-0.1,0.1);
\draw [thick] (0-0.1,0) .. controls (-1.5-0.1,-1.2) and (-1.5-0.1,1.2) ..  (0-0.1,0);
\node [align=center,align=center] at (2.1,0) {$\bigg)=2\bigg($};
\draw [fill] (2.9+0.1,0) circle [radius=0.06];
\draw [fill] (3.4+0.1,0) circle [radius=0.06];
\node [align=center,align=center] at (3,-0.35) {$v_1$};
\node [align=center,align=center] at (3.5,-0.35) {$v_2$};
\draw [thick] (0+3.4+0.1,0) .. controls (1.8+3.4+0.1,-1.8) and (1.8+3.4+0.1,1.8) ..  (0+3.4+0.1,0);
\draw [thick] (0+3.4+0.1,0) .. controls (1.2+3.4+0.1,-0.8) and (1.2+3.4+0.1,0.8) ..  (0+3.4+0.1,0);
\draw [->,>=stealth] (-0.2+3.4+0.1,-0.1) .. controls (0.3+3.4+0.1,-0.5) and (0.3+3.4+0.1,0.5) ..  (-0.2+3.4+0.1,0.1);
\node [align=center,align=center] at (5.6+0.1,0) {$\bigg)+2\bigg($};
\draw [fill] (2.9+3.7,0) circle [radius=0.06];
\draw [fill] (3.4+0.1+3.6,0) circle [radius=0.06];
\node [align=center,align=center] at (2.9+3.7,-0.35) {$v_2$};
\node [align=center,align=center] at (3.4+0.1+3.6,-0.35) {$v_1$};
\draw [thick] (0+3.4+0.1+3.6,0) .. controls (1.8+3.4+0.1+3.6,-1.8) and (1.8+3.4+0.1+3.6,1.8) ..  (0+3.4+0.1+3.6,0);
\draw [thick] (0+3.4+0.1+3.6,0) .. controls (1.2+3.4+0.1+3.6,-0.8) and (1.2+3.4+0.1+3.6,0.8) ..  (0+3.4+0.1+3.6,0);
\draw [->,>=stealth] (-0.2+3.4+0.1+3.6,-0.1) .. controls (0.3+3.4+0.1+3.6,-0.5) and (0.3+3.4+0.1+3.6,0.5) ..  (-0.2+3.4+0.1+3.6,0.1);
\node [align=center,align=center] at (5.6+0.1+3.3,0) {$\bigg)$};
\end{tikzpicture}\\
&\qquad\qquad\qquad\qquad\qquad\qquad\qquad
\begin{tikzpicture}[scale=0.925]
\node [align=center,align=center] at (5.8,0) {$+2\bigg($};
\draw [fill] (7.5,0) circle [radius=0.06];
\draw [fill] (8.1,0) circle [radius=0.06];
\node [align=center,align=center] at (7.5,-0.4) {$v_1$};
\node [align=center,align=center] at (8.1,-0.4) {$v_2$};
\draw [thick] (7.5,0) .. controls (6,-1.2) and (6,1.2) ..  (7.5,0);
\draw [thick] (8.1,0) .. controls (9.6,-1.2) and (9.6,1.2) ..  (8.1,0);
\node [align=center,align=center] at (9.6,0) {$\bigg)$};
\draw [->,>=stealth] (7.3,-0.1) .. controls (7.8,-0.5) and (7.8,0.5) ..  (7.3,0.1);
\draw [->,>=stealth] (7.9,-0.1) .. controls (8.4,-0.5) and (8.4,0.5) ..  (7.9,0.1);
\node [align=center,align=center] at (10,-0.25) {$,$};
\end{tikzpicture}
\end{split}&&
\end{flalign*}
\begin{flalign*}
\begin{tikzpicture}[scale=1.1]
\draw [fill] (0.5,0) circle [radius=0.05];
\draw [fill] (1.5,0) circle [radius=0.05];
\node [align=center,align=center] at (-0.05,0) {$K_1\bigg($};
\draw [thick] (0.5,0) -- (1.5,0);
\draw [thick] (1.5,0) .. controls (3,1.2) and (3,-1.2) ..  (1.5,0);
\draw [->,>=stealth] (1.3,-0.1) .. controls (1.8,-0.5) and (1.8,0.5) ..  (1.3,0.1);
\node [align=center,align=center] at (0.5,-0.3) {$v_2$};
\node [align=center,align=center] at (1.5,-0.3) {$v_1$};
\node [align=center,align=center] at (3.25,0) {$\bigg)=\bigg($};
\draw [fill] (4.2-0.3,0) circle [radius=0.05];
\draw [fill] (4.7-0.3,0) circle [radius=0.05];
\draw [fill] (5.2-0.3,0) circle [radius=0.05];
\node [align=center,align=center] at (4.2-0.3,-0.3) {$v_3$};
\node [align=center,align=center] at (4.7-0.3,-0.3) {$v_1$};
\node [align=center,align=center] at (5.2-0.3,-0.3) {$v_2$};
\draw [thick] (4.2-0.3,0) -- (4.7-0.3,0);
\node [align=center,align=center] at (5.6,0) {$\bigg)+\bigg($};
\draw [fill] (4.2-0.3+2.3,0) circle [radius=0.05];
\draw [fill] (4.7-0.3+2.3,0) circle [radius=0.05];
\draw [fill] (5.2-0.3+2.3,0) circle [radius=0.05];
\node [align=center,align=center] at (4.2-0.3+2.3,-0.3) {$v_3$};
\node [align=center,align=center] at (4.7-0.3+2.3,-0.3) {$v_2$};
\node [align=center,align=center] at (5.2-0.3+2.3,-0.3) {$v_1$};
\draw [thick] (4.2-0.3+2.3,0) -- (4.7-0.3+2.3,0);
\node [align=center,align=center] at (5.7+4-2,0) {$\bigg)+$};
\draw [fill] (6.3+4-2,0) circle [radius=0.05];
\node [align=center,align=center] at (6.3+4-2,-0.3) {$v_1$};
\draw [thick] (6.3+4-2,0) .. controls (7.8+4-2,-1.2) and (7.8+4-2,1.2) ..  (6.3+4-2,0);
\draw [->,>=stealth] (6.1+4-2,-0.1) .. controls (6.6+4-2,-0.5) and (6.6+4-2,0.5) ..  (6.1+4-2,0.1);
\node [align=center,align=center] at (8+4-2.2,-0.25) {$,$};
\end{tikzpicture}&&
\end{flalign*}
\begin{flalign*}
\begin{tikzpicture}[scale=0.925]
\draw [fill] (0,0) circle [radius=0.06];
\draw [->,>=stealth] (-0.2,-0.1) .. controls (0.3,-0.5) and (0.3,0.5) ..  (-0.2,0.1);
\node [align=center,align=center] at (-1.6,0) {$K_1\biggl($};
\draw [thick] (0,0) arc (180:540:0.8);
\draw [thick] (0,0) arc (90:10:0.8);
\draw [thick] (0,0) arc (90:350:0.8);
\node [align=center,align=center] at (2.6,0) {$\bigg)=4\bigg($};
\draw [fill] (3.45,0) circle [radius=0.06];
\draw [fill] (4.35,0) circle [radius=0.06];
\draw [thick] (3.45,0) -- (4.35,0);
\node [align=center,align=center] at (4.8,0) {$\bigg)$};
\node [align=center,align=center] at (5.2,-0.25) {$.$};
\end{tikzpicture}&&
\end{flalign*}
\end{Example}

\subsection{A quadratic recursion for the abstract correlators}
\label{sec-abs-qrec}

In this subsection,
we derive a quadratic recursion relation for the abstract correlators $\cF_g^\mu$
using the edge-contracting operator $K_1$.

\begin{Example}
First let us check some examples of
the action of the edge-contraction operator on abstract correlators.
Using the expressions presented in Example \ref{eg-abstract-fe},
we can easily make the following observations:
\ben
&&K_1 \cF_0^{(2)}\text{ `=' }(\cF_0^{(0)})^2,\\
&&K_1 \cF_0^{(4)}\text{ `=' }4\cF_0^{(0)}\cF_0^{(2)},\\
&&K_1 \cF_0^{(6)}\text{ `=' }8\cF_0^{(0)}\cF_0^{(4)}+4(\cF_0^{(2)})^2,\\
&&K_1 \cF_0^{(8)}\text{ `=' }
12\cF_0^{(0)}\cF_0^{(6)}+16\cF_0^{(2)}\cF_0^{(4)},\\
&&K_1 \cF_0^{(1,1)}\text{ `=' }\cF_0^{(0)},\\
&&K_1 \cF_0^{(3,1)}\text{ `=' }2\cF_0^{(0)}\cF_0^{(1,1)}+2\cF_0^{(2)},\\
&&K_1 \cF_1^{(4)}\text{ `=' }\cF_0^{(1,1)}.
\een
Here we use `=' because the above equalities do not strictly hold--
they hold only up to a relabelling of the vertices
on the right-hand side.

\end{Example}

The relations in the above example can be modified so that
they hold strictly.
In fact,
in what follows we will describe a method to relabel the vertices on fat graphs,
and in this way we are able to reformulate the right-hand sides
such that the labels match up to the left-hand sides.

Let $I=\{i_1,i_2,\cdots,i_n\}\subset \bZ_{>0}$ be
a finite set of indices with $i_1<i_2<\cdots<i_n$,
and $\mu=(\mu_1,\mu_2\cdots,\mu_n)\in \bZ_{>0}$.
Denote by $\mathfrak{Fat}_{g,I}^{\mu,c}$ the set of connected fat graphs
obtained from graphs in $\mathfrak{Fat}_g^{\mu,c}$ by
replacing the labels $v_1,v_2,\cdots,v_n$
by new labels $v_{i_1},v_{i_2},\cdots,v_{i_n}$.
Now we denote by $\cF_{g,I}^{\mu}$ the abstract correlators
with new labels $v_{i_1},\cdots,v_{i_n}$:
\be
\cF_{g,I}^{\mu}:=
\sum_{\Gamma\in\mathfrak{Fat}_{g,I}^{\mu,c}}
\frac{1}{|\Aut(\Gamma)|}\Gamma.
\ee
Also,
we formally denote:
\begin{equation*}
\begin{tikzpicture}
\draw [fill] (0.15,0) circle [radius=0.06];
\node [align=center,align=center] at (-1.1,0) {$\cF_{0,\{k\}}^{(0)}:=$};
\draw [thick] (0,0) .. controls (0,0) and (0,0) ..  (0,0);
\node [align=center,align=center] at (0.15,-0.3) {$v_k$};
\node [align=center,align=center] at (0.55,-0.25) {$.$};
\end{tikzpicture}
\end{equation*}

Given a positive integer $n$,
denote $[n]:=\{1,2,\cdots,n\}$.
Then our main theorem in this subsection is:

\begin{Theorem}
\label{thm-abstract-rec}
The following quadratic recursion relation holds:
\be
\label{eq-abstract-rec}
\begin{split}
&K_1 \cF_g^{\mu}=
\delta_{g,0}\delta_{n,1}\delta_{\mu_1,2}\cF_{0,\{1\}}^{(0)}\cF_{0,\{2\}}^{(0)}+
\sum_{j=2}^n(\mu_1+\mu_j-2)\cF_g^{(\mu_1+\mu_j-2,\mu_{[n]\backslash\{1,j\}})}\\
&\quad
+\sum_{\substack{\alpha+\beta=\mu_1-2\\\alpha\geq 1,\beta\geq 1}}
\alpha\beta\bigg(
\cF_{g-1}^{(\alpha,\beta,\mu_{[n]\backslash\{1\}})}
+\sum_{\substack{g_1+g_2=g\\I\sqcup J=[n]\backslash\{1\}}}
\cF_{g_1,\{1\}\sqcup (I+1)}^{(\alpha,\mu_I)}\cF_{g_2,\{2\}\sqcup (J+1)}^{(\beta,\mu_J)}
\bigg)\\
&\quad
+ (\mu_1-2)\cdot
\cF_{0,\{1\}}^{(0)}\cF_{g,[n+1]\backslash\{1\}}^{(\mu_1-2,\mu_{[n]\backslash\{1\}})}
+ (\mu_1-2)\cdot
\cF_{0,\{2\}}^{(0)}\cF_{g,[n+1]\backslash\{2\}}^{(\mu_1-2,\mu_{[n]\backslash\{1\}})},
\end{split}
\ee
where we use the convention $\cF_{g}^{(\mu_1-2,\mu_{[n]\backslash\{1\}})}:=0$
for $\mu_1< 2$;
and for a set of indices $I=\{i_1,\cdots,i_k\}$ with $i_1<i_2<\cdots<i_n$,
we denote:
\ben
I+1:=\{i_1+1,\cdots,i_k+1\}.
\een

\end{Theorem}

\begin{proof}

The proof is similar to the proof of the counting problem \cite[Theorem 3.3]{dmss},
see also \cite{tu2, wl1}.
The authors of those literatures have considered the enumeration problem
of fat graphs,
and obtained a similar recursion formula for
the weighted number of graphs of a given type.
The case we are dealing with here can be proved using the same method.
Since the labels on vertices can be automatically fixed by
analysing the valences,
thus in what follows we will only consider the edge-contraction procedure
without relabelling of vertices,
i.e.,
we will construct the correspondence of the graphs (without labels)
appearing in two sides of the above recursion
and compare their coefficients,
then the conclusion follows easily by analyzing the relabelling procedure.

Define $f_g^\mu$ to be the following linear combination:
\ben
f_g^\mu:=\sum_{\vec{\Gamma}\in \vec{\Gamma}_g^{\mu,c}}\vec{\Gamma}
\een
for every $\mu\not=(0)$,
where $\vec{\Gamma}_g^{\mu,c}$ is the set of
fat graphs of genus $g$ and type $\mu=(\mu_1,\cdots,\mu_n)$,
with an arrow placed on one of the half edge
incident at $v_i$ for every vertex $v_i$.
For example,
\begin{equation*}
\begin{split}
&\begin{tikzpicture}[scale=0.9]
\draw [fill] (0,0) circle [radius=0.06];
\node [align=center,align=center] at (-2,0) {$f_0^{(6)}=$};
\draw [thick] (0,0) .. controls (1.8,-1.8) and (1.8,1.8) ..  (0,0);
\draw [thick] (0,0) .. controls (1.2,-0.8) and (1.2,0.8) ..  (0,0);
\draw [->,>=stealth] (-0.2,-0.1) .. controls (0.3,-0.5) and (0.3,0.5) ..  (-0.2,0.1);
\draw [thick] (0,0) .. controls (-1.5,-1.2) and (-1.5,1.2) ..  (0,0);
\draw [thick,->,>=stealth] (0.3,0.27) -- (0.4,0.35);
\node [align=center,align=center] at (2,0) {$+$};
\draw [fill] (0+3.8,0) circle [radius=0.06];
\draw [thick] (0+3.8,0) .. controls (1.8+3.8,-1.8) and (1.8+3.8,1.8) ..  (0+3.8,0);
\draw [thick] (0+3.8,0) .. controls (1.2+3.8,-0.8) and (1.2+3.8,0.8) ..  (0+3.8,0);
\draw [->,>=stealth] (-0.2+3.8,-0.1) .. controls (0.3+3.8,-0.5) and (0.3+3.8,0.5) ..  (-0.2+3.8,0.1);
\draw [thick] (0+3.8,0) .. controls (-1.5+3.8,-1.2) and (-1.5+3.8,1.2) ..  (0+3.8,0);
\draw [thick,->,>=stealth] (4.1,0.16) -- (4.2,0.2);
\node [align=center,align=center] at (2+3.8,0) {$+$};
\draw [fill] (0+3.8+3.8,0) circle [radius=0.06];
\draw [thick] (0+3.8+3.8,0) .. controls (1.8+3.8+3.8,-1.8) and (1.8+3.8+3.8,1.8) ..  (0+3.8+3.8,0);
\draw [thick] (0+3.8+3.8,0) .. controls (1.2+3.8+3.8,-0.8) and (1.2+3.8+3.8,0.8) ..  (0+3.8+3.8,0);
\draw [->,>=stealth] (-0.2+3.8+3.8,-0.1) .. controls (0.3+3.8+3.8,-0.5) and (0.3+3.8+3.8,0.5) ..  (-0.2+3.8+3.8,0.1);
\draw [thick] (0+3.8+3.8,0) .. controls (-1.5+3.8+3.8,-1.2) and (-1.5+3.8+3.8,1.2) ..  (0+3.8+3.8,0);
\draw [thick,->,>=stealth] (7.3,0.2) -- (7.2,0.25);
\end{tikzpicture}\\
&\qquad\qquad\qquad\begin{tikzpicture}[scale=0.9]
\node [align=center,align=center] at (2,0) {$+$};
\draw [thick] (3.7,0) .. controls (2.2,-1.5) and (2.2,0.8) ..  (3.7,0);
\draw [thick] (3.7,0) .. controls (5.2,-1.5) and (5.2,0.8) ..  (3.7,0);
\draw [thick] (3.7,0) .. controls (2.5,1.5) and (4.9,1.5) ..  (3.7,0);
\draw [->,>=stealth] (3.5,-0.1) .. controls (4,-0.5) and (4,0.5) ..  (3.5,0.1);
\draw [fill] (3.7,0) circle [radius=0.06];
\draw [thick,->,>=stealth] (3.9,0.3) -- (3.95,0.4);
\node [align=center,align=center] at (2+3.5,0) {$+$};
\draw [thick] (3.7+3.5,0) .. controls (2.2+3.5,-1.5) and (2.2+3.5,0.8) ..  (3.7+3.5,0);
\draw [thick] (3.7+3.5,0) .. controls (5.2+3.5,-1.5) and (5.2+3.5,0.8) ..  (3.7+3.5,0);
\draw [thick] (3.7+3.5,0) .. controls (2.5+3.5,1.5) and (4.9+3.5,1.5) ..  (3.7+3.5,0);
\draw [->,>=stealth] (3.5+3.5,-0.1) .. controls (4+3.5,-0.5) and (4+3.5,0.5) ..  (3.5+3.5,0.1);
\draw [fill] (3.7+3.5,0) circle [radius=0.06];
\draw [thick,->,>=stealth] (7,0.3) -- (6.95,0.4);
\node [align=center,align=center] at (8.75,-0.25) {$.$};
\end{tikzpicture}
\end{split}
\end{equation*}
Let $\pi$ be the operation of `forgetting the arrows on half-edges',
then for every $\mu\not=(0)$ we have:
\be\label{eq-integcoeff}
\pi(f_g^\mu)=\mu_1\cdots\mu_n\cdot\cF_g^\mu.
\ee
In the special case $\mu=(0)$,
we define $\pi f_g^{(0)}=f_g^{(0)}:=\delta_{g,0}\cdot \cF_0^{(0)}$.

Now consider the process of shrinking the edge $e$ which contains
the arrowed half-edge incident at $v_1$.
There are two cases:
\begin{itemize}
\item[1)]
If this edge $e$ connects $v_1$ and $v_j$ with $j\not=1$,
then we obtain a new vertex of valence $\mu_1+\mu_j-2$.
Next we place the arrow of this new vertex on the half-edge next to $e$
around $p_1$ with respect to the counterclockwise cyclic order,
and delete the original arrow on $v_j$.
\item[2)]
If this edge $e$ is a loop incident at $v_1$,
then we separate the vertex $v_1$ into two new vertices
with total valence $\mu_1-2$.
The two new arrows on these two new vertices are placed on
the two half-edges next to the loop $e$ around $p_1$
with respect to the counterclockwise cyclic order.
\end{itemize}

Similar to the case of \cite[Theorem 3.3]{dmss},
the above procedure gives us a one-to-one correspondence from
$\vec{\Gamma}_g^{(\mu_1,\cdots,\mu_n),c}$ to the set:
\ben
&&\bigg(
\bigcup_{j=2}^n\vec{\Gamma}_g^{(\mu_1+\mu_j-2,\mu_{[n]\backslash\{1,j\}}),c}
\bigg)\cup\\
&&\bigg[
\bigcup_{\alpha+\beta=\mu_1-2}
\vec{\Gamma}_{g-1}^{(\alpha,\beta,\mu_{[n]\backslash\{1\}}),c}
\cup\bigg(\bigcup_{\substack{g_1+g_2=g\\I\sqcup J=[n]\backslash\{1\}}}
\vec{\Gamma}^{(\alpha,\mu_I),c}\oplus\vec{\Gamma}^{(\beta,\mu_J),c}
\bigg)\bigg].
\een
Now we understand the above procedure of `shrinking the arrowed edge'
using the operator $K_1$ and compare the coefficients,
we may obtain:
\be\label{eq-proof-rec}
\begin{split}
K_1\pi f_g^\mu=&\sum_{j=2}^n\mu_j \pi
f_g^{(\mu_1+\mu_j-2,\mu_{[n]\backslash\{1,j\}})}\\
&+\sum_{\alpha+\beta=\mu_1-2}\bigg(
\pi f_{g-1}^{(\alpha,\beta,\mu_{[n]\backslash\{1\}})}
+\sum_{\substack{g_1+g_2=g\\I\sqcup J=[n]\backslash\{1\}}}
\pi f_{g_1}^{(\alpha,\mu_I)}\pi f_{g_2}^{(\beta,\mu_J)}
\bigg).
\end{split}
\ee
Recall that we have \eqref{eq-integcoeff} for $\mu\not=(0)$,
and $f_0^{(0)}=\cF_0^{(0)}$  in the unstable case $\mu=(0)$,
thus \eqref{eq-proof-rec} implies the recursion relation \eqref{eq-abstract-rec}
once we perform the relabelling of the vertices.

\end{proof}

\begin{Example}

Let us present an example
to show how the quadratic recursion relation
\eqref{eq-abstract-rec} works.
The explicit expression for $\cF_0^{(4,1,1)}$ is:
\ben
\begin{tikzpicture}[scale=1.15]
\node [align=center,align=center] at (-2,0) {$\cF_0^{(4,1,1)}=$};
\draw [fill] (0,0) circle [radius=0.05];
\draw [thick] (0,0) .. controls (1.5,-1.2) and (1.5,1.2) ..  (0,0);
\draw [->,>=stealth] (-0.2,-0.1) .. controls (0.3,-0.5) and (0.3,0.5) ..  (-0.2,0.1);
\draw [thick] (-0.6,0.3) -- (0,0);
\draw [thick] (-0.6,-0.3) -- (0,0);
\draw [fill] (-0.6,0.3) circle [radius=0.05];
\draw [fill] (-0.6,-0.3) circle [radius=0.05];
\node [align=center,left] at (-0.6,0.3) {$v_2$};
\node [align=center,left] at (-0.6,-0.3) {$v_3$};
\node [align=center,align=center] at (0,-0.3) {$v_1$};
\node [align=center,align=center] at (1.6,0) {$+$};
\draw [fill] (0+3,0) circle [radius=0.05];
\draw [thick] (0+3,0) .. controls (1.5+3,-1.2) and (1.5+3,1.2) ..  (0+3,0);
\draw [->,>=stealth] (-0.2+3,-0.1) .. controls (0.3+3,-0.5) and (0.3+3,0.5) ..  (-0.2+3,0.1);
\draw [thick] (-0.6+3,0.3) -- (0+3,0);
\draw [thick] (-0.6+3,-0.3) -- (0+3,0);
\draw [fill] (-0.6+3,0.3) circle [radius=0.05];
\draw [fill] (-0.6+3,-0.3) circle [radius=0.05];
\node [align=center,left] at (-0.6+3,0.3) {$v_3$};
\node [align=center,left] at (-0.6+3,-0.3) {$v_2$};
\node [align=center,align=center] at (0+3,-0.3) {$v_1$};
\node [align=center,align=center] at (4.6,0) {$+$};
\draw [fill] (0+6,0) circle [radius=0.05];
\draw [thick] (0+6,0) .. controls (1.8+6,-1.5) and (1.8+6,1.5) ..  (0+6,0);
\draw [->,>=stealth] (-0.2+6,-0.1) .. controls (0.3+6,-0.5) and (0.3+6,0.5) ..  (-0.2+6,0.1);
\draw [thick] (-0.6+6,0) -- (0+6,0);
\draw [thick] (0.6+6,0) -- (0+6,0);
\draw [fill] (-0.6+6,0) circle [radius=0.05];
\draw [fill] (0.6+6,0) circle [radius=0.05];
\node [align=center,left] at (-0.6+6,0) {$v_2$};
\node [align=center,right] at (0.6+6,0) {$v_3$};
\node [align=center,align=center] at (0+6,-0.3) {$v_1$};
\node [align=center,align=center] at (7.75,-0.25) {$.$};
\end{tikzpicture}
\een
Then direct computation gives us:
\begin{equation*}
\begin{split}
&\begin{tikzpicture}[scale=1.075]
\node [align=center,align=center] at (-1.9,0) {$K_1\cF_0^{(4,1,1)}=6$};
\draw [fill] (0,0) circle [radius=0.05];
\node [align=center,align=center] at (0,-0.3) {$v_1$};
\node [align=center,align=center] at (-0.6,-0.3) {$v_2$};
\draw [fill] (-0.6,0) circle [radius=0.05];
\draw [thick] (0,0) .. controls (1.5,-1.2) and (1.5,1.2) ..  (0,0);
\draw [->,>=stealth] (-0.2,-0.1) .. controls (0.3,-0.5) and (0.3,0.5) ..  (-0.2,0.1);
\draw [thick] (-0.6,0) -- (0,0);
\node [align=center,align=center] at (1.8,0) {$+2\bigg($};
\draw [fill] (2.4,0) circle [radius=0.05];
\draw [fill] (3,0) circle [radius=0.05];
\draw [fill] (3.6,0) circle [radius=0.05];
\draw [fill] (4.2,0) circle [radius=0.05];
\draw [thick] (2.4,0) -- (3.6,0);
\node [below,align=center] at (3,0) {$v_1$};
\node [below,align=center] at (2.4,0) {$v_3$};
\node [below,align=center] at (3.6,0) {$v_4$};
\node [below,align=center] at (4.2,0) {$v_2$};
\node [align=center,align=center] at (5,0) {$\bigg)+2\bigg($};
\draw [fill] (5.8,0) circle [radius=0.05];
\draw [fill] (6.4,0) circle [radius=0.05];
\draw [fill] (7,0) circle [radius=0.05];
\draw [fill] (7.6,0) circle [radius=0.05];
\node [below,align=center] at (7.6,0) {$v_1$};
\node [below,align=center] at (5.8,0) {$v_3$};
\node [below,align=center] at (7,0) {$v_4$};
\node [below,align=center] at (6.4,0) {$v_2$};
\node [align=center,align=center] at (8,0) {$\bigg)$};
\draw [thick] (5.8,0) -- (7,0);
\end{tikzpicture}\\
&\qquad\qquad\qquad\begin{tikzpicture}[scale=1.075]
\node [align=center,align=center] at (1.8,0) {$+\bigg($};
\draw [fill] (2.4,0) circle [radius=0.05];
\draw [fill] (3,0) circle [radius=0.05];
\draw [fill] (3.6,0) circle [radius=0.05];
\draw [fill] (4.2,0) circle [radius=0.05];
\draw [thick] (2.4,0) -- (3,0);
\draw [thick] (3.6,0) -- (4.2,0);
\node [below,align=center] at (3,0) {$v_1$};
\node [below,align=center] at (2.4,0) {$v_3$};
\node [below,align=center] at (3.6,0) {$v_2$};
\node [below,align=center] at (4.2,0) {$v_4$};
\node [align=center,align=center] at (5,0) {$\bigg)+\bigg($};
\draw [fill] (5.8,0) circle [radius=0.05];
\draw [fill] (6.4,0) circle [radius=0.05];
\draw [fill] (7,0) circle [radius=0.05];
\draw [fill] (7.6,0) circle [radius=0.05];
\node [align=center,align=center] at (8.05,0) {$\bigg)$};
\draw [thick] (5.8,0) -- (6.4,0);
\draw [thick] (7,0) -- (7.6,0);
\node [below,align=center] at (6.4,0) {$v_1$};
\node [below,align=center] at (7.6,0) {$v_3$};
\node [below,align=center] at (7,0) {$v_2$};
\node [below,align=center] at (5.8,0) {$v_4$};
\node [align=center,align=center] at (8.35,-0.25) {$.$};
\end{tikzpicture}
\end{split}
\end{equation*}
Or equivalently,
\begin{equation*}
\begin{split}
K_1\cF_0^{(4,1,1)}=&6\cF_0^{(3,1)}
+2\cF_{0,\{2\}}^{(0)}\cF_{0,\{1,3,4\}}^{(2,1,1)}
+2\cF_{0,\{1\}}^{(0)}\cF_{0,\{2,3,4\}}^{(2,1,1)}
\\
&+
\cF_{0,\{1,3\}}^{(1,1)}\cF_{0,\{2,4\}}^{(1,1)}
+\cF_{0,\{1,4\}}^{(1,1)}\cF_{0,\{2,3\}}^{(1,1)}.
\end{split}
\end{equation*}
\end{Example}

\section{Examples of Realizations of the Abstract QFT}
\label{sec-fat-realization}

In this section we construct realizations of the abstract QFT
by assigning Feynman rules to fat graphs.
We will see that matrix models will
provide some natural examples of the realizations.
In particular,
when we consider the realization of the abstract QFT by the Hermitian one-matrix models,
the realization of the quadratic recursion for the abstract correlators
is equivalent to the fat Virasoro constraints for the Hermitian one-matrix models.

\subsection{Feynman rules and realization of the edge-contracting operator}

In this subsection we consider the realizations of the abstract quantum field theory
for fat graphs and realization of the recursion relations.

A `Feynman rule' is an assignment of a `weight' $w_\Gamma$
(which can be a formal variable, a function, or a formal power series, etc.)
to each fat graph $\Gamma$.
In this way,
we associate weights
\be
\begin{split}
& F_g^\mu:=\sum_{\Gamma\in\mathfrak{Fat}_g^{\mu,c}}
\frac{1}{|\Aut(\Gamma)|}w_\Gamma,\\
& F_g:=
\sum_{n\geq 1}\frac{1}{n!}
\sum_{\mu\in\bZ_{>0}^n}F_g^\mu,\\
& Z:=\exp\bigg(
\sum_{g\geq 0}g_s^{2g-2} F_g
\bigg),
\end{split}
\ee
to the abstract correlators $\cF_g^\mu$,
abstract free energy $\cF$,
and abstract partition function $ \cZ $ respectively
(see Definition \ref{def-fat-abs-cor}, \ref{def-fat-abs-fe}, and \ref{def-fat-abs-par}).
We call $ F_g^\mu$, $ F_g$, $Z$ the
`realizations' of $\cF_g^\mu$, $\cF_g$, $ \cZ $ respectively.

For example,
a large class of Feynman rule can be taken in the following form:
one assigns weights $w_v$, $w_e$, $w_f$ to
each vertex $v$, edge $e$, face $f$ in a fat graph $\Gamma$ respectively:
\ben
&&v\mapsto w_v , \qquad\forall v\in V(\Gamma);\\
&&e\mapsto w_e , \qquad\forall e\in E(\Gamma);\\
&&f\mapsto w_f , \qquad\forall f\in F(\Gamma),\\
\een
where $w_v$, $w_e$, $w_f$ are some formal variables
(functions, formal power series, etc.)
Then for every fat graph $\Gamma$,
the weight of $\Gamma$ is defined to be:
\be
w_\Gamma:=\prod_{v\in V(\Gamma)}w_v\cdot
\prod_{e\in E(\Gamma)}w_e \cdot \prod_{f\in F(\Gamma)}w_f.
\ee

\begin{Remark}
If we assign $w_v$, $w_e$, $w_f$ to be of degree one,
then Euler's formula \eqref{eq-euler} tells us that
$w_\Gamma$ is of degree $2-2g$ for every fat graph $\Gamma$ of genus $g$.
\end{Remark}

Given a realization $\{w_\Gamma\}$ of the abstract QFT,
let us consider the realization of the edge-contracting operator $K_1$.

\begin{Definition}
By a `realization' of the edge-contracting operator $K_1$,
we mean an operator $\tK_1$ satisfying:
\be
w_{K_1 (\Gamma)}=\tK_1(w_\Gamma)
\ee
for every fat graph $\Gamma$.
\end{Definition}

Using such an operator $\tK_1$,
the recursion relation in Theorem \ref{thm-abstract-rec} can be realized
as a quadratic recursion for $ F_g^\mu$.
In the rest of this subsection,
we will consider some special examples.
In these examples we are able to write down the realization of $K_1$ explicitly,
thus the quadratic recursion relations for the realized fat correlators $ F_g^\mu$
can be derived easily.

\begin{Example}[Enumeration of fat graphs]
\label{eg-enum}
Consider the simplest Feynman rule:
\be\label{eq-FR-eg1}
\begin{split}
& w_v=1 , \quad\forall v\in V(\Gamma);\\
& w_e=1 , \quad\forall e\in E(\Gamma);\\
& w_f=1 , \quad\forall f\in F(\Gamma),\\
\end{split}
\ee
then clearly we have $w_\Gamma=1$ for every fat graph $\Gamma$,
and the problem becomes the enumeration problem of
connected fat graphs $\Gamma$ of a fixed genus and a fixed type,
with coefficient $\frac{1}{|\Aut(\Gamma)|}$:
\be\label{eq-fe-enum}
 F_g^\mu=\sum_{\Gamma\in\mathfrak{Fat}_g^{\mu,c}}
\frac{1}{|\Aut(\Gamma)|},
\ee
see \cite{wl1,wl2,wl3} and \cite{dmss} for this enumeration problem.

Since for every fat graph $\Gamma\in\mathfrak{Fat}_g^{\mu,c}$,
the valence of the vertex $v_1$ is $\mu_1\in \bZ$,
thus by the definition \eqref{eq-def-opr} we know that
the operator $K_1$ in this case can be realized by multiplying by $\mu_1$
in this case,
since:
\ben
w_{K_1(\Gamma)}=\sum_{h\in H(v_1)}w_{\Gamma^h}
=\sum_{h\in H(v_1)}1 =\mu_1
\een
for every $\Gamma\in \mathfrak{Fat}_g^{\mu,c}$.
Notice that in this case the weight of a vertex is independent of the labels,
thus $\cF_{g,I}^\mu$ is realized by $ F_g^\mu$ for every $I$.
Then Theorem \ref{thm-abstract-rec}
give us the following recursion relation for $ F_g^\mu$
(see eg. \cite[Theorem 3.3]{dmss}):
\be
\begin{split}
\mu_1\cdot  F_g^{\mu}=&
\sum_{j=2}^n(\mu_1+\mu_j-2) F_g^{(\mu_1+\mu_j-2,\mu_{[n]\backslash\{1,j\}})}\\
&+\sum_{\substack{\alpha+\beta=\mu_1-2\\\alpha\geq 1,\beta\geq 1}}
\alpha\beta\bigg(
 F_{g-1}^{(\alpha,\beta,\mu_{[n]\backslash\{1\}})}
+\sum_{\substack{g_1+g_2=g\\I\sqcup J=[n]\backslash\{1\}}}
 F_{g_1}^{(\alpha,\mu_I)}\cdot F_{g_2}^{(\beta,\mu_J)}
\bigg)\\
&+  2(\mu_1-2)\cdot
 F_{g}^{(\mu_1-2,\mu_{[n]\backslash\{1\}})}
+\delta_{n,1}\delta_{g,0}\delta_{\mu_1,2},
\end{split}
\ee
where we use the convention $ F_{g}^{(\mu_1-2,\mu_{[n]\backslash\{1\}})}:=0$
for $\mu_1< 2$.

\end{Example}

\begin{Example}
\label{eg-realization2}
Now let us consider the following Feynman rule:
\be\label{eq-FR-eg2}
\begin{split}
& w_v=1 , \quad\forall v\in V(\Gamma);\\
& w_e=x , \quad\forall e\in E(\Gamma);\\
& w_f=y , \quad\forall f\in F(\Gamma),\\
\end{split}
\ee
where $x$ and $y$ are two a formal variables.
Then the realized abstract correlator $ F_g^\mu$
of genus $g$ and type $\mu=(\mu_1,\cdots,\mu_n)$ is:
\be\label{eq-fe-eg2}
\begin{split}
 F_g^\mu=&\sum_{\Gamma\in\mathfrak{Fat}_g^{\mu,c}}
\frac{1}{|\Aut(\Gamma)|}x^{|E(\Gamma)|}y^{|F(\Gamma)|}\\
=&\sum_{\Gamma\in\mathfrak{Fat}_g^{\mu,c}}
\frac{1}{|\Aut(\Gamma)|}x^{\half|\mu|}y^{2-2g-n+\half|\mu|}.
\end{split}
\ee
An easy observation is that:
for every fat graph $\Gamma$ and every half-edge $h\in H(v_1)$,
we have the following relations:
\ben
&&|E(\Gamma^h)|=|E(\Gamma)|-1,\\
&&|F(\Gamma^h)|=|F(\Gamma)|.
\een
(Recall that $\Gamma^h$ is obtained from $\Gamma$ by contracting
the edge containing $h$,
see Definition \ref{def-abs-opr}).
Let $e$ be the edge containing $h$, then we have:
\ben
&&|V(\Gamma^h)|=|V(\Gamma)|-1,\quad\text{if $e$ is not a loop};\\
&&|V(\Gamma^h)|=|V(\Gamma)|+1,\quad\text{if $e$ is a loop}.
\een
From the above relations we know that:
\ben
w_{\Gamma^h}=x^{-1}\cdot w_{\Gamma},
\een
therefore in this case the operator $\tK_1$ acting on $ F_g^\mu (x,y)$
can be taken to be multiplying by $\mu_1\cdot x^{-1}$.
In this case it is also clear that
$\cF_{g,I}^\mu$ is realized by $ F_g^\mu$ for every $I$,
thus Theorem \ref{thm-abstract-rec}
give us the following:
\begin{Proposition}
\label{prop-eg-realization2}
For the Feynman rule \eqref{eq-FR-eg2} and \eqref{eq-fe-eg2},
we have the following quadratic recursion relation:
\be
\begin{split}
\frac{\mu_1}{x}\cdot  F_g^{\mu}=&
\sum_{j=2}^n(\mu_1+\mu_j-2) F_g^{(\mu_1+\mu_j-2,\mu_{[n]\backslash\{1,j\}})}\\
&+\sum_{\substack{\alpha+\beta=\mu_1-2\\\alpha\geq 1,\beta\geq 1}}
\alpha\beta\bigg(
 F_{g-1}^{(\alpha,\beta,\mu_{[n]\backslash\{1\}})}
+\sum_{\substack{g_1+g_2=g\\I\sqcup J=[n]\backslash\{1\}}}
 F_{g_1}^{(\alpha,\mu_I)}\cdot F_{g_2}^{(\beta,\mu_J)}
\bigg)\\
&+  2(\mu_1-2)\cdot y \cdot
 F_{g}^{(\mu_1-2,\mu_{[n]\backslash\{1\}})}
+\delta_{n,1}\delta_{g,0}\delta_{\mu_1,2}\cdot y^2,
\end{split}
\ee
where we use the convention $ F_{g}^{(\mu_1-2,\mu_{[n]\backslash\{1\}})}:=0$
for $\mu_1< 2$.
\end{Proposition}

\end{Example}

\subsection{Hermitian one-matrix model as a realization of the abstract QFT}
\label{sec-realization-herm-1mm}

In this subsection
we interpret the Hermitian one-matrix models as a realization
of the abstract QFT for fat graphs.

First let us recall the definition of Hermitian one-matrix models
and the genus expansion of the free energy using fat graphs (at finite $N$),
see \cite{zhou1}.
See \cite{me} for an introduction for the history
and various developments of matrix models.

Denote by $\cH_N$ the space of Hermitian matrices of size $N\times N$.
In what follows,
we use a superscript `Herm' to indicate the case of Hermitian one-matrix models.
The partition function $Z_N^{\text{Herm}}$ of the Hermitian one-matrix model is defined to be:
\be
\label{eq-def-1mm-partition}
Z_N^{\text{Herm}}=\frac{
\int_{\cH_N}dM \exp\bigg(\tr\sum\limits_{n=1}^\infty\frac{g_n-\delta_{n,2}}{ng_s}M^n\bigg)
}{
\int_{\cH_N}dM \exp\bigg(-\frac{1}{2g_s}\tr(M^2)\bigg)
},
\ee
where $dM$ is the Haar measure:
\ben
dM=2^{N(N-1)/2}\prod_{i=1}^N dM_{ii} \prod_{1\leq i<j\leq N}d Re(M_{ij}) d Im(M_{ij}).
\een
The free energy $F_N^{\text{Herm}}$ is defined to be
\be
F_N^{\text{Herm}}:=\log Z_N^{\text{Herm}}.
\ee

One way to define a genus expansion of the free energy $F_N^{\text{Herm}}$
is to introduce the 't Hooft coupling constant
\be
t:=N g_s.
\ee
Using the 't Hooft coupling constant,
$F_N^{\text{Herm}}$ can be expressed as:
\be
F_N^{\text{Herm}}=\sum_{g\geq 0}g_s^{2g-2}F_g^{\text{Herm}}(t),
\ee
where $F_g^{\text{Herm}}(t)=\sum\limits_{m\geq 1}F_{g,m}^{\text{Herm}} t^m$ is a formal power series in $t$.

\begin{Remark}
This genus expansion is called the `fat genus expansion' in an earlier work
\cite{zhou1} of the second author.
In that work another genus expansion (the `thin genus expansion')
is introduced by regarding the fat graphs as thin graphs.
\end{Remark}

Now let $\lambda=(\lambda_1\geq\cdots\geq\lambda_l)$ be a partition.
Then the fat correlators $\langle p_\lambda\rangle_N$ is defined to be:
\begin{equation*}
\begin{split}
\langle p_\lambda\rangle_N=&g_s^{-l}
\langle(\tr(M^{\lambda_1}))\cdots (\tr(M^{\lambda_l}))
\rangle_N\\
=&\frac{
\int_{\cH_N}dM (\tr(M^{\lambda_1}))\cdots (\tr(M^{\lambda_l}))
\exp\bigg(-\frac{1}{2g_s}\tr(M^2)\bigg)
}{
\int_{\cH_N}dM \exp\bigg(-\frac{1}{2g_s}\tr(M^2)\bigg)
}.
\end{split}
\end{equation*}
The connected fat correlators of genus $g$
\be\label{eq-1mm}
\langle \frac{1}{z_\lambda} p_\lambda \rangle_g^c (t):=
\frac{\pd^l F_g^{\text{Herm}}(t)}{\pd g_{\lambda_1}\cdots \pd g_{\lambda_l}}
\bigg|_{g_i=0,i\geq 1,2,\cdots}
\ee
can be expressed as the following summation over all connected fat graphs
of genus $g$ and type $\lambda$:
\be\label{eq-1mm-cor}
\langle \frac{1}{z_\lambda} p_\lambda \rangle_g^c (t)=
\sum_{\Gamma\in\widetilde{\mathfrak{Fat}}_g^{\lambda, c}}
\frac{1}{|\Aut(\Gamma)|}
t^{|F(\Gamma)|},
\ee
where $z_\lambda:=\prod\limits_{i=1}^m (i^{m_i} \cdot m_i !)$ for a partition
$\lambda=(1^{m_1}2^{m_2}\cdots)$,
and $\widetilde{\mathfrak{Fat}}_g^{\lambda,c}$ is the set of fat graphs with vertices
of valences $\lambda_1,\cdots,\lambda_n$ while we ignore the orders of vertices:
\ben
\widetilde{\mathfrak{Fat}}_g^{\lambda,c}=
\bigcup_{\sigma\in S_n}\mathfrak{Fat}_g^{(\lambda_{\sigma(1)},\cdots,\lambda_{\sigma(n)}),c}.
\een

The relation \eqref{eq-euler} gives us the following fat selection rule:
$\langle \frac{1}{z_\lambda} p_\lambda \rangle_g^c (t)=0$,
unless
\ben
2g-2=\half |\lambda|-l(\lambda)-m
\een
for some $m\geq 1$.

Now let us understand the Hermitian one-matrix models as a realization
of the abstract QFT for fat graphs.
The summation \eqref{eq-1mm-cor} suggests us to consider the following Feynman rule:
\be\label{eq-1mm-FR}
\begin{split}
&w_v:=1, \qquad \forall v\in V(\Gamma);\\
&w_e:=1, \qquad \forall e\in E(\Gamma);\\
&w_f:=t, \qquad \forall f\in F(\Gamma).\\
\end{split}
\ee
Then the abstract correlator $\cF_g^\mu$ is realized by:
\be
\label{thm-1mm-realfe}
F_g^\mu (t)=\sum_{\Gamma\in\mathfrak{Fat}_g^{\mu,c}}
\frac{1}{|\Aut(\Gamma)|}t^{|F(\Gamma)|}
=\langle\frac{p_{\mu_1}}{\mu_1}\cdots\frac{p_{\mu_n}}{\mu_n}
\rangle_g^c.
\ee
Then \eqref{eq-1mm} gives us:
\be
F_g^{\text{Herm}}(t)=\sum_{n\geq 1}\sum_{\mu\in \bZ_{>0}^n}\frac{1}{n!}\cdot F_g^\mu(t)
\cdot g_{\mu_1}\cdots g_{\mu_n},
\ee
or equivalently,
\be
F_N^{\text{Herm}}=\sum_{g\geq 0}g_s^{2g-2}\sum_{(m)\in\mathcal{P}}
\bigg( F_g^{\lambda_{(m)}}(t)\cdot
\frac{\prod_{j\geq 1}g_j^{m_j}}{\prod_{j\geq 1}m_j !}\bigg),
\ee
where the notations $\mathcal P$ and $\lambda_{(m)}$
are the same as in \S \ref{sec-abstractpartition}.
This simply tells us that the abstract correlators $\cF_g^\mu$
are realized by the correlators of the Hermitian one-matrix models
if we assign Feynman rule \eqref{eq-1mm-FR}.
Therefore,
the Hermitian one-matrix model is a refinement of the enumeration of fat graphs
(see Example \ref{eg-enum}) where the degree of $t$ encoded the number of faces in each graph.

There is another point of view of the above realization.
Instead of considering the realization of the abstract correlators
by the Feynman rule \eqref{eq-1mm-FR},
we may consider the realization of the abstract free energy
and abstract partition function by the following Feynman rule:
\be\label{eq-1mm-FR2}
\begin{split}
&w_v:=g_{\val(v)}, \qquad \forall v\in V(\Gamma);\\
&w_e:=1, \qquad \forall e\in E(\Gamma);\\
&w_f:=t, \qquad \forall f\in F(\Gamma),\\
\end{split}
\ee
where $\val(v)$ is the valence of vertex $v\in V(\Gamma)$.
In this way,
the abstract free energy $\cF_g$, $\cF$ and abstract partition function $ \cZ $
are realized by $F_g^{\text{Herm}}$, $F_N^{\text{Herm}}$ and $Z_N^{\text{Herm}}$ respectively.
(Here the realization of the empty graph `1'
in Definition \ref{def-fat-abs-par} is simply $1$.)

In the next subsection,
we will discuss the realization of
the quadratic reursion relation \eqref{eq-abstract-rec} in this example in detail,
and such a consideration inspires the formalisms
developed in \S \ref{sec-fat-abs-virasoro} and \S \ref{sec-abstract-eorec}.

\subsection{Realization of the abstrat quadratic recursion relation}
\label{sec-fatreal-qrec}

In this subsection
let us construct the realization of the recursion relation \eqref{eq-abstract-rec}
in the example of the Hermitian one-matrix models.
This provides a quadratic recursion relation for
the fat correlators \eqref{eq-1mm} or \eqref{thm-1mm-realfe}.
We will follow the notations in \S \ref{sec-realization-herm-1mm}.

Notice that the Feynman rule \eqref{eq-1mm-FR} is a special case
of Example \ref{eg-realization2} if we take $x=1$ and $y=t$.
Therefore,
Proposition \ref{prop-eg-realization2} gives us the following:

\begin{Theorem}
\label{thm-1mm-rec}
The fat correlators $ F_g^\mu(t)$
of the Hermitian one-matrix models
can be determined by the quadratic recursion:
\be
\label{eq-1mm-rec}
\begin{split}
&\mu_1\cdot  F_g^{\mu}(t)=
\sum_{j=2}^n(\mu_1+\mu_j-2) F_g^{(\mu_1+\mu_j-2,\mu_{[n]\backslash\{1,j\}})}(t)\\
&\qquad+\sum_{\substack{\alpha+\beta=\mu_1-2\\\alpha\geq 1,\beta\geq 1}}
\alpha\beta\bigg(
 F_{g-1}^{(\alpha,\beta,\mu_{[n]\backslash\{1\}})}(t)
+\sum_{\substack{g_1+g_2=g\\I\sqcup J=[n]\backslash\{1\}}}
 F_{g_1}^{(\alpha,\mu_I)}(t)\cdot F_{g_2}^{(\beta,\mu_J)}(t)
\bigg)\\
&\qquad+
 2(\mu_1-2)\cdot t\cdot
 F_{g}^{(\mu_1-2,\mu_{[n]\backslash\{1\}})}(t)
+\delta_{n,1}\delta_{g,0}\delta_{\mu_1,2}\cdot t^2
\end{split}
\ee
together with the initial value $F_0^{(0)}(t):=t$,
where we use the convention:
\ben
 F_{g}^{(\mu_1-2,\mu_{[n]\backslash\{1\}})}(t):=0,
\qquad \text{if}\quad \mu_1< 2.
\een
\end{Theorem}

\begin{Example}
\label{eg-1mm-correlator}
First consider the case $(g,n)=(0,1)$.
The recursion \eqref{eq-1mm-rec} gives us:
\ben
2m\cdot  F_g^{(2m)}(t)=
\sum_{\substack{\alpha+\beta=2m-2\\ \alpha\geq 0, \beta\geq 0}}
 F_0^{(\alpha)}(t)\cdot F_0^{(\beta)}(t).
\een
Thus we have:
\ben
 F_0^{(4)}(t)=\half t^3 ,\quad
 F_0^{(6)}(t)=\frac{5}{6}t^4,\quad
 F_0^{(8)}(t)=\frac{7}{4}t^5,\quad
 F_0^{(10)}(t)=\frac{21}{5}t^6,
\quad \cdots
\een
It is known that
$ F_0^{(2m)}(t)=\frac{1}{2m}C_m  t^{m+1}$
where $C_m=\frac{1}{m+1}\binom{2m}{m}$ is the Catalan numbers
(see eg. \cite{dmss}),
and $C_0:=0$.

Now let us consider the case $(g,n)=(0,2)$.
This case has been calculated in \cite{kp}.
Using the quadratic recursion relation \eqref{eq-1mm-rec},
we have:
\begin{equation*}
\begin{split}
\mu_1\cdot F_0^{(\mu_1,\mu_2)}=&
(\mu_1+\mu_2-2) F_0^{(\mu_1+\mu_2-2)}
+2\sum_{\substack{\alpha+\beta=\mu_1-2\\\alpha\geq 1,\beta\geq 1}}
\alpha\beta
 F_{0}^{(\alpha,\mu_2)}\cdot F_{0}^{(\beta)}\\
&+
2(\mu_1-2)\cdot t\cdot
 F_{g}^{(\mu_1-2,\mu_2)}.
\end{split}
\end{equation*}
In particular, for $\mu_1=1$, we have:
\ben
 F_0^{(1,2m+1)}=2m\cdot F_0^{(2m)}=C_m t^{m+1},
\qquad \forall m\geq 0;
\een
for $\mu_1=2$, we have:
\ben
 F_0^{(2,2m)}=\frac{2m}{2}\cdot F_0^{(2m)}=\half C_m t^{m+1},
\qquad \forall m\geq 1;
\een
for $\mu_1=3$, we have:
\begin{equation*}
\begin{split}
 F_0^{(3,2m-1)}=&\frac{2m}{3}\cdot F_0^{(2m)}+\frac{2t}{3}\cdot F_0^{(1,2m-1)}\\
=&(\frac{1}{3}C_m+\frac{2}{3}C_{m-1})t^{m+1},
\qquad \forall m\geq 1.
\end{split}
\end{equation*}
\end{Example}

\begin{Example}
Using the quadratic recursion relation \eqref{eq-1mm-rec} for the fat correlators
together with the initial value $ F_g^{(0)}:=t$,
we can compute first few terms of the free energy
of the Hermitian one-matrix models:
\begin{equation*}
\begin{split}
F_0^{\text{Herm}}=&
t(\half g_1^2+\half g_1^2 g_2+\half g_1^2 g_2^2+\frac{1}{3}g_1^3 g_3+\cdots)\\
&+t^2(\half g_2+\frac{1}{4} g_2^2+ g_1g_3+\frac{1}{6}g_2^3+\cdots)
+t^3(\half g_4+\frac{2}{3}g_3^2+\cdots)\\
&+t^4(\frac{5}{6}g_6+\cdots)+t^5(\frac{7}{4}g_8+\cdots)+\cdots,\\
F_1^{\text{Herm}}=&
t(\frac{1}{4}g_4+g_1g_5+\half g_2g_4+\cdots)
+t^2(\frac{5}{3}g_6+\cdots)+\cdots.
\end{split}
\end{equation*}

\end{Example}

\subsection{Equivalence to fat Virasoro constraints}
\label{sec-fatVirasoro}

In this subsection,
we first recall the fat Virasoro constraints for
the partition function of Hermitian one-matrix models,
and then prove that it is equivalent to the quadratic recursion relation
derived in last subsection.

The partition function $Z_N^{\text{Herm}}$ of the Hermitian one-matrix models
satisfies the following fat Virasoro constraints:
\be\label{eq-fatVirasoro}
L_{m,t}^{\text{Herm}} Z_N^{\text{Herm}}=0, \qquad \forall m\geq -1,
\ee
where the fat Virasoro operators $L_{m,t}$ are defined by:
\be
\label{eq-1mm-fatvir-opr-1}
\begin{split}
&L_{-1,t}^{\text{Herm}}=-\frac{\pd}{\pd g_1}+
\sum_{n\geq 1}ng_{n+1}\frac{\pd}{\pd g_n}
+tg_1g_s^{-2},\\
&L_{0,t}^{\text{Herm}}=-2\frac{\pd}{\pd g_2}+
\sum_{n\geq 1}ng_{n}\frac{\pd}{\pd g_n}
+t^2g_s^{-2},\\
&L_{1,t}^{\text{Herm}}=-3\frac{\pd}{\pd g_3}+
\sum_{n\geq 1}(n+1)g_{n}\frac{\pd}{\pd g_{n+1}}
+2t\frac{\pd}{\pd g_1},\\
\end{split}
\ee
and for $m\geq 2$,
\be
\label{eq-1mm-fatvir-opr-2}
L_{m,t}^{\text{Herm}}=\sum_{k\geq 1}(k+m)(g_k-\delta_{k,2})\frac{\pd}{\pd g_{k+m}}
+g_s^2\sum_{k=1}^{m-1}k(m-k)\frac{\pd}{\pd g_k}
\frac{\pd}{\pd g_{m-k}}+2tm\frac{\pd}{\pd g_m},
\ee
where $t=Ng_s$ is the 't Hooft coupling constant.
Moreover,
these operators $\{L_{m,t}^{\text{Herm}}\}$ satisfies the Virasoro commutation relations:
\be
[L_{m,t}^{\text{Herm}},L_{m,t}^{\text{Herm}}]=(m-n)L_{m+n,t}^{\text{Herm}},
\qquad \forall m,n\geq -1.
\ee
The fat Virasoro constraints \eqref{eq-fatVirasoro}
is a family of quadratic recursion relations,
and the free energies $F_{g}^{\text{Herm}}(t)$ can be uniquely determined from them.

Our main result in this subsection is the following:

\begin{Theorem}
The fat Virasoro constraints \eqref{eq-fatVirasoro}
is equivalent to the quadratic recursion relation \eqref{eq-1mm-rec}.
\end{Theorem}
\begin{proof}
Let us rewrite the Virasoro constraints \eqref{eq-fatVirasoro} for the free energy
\ben
F_N^{\text{Herm}}=\sum_{g\geq 0}g_s^{2g-2}
\sum_{n\geq 1}\sum_{\mu_1,\cdots,\mu_n>0}\frac{1}{n!}
\langle \frac{p_{\mu_1}}{\mu_1}\cdots\frac{p_{\mu_n}}{\mu_n}\rangle_g^c
\cdot g_{\mu_1}\cdots g_{\mu_n}
\een
as constraints for the connected fat correlators
$\langle \frac{p_{\mu}}{\mu}\rangle_g^c=
\langle \frac{p_{\mu_1}}{\mu_1}\cdots \frac{p_{\mu_n}}{\mu_n}\rangle_g^c$:
\begin{equation}
\label{eq-1mm-Vir-cor-1}
\begin{split}
\langle p_1 \frac{p_{\mu_1}}{\mu_1}\cdots \frac{p_{\mu_n}}{\mu_n}\rangle_g^c
=&\sum_{j=1}^n (\mu_j-1)\cdot
\langle \frac{p_{\mu_1}}{\mu_1}\cdots \frac{p_{\mu_j-1}}{\mu_j-1}
\cdots\frac{p_{\mu_n}}{\mu_n}\rangle_g^c,\\
\langle p_2 \frac{p_{\mu_1}}{\mu_1}\cdots \frac{p_{\mu_n}}{\mu_n}\rangle_g^c
=&\sum_{j=1}^n \mu_j \cdot
\langle \frac{p_{\mu_1}}{\mu_1}\cdots
\frac{p_{\mu_n}}{\mu_n}\rangle_g^c,\\
\langle p_3 \frac{p_{\mu_1}}{\mu_1}\cdots \frac{p_{\mu_n}}{\mu_n}\rangle_g^c
=&\sum_{j=1}^n (\mu_j+1)\cdot
\langle \frac{p_{\mu_1}}{\mu_1}\cdots \frac{p_{\mu_j+1}}{\mu_j+1}
\cdots\frac{p_{\mu_n}}{\mu_n}\rangle_g^c\\
&+2t\cdot\langle p_1 \frac{p_{\mu_1}}{\mu_1}\cdots\frac{p_{\mu_n}}{\mu_n}
\rangle_g^c,\\
\end{split}
\end{equation}
and for $m\geq 4$,
\begin{equation}
\label{eq-1mm-Vir-cor-2}
\begin{split}
&\langle p_m \frac{p_{\mu_1}}{\mu_1}\cdots \frac{p_{\mu_n}}{\mu_n}\rangle_g^c
=\sum_{j=1}^n (\mu_j+m-2)\cdot
\langle \frac{p_{\mu_1}}{\mu_1}\cdots \frac{p_{\mu_j+m-2}}{\mu_j+m-2}
\cdots\frac{p_{\mu_n}}{\mu_n}\rangle_g^c\\
&\qquad\qquad
+2t\cdot\langle p_{m-2}\frac{p_{\mu_1}}{\mu_1}\cdots\frac{p_{\mu_n}}{\mu_n}
\rangle_g^c
+\sum_{k=1}^{m-3}\langle p_k p_{m-2-k}
\frac{p_{\mu_1}}{\mu_1}\cdots\frac{p_{\mu_n}}{\mu_n}
\rangle_{g-1}^c\\
&\qquad\qquad
+\sum_{k=1}^{m-3}\sum_{\substack{g_1+g_2=g\\I\sqcup J=[n]}}
\langle p_k \frac{p_{\mu_I}}{\mu_I}\rangle_{g_1}^c
\cdot
\langle p_{m-2-k} \frac{p_{\mu_J}}{\mu_J}\rangle_{g_2}^c.
\end{split}
\end{equation}
Notice that $ F_g^\mu(t)=
\langle \frac{p_{\mu_1}}{\mu_1}\cdots \frac{p_{\mu_n}}{\mu_n}\rangle_g^c$
is symmetric with respect to $\mu_1,\mu_2,\cdots,\mu_n$,
therefore the above recursion is equivalent to the quadratic recursion
\eqref{eq-1mm-rec}.
\end{proof}

\subsection{Matrix models of other types as realizations of the abstract QFT}

In this subsection,
we discuss some matrix models of other types.
We regard them as different realizations of the
abstract QFT for fat graphs.

\begin{Example}
[Penner model]

The Penner matrix model \cite{pe} was introduced to
give a proof to the Harer-Zagier formula \cite{hz},
which computes the orbifold Euler characteristics of
the moduli space of stable curves $\cM_{g,n}$ for $2g-2+n>0$.
See also \cite{mu2}.

The partition function of the Penner model is:
\be
Z_N^{Pe}=\frac{
\int_{\cH_N}dM
\exp\bigg(-\half \tr (M^2)\bigg)
\exp\bigg(-\sum\limits_{n=3}^\infty\frac{(\sqrt{z})^{n-2}}{n}\tr(M^n)\bigg)
}{
\int_{\cH_N}dM \exp\bigg(-\frac{1}{2}\tr(M^2)\bigg)
},
\ee
where $z$ is a formal variable.
One can easily see that
this is a special case of the Hermitian one-matrix model,
after taking the particular choice
$g_1=g_2=0$ and
$g_n = -(\sqrt{z})^{n-2}$ for $n\geq 3$.
In other words,
if one assign the Feynman rule:
\be
\begin{split}
&w_v:=0, \qquad\qquad\qquad\quad \val(v)=1,2;\\
&w_v:=-(\sqrt{z})^{\val(v)-2}, \qquad  \val(v)\geq 3;\\
&w_e:=1, \qquad\qquad\qquad\quad \forall e\in E(\Gamma);\\
&w_f:=N, \qquad\qquad\qquad\quad \forall f\in F(\Gamma),\\
\end{split}
\ee
then the abstract partition function $ \cZ $ is realized by $Z_N^{Pe}$,
and the abstract free energy $\cF$ is realized by $\log(Z_N^{Pe})$
(here we take $g_s=1$).
In this case $\log(Z_N^{Pe})$ is a summation over stable fat graphs:
\begin{equation*}
\begin{split}
\log(Z_N^{Pe})=&\sum_{g\geq 0}\sum_{n\geq 1}\sum_{\mu\in \bZ_{\geq 3}^n}\frac{1}{n!}
\sum_{\Gamma\in \mathfrak{Fat}_g^{\mu,st,c}}\frac{(-1)^{|V(\Gamma)|}
z^{|E(\Gamma)|-|V(\Gamma)|}N^{|F(\Gamma)|}}{|\Aut(\Gamma)|}\\
=&
\sum_{g\geq 0}\sum_{n\geq 1}\sum_{\mu\in \bZ_{\geq 3}^n}\frac{1}{n!}
\sum_{\Gamma\in \mathfrak{Fat}_g^{\mu,st,c}}\frac{(-1)^{n}
z^{\half |\mu|-n}N^{2-2g-n+\half |\mu|}}{|\Aut(\Gamma)|},
\end{split}
\end{equation*}
where $\mathfrak{Fat}_g^{\mu,st,c}$ is the set of all connected stable fat graphs
of genus $g$ and type $\mu$.

\end{Example}

\begin{Example}
[Kontsevich model]

The Kontsevich model was introduced by Kontsevich \cite{kon}
in his proof of Witten Conjecture \cite{wit}.
In that work,
Kontsevich related the correlators of such a matrix integration
to the intersection numbers of $\psi$-classes on the
Deligne-Mumford moduli spaces $\Mbar_{g,n}$ of stable curves
(see \cite[\S 3, the Main Identity]{kon}).

The partition function of this matrix model is:
\be
Z_N^{Ko}=\frac{
\int_{\cH_N}dX
\exp\bigg(-\half \tr (\Lambda X^2)\bigg)
\exp\bigg(\frac{\sqrt{-1}}{6}\tr(X^3)\bigg)
}{
\int_{\cH_N}dX \exp\bigg(-\frac{1}{2}\tr(\Lambda X^2)\bigg)
},
\ee
where $\Lambda=\diag(\lambda_1,\cdots, \lambda_N)$ is a diagonal matrix.
Kontsevich showed that the free energy $\log(Z^{Ko})$ can be
represented as a summation over trivalent fat graphs.
The Feynman rule is as follows.
Given a fat graph $\Gamma$ whose vertices are all of valence three,
denote by $f:=|F(\Gamma)|$ the number of faces of $\Gamma$.
We symbolically label $\lambda_{i_1},\lambda_{i_2},\cdots,\lambda_{i_f}$ to these $n$ faces
in an arbitrary way,
and require the weight of an edge $e$ to be:
\be
\label{eq-Kmodel-FR-edge}
w_e:=\frac{2}{\lambda_i+\lambda_j}
\ee
if the two sides of $e$ are incident at the two faces labelled by
$\lambda_i$ and $\lambda_j$ respectively
(there might be the case $i=j$).
Then the weight of the graph $\Gamma$ is given by:
\be
\label{eq-Kont-Feynmanrule}
w_\Gamma^{Ko}:=2^{-2|V(\Gamma)|}\cdot
\sum_{i_1,\cdots,i_f=1}^N \frac{1}{f!}\bigg(
\prod_{e\in E(\Gamma)}w_e+\perm.\bigg),
\ee
where $\perm.$ is the summation of all such terms
obtained by permutations of the labels $\lambda_{i_1},\cdots,\lambda_{i_f}$
on the fat graph $\Gamma$.
Then the free energy of the Kontsevich model can be represented as:
\be
\label{eq-Kmodel-FR}
F^{Ko}:=
\log(Z_N^{Ko})=\sum_{\Gamma\in \mathfrak{Fat}^{tri,c}}\frac{1}{|V(\Gamma)|!}\cdot
\frac{w_\Gamma^{Ko}}{|\Aut(\Gamma)|},
\ee
where $\mathfrak{Fat}^{tri,c}$ is the set of all connected trivalent fat graphs.

Now let us set $w_\Gamma=0$ if a fat graph $\Gamma$ is not trivalent,
and assign \eqref{eq-Kont-Feynmanrule} for a trivalent fat graph $\Gamma$,
then it is clear that the free energy $F^{Ko}$ is a realization of the
abstract free energy $\cF$
(see Definition \ref{def-fat-abs-fe}, here we take $g_s=1$).

\end{Example}

\begin{Example}
[Kontsevich-Penner model]

The Kontsevich-Penner model was introduced by Chekhov and Makeenko in \cite{cm}.
This model describes the intersection theory on the discretized moduli space
due to Chekhov \cite{ch,ch2}.

The partition function of the Kontsevich-Penner model is:
\be
Z_N^{K-P}=
\int_{\cH_N}dX
\exp\bigg[
N\cdot\tr\bigg(
-\half \Lambda X\Lambda X+\alpha\big(\log(1+X)-X\big)
\bigg)\bigg],
\ee
where $\Lambda=\diag(\mu_1,\mu_2,\cdots, \mu_N)$.
The Feynman rules for the Kontsevich-Penner model has been discussed in \cite{ch}.
In this case,
the free energy is represented by the following summation:
\be
\label{eq-KPmodel-FR}
F_N^{K-P}=\sum_{\Gamma\in \mathfrak{Fat}^{st,c}}\frac{1}{|V(\Gamma)|!}\cdot
\frac{w_\Gamma^{K-P}}{|\Aut(\Gamma)|},
\ee
where $\mathfrak{Fat}^{st,c}$ is the set of connected stable fat graphs,
and the weight $w_\Gamma^{K-P}$ of a stable fat graph $\Gamma$
(not necessarily trivalent) is given by
the same formula \eqref{eq-Kont-Feynmanrule} for the Kontsevich model,
where the weight of an edge $e$ is replace by:
\be
\label{eq-KPmodel-FR-edge}
w_e:=\frac{2}{\mu_{i}\mu_{j}-1}
\ee
if the two sides of $e$ are incident at the two faces labelled by
$\mu_i$ and $\mu_j$ respectively.
Similar to the Kontsevich model,
the Kontsevich-Penner model also provides us a natural realization
of the abstract QFT for fat graphs.

\end{Example}

Not like the case of the Hermitian one-matrix models
where the edge-contraction operators on graphs can be simply realized by some
multiplications and partial derivatives,
the realization of the edge-contraction operators in the above three examples
are much more complicated.
It is a natural question to realize the quadratic recursion for these models
and relate the resulting recursions to some known recursions such as
Virasoro constraints and Eynard-Orantin topological recursions.
We hope to address such a problem in future works.

\section{Abstract Virasoro Constraints for Fat Graphs}
\label{sec-fat-abs-virasoro}

We have already shown that when we regard the Hermitian one-matrix models
as a realization of the abstract QFT for fat graphs,
the realization of the quadratic recursion \eqref{eq-abstract-rec}
is equivalent to the fat Virasoro constraints.
This inspires us to consider the following question --
Can we reformulate this quadratic recursion
in terms of a family of abstract operators which annihilate
the abstract partition function $\cZ$,
such that these abstract operators satisfy the Virasoro commutation relations?
We will answer this question affirmatively in this section.

First,
we construct some operators (called `vertex-splitting') acting on fat graphs
which `inverse' the edge-contraction procedures.
Then we use them to formulate a family of abstract Virasoro operators
which annihilate the abstract partition function $\cZ$.
We show that these abstract Virasoro operators satisfies the
Virasoro `commutation' relations in a special sense.
Furthermore,
we derive a formula of the form $\cZ=e^{\cM}(1)$ for the abstract partition function $\cZ$
using the abstract Virasoro constraints.

Since in this section we are dealing with the
abstract free energies and abstract partition function,
we will abandon the labels $v_1,\cdots,v_n$
on vertices of fat graphs during the whole section
(see \S \ref{sec-abstractpartition}).

\subsection{Fat graphs with hollow vertices}
\label{sec-hollow-ord}

In this subsection,
let us introduce a new type of fat graphs
(called `fat graphs with hollow vertices'),
and some abstract operators relating the two types of fat graphs.

In previous sections,
we draw vertices of fat graphs as some solid dots (see \S \ref{sec-absqft}).
From now on,
we will introduce a new type of vertices,
which is drawn as hollow dots in a fat graph.
We call these vertices `hollow vertices'.
A `fat graphs with hollow vertices' is a fat graph consisting of some
solid vertices, some hollow vertices,
and internal edges connecting these vertices.
The notions of `genus', `valence', etc. for fat graphs with hollow vertices
are defined as usual.
Notice that in the definition of the automorphism group $\Aut(\Gamma)$,
we only consider the automorphism which preserves each vertex of $\Gamma$
(see \S \ref{sec-cellgraph}),
therefore it is clear that
\ben
\Aut(\Gamma)\cong\Aut(\Gamma')
\een
where $\Gamma'$ is a fat graph obtained from $\Gamma$ by
changing some solid (resp. hollow) vertices into hollow (resp. solid) ones.

Define
\ben
\cV^{ho}:=\prod_{\Gamma\in \mathfrak{Fat}^{ho}}\bQ[g_s,g_s^{-1}]\cdot\Gamma,
\een
where $\mathfrak{Fat}^{ho}$ is the set of all
fat graphs with hollow vertices (not necessarily connected).
Now let us define some operators on this space
which produce new hollow vertices in fat graphs.

\begin{Definition}
Given a positive integer $n$,
define the linear operator $\pd_n$ on $\cV^{ho}$ by:
\be
\pd_n (\Gamma):=\sum_{v\in V_n(\Gamma)}\pd (\Gamma,v),
\ee
where $V_n(\Gamma)$ is the set of solid vertices of valence $n$ in $\Gamma$,
and $\pd (\Gamma,v)$ is the graph obtained from $\Gamma$ by
changing this solid vertex $v$ into a hollow one.
\end{Definition}

It is not hard to see that:
\begin{Lemma}
\label{eq-pd-commute}
For every $m,n\geq 1$,
we have $\pd_m\pd_n=\pd_n\pd_m$.
\end{Lemma}

\begin{Example}

We give some examples of $\pd_n$:
\begin{flalign*}
\begin{tikzpicture}[scale=0.9]
\node [align=center,align=center] at (3.3,0) {$\pd_3\biggl($};
\draw [fill] (5,0) circle [radius=0.0575];
\draw [->,>=stealth] (-0.2+5,-0.1) .. controls (0.3+5,-0.5) and (0.3+5,0.5) ..  (-0.2+5,0.1);
\draw [thick] (0+5,0) .. controls (1.5+5,-1.2) and (1.5+5,1.2) ..  (0+5,0);
\draw [thick] (0+5,0) .. controls (-1.5+5,-1.2) and (-1.5+5,1.2) ..  (0+5,0);
\node [align=center,align=center] at (6.85,0) {$\bigg)=0,$};
\end{tikzpicture}&&
\end{flalign*}
\begin{flalign*}
\begin{tikzpicture}[scale=0.9]
\node [align=center,align=center] at (3.3,0) {$\pd_4\biggl($};
\draw [fill] (5,0) circle [radius=0.0575];
\draw [->,>=stealth] (-0.2+5,-0.1) .. controls (0.3+5,-0.5) and (0.3+5,0.5) ..  (-0.2+5,0.1);
\draw [thick] (0+5,0) .. controls (1.5+5,-1.2) and (1.5+5,1.2) ..  (0+5,0);
\draw [thick] (0+5,0) .. controls (-1.5+5,-1.2) and (-1.5+5,1.2) ..  (0+5,0);
\node [align=center,align=center] at (6.7,0) {$\bigg)=$};
\draw (5+3.5,0) circle [radius=0.07];
\draw [->,>=stealth] (-0.2+5+3.5,-0.1) .. controls (0.3+5+3.5,-0.5) and (0.3+5+3.5,0.5) ..  (-0.2+5+3.5,0.1);
\draw [thick] (0+5+3.5+0.05,-0.05) .. controls (1.5+5+3.5,-1.2) and (1.5+5+3.5,1.2) ..  (0+5+3.5+0.05,0.05);
\draw [thick] (0+5+3.5-0.05,-0.05) .. controls (-1.5+5+3.5,-1.2) and (-1.5+5+3.5,1.2) ..  (0+5+3.5-0.05,0.05);
\node [align=center,align=center] at (10,-0.2) {$,$};
\end{tikzpicture}&&
\end{flalign*}
\begin{flalign*}
\begin{tikzpicture}[scale=1.1]
\node [align=center,align=center] at (-0.8,0) {$\pd_2\biggl($};
\draw [fill] (0,0) circle [radius=0.0525];
\draw [fill] (1.5,0) circle [radius=0.0525];
\draw [thick] (0,0) .. controls (0.3,0.6) and (1.2,0.6) ..  (1.5,0);
\draw [thick] (0,0) .. controls (0.3,-0.6) and (1.2,-0.6) ..  (1.5,0);
\draw [->,>=stealth] (1.3,-0.1) .. controls (1.8,-0.5) and (1.8,0.5) ..  (1.3,0.1);
\draw [->,>=stealth] (-0.2,-0.1) .. controls (0.3,-0.5) and (0.3,0.5) ..  (-0.2,0.1);
\node [align=center,align=center] at (2.4,0) {$\bigg)=2$};
\draw (0+3.2+0.17,0) circle [radius=0.0625];
\draw [fill] (1.5+3.2+0.15,0) circle [radius=0.0525];
\draw [thick] (0+3.25+0.15,0.05) .. controls (0.3+3.2+0.15,0.6) and (1.2+3.2+0.15,0.6) ..  (1.5+3.2+0.15,0);
\draw [thick] (0+3.25+0.15,-0.05) .. controls (0.3+3.2+0.15,-0.6) and (1.2+3.2+0.15,-0.6) ..  (1.5+3.2+0.15,0);
\draw [->,>=stealth] (1.3+3.2+0.15,-0.1) .. controls (1.8+3.2+0.15,-0.5) and (1.8+3.2+0.15,0.5) ..  (1.3+3.2+0.15,0.1);
\draw [->,>=stealth] (-0.2+3.2+0.15,-0.1) .. controls (0.3+3.2+0.15,-0.5) and (0.3+3.2+0.15,0.5) ..  (-0.2+3.2+0.15,0.1);
\node [align=center,align=center] at (5.25+0.15,-0.2) {$.$};
\end{tikzpicture}&&
\end{flalign*}
\end{Example}

\subsection{Vertex-splitting operators on fat graphs}
\label{sec-def-S&J}

In this subsection,
we define two families of operators $\{\cS_{k,l}\}$ and $\{\cJ_{k,l}\}$
on fat graphs as `inverses' of edge-contraction procedures.

The `vertex-splitting operators' $\cS_{n,k}$ are defined as follows:

\begin{Definition}
\label{def-vertexsplit-opr}
Let $n\geq k\geq 0$ be two positive integers.
Define vertex-splitting operators $\cS_{n,k}$ to be the linear operator
on $\cV^{ho}$ defined by:
\be
\cS_{n,k}(\Gamma):=\sum_{v\in V_{n}(\Gamma)}\sum_{h\in H(v)}
\cS(\Gamma,h),
\ee
where $V_n(\Gamma)$ is the set of solid vertices of valence $n$ in $\Gamma$,
$H(v)$ is the set of half-edges attached to $v$;
and $\cS(\Gamma,h)$ is the graph obtained from $\Gamma$
by splitting the vertex $v$ into two vertices $v'$ and $v''$,
such that:
\begin{itemize}
\item[1)]
$v'$ and $v''$ are connected by a new internal edge;

\item[2)]
$v'$ is a solid vertex of valence $k+1$.
The $k$ adjacent half-edges in $H(v)$ starting from $h$ (w.r.t. the cyclic order)
are attached to $v'$ in $\cS(\Gamma,h)$,
and in the new cyclic order the new half-edge is placed right before $h$;

\item[3)]
$v''$ is a hollow vertex of valence $n-k+1$.
The $n-k$ remaining half-edges
(say, $h_1,\cdots, h_{n-k}$ according to the cyclic order,
such that $h_{n-k}$ is right before $h$ in $\Gamma$)
in $H(v)$ are attached to $v''$,
such that in the new cyclic order the new half-edge is
between $h_{n-k}$ and $h_1$.
\end{itemize}

\end{Definition}

\begin{Remark}
Notice that in the above definition,
we allow two special cases $k=n$ and $k=0$.
The case $k=n$ simply means attaching
a new internal edge on a solid vertex of valence $n$
such that the other side of this new internal edge
is a new hollow vertex of valence one,
and then taking summation over all possible locations
(w.r.t. the cyclic order) to attach this edge;
and the case $k=0$ means attaching a new internal edge
on a solid vertex of valence $n$
and changing this vertex into a hollow one
such that the other side of this new internal edge is
a new solid vertex of valence one,
and then taking summation over all possible locations.

\end{Remark}

\begin{Example}
\label{eg-fat-vertexsplit}
We give some examples to make the above definition clear.
Instead of drawing graphs,
we just show how these operators act on a single vertex
together with half-edges attached to it:
\begin{flalign*}
\begin{split}
&\begin{tikzpicture}[scale=1.125]
\node [align=center,align=center] at (3.5,0) {$\cS_{5,2}\bigg($};
\draw [fill] (5,0) circle [radius=0.055];
\draw [thick] (0+5,0) -- (-0.5+5,0.5);
\draw [thick] (0+5,0) -- (-0.5+5,-0.5);
\draw [thick] (5,0) -- (5.5,0.5);
\draw [thick] (5,0) -- (5.5,-0.5);
\draw [thick] (5,0) -- (5.5,0);
\draw [->,>=stealth] (-0.2+5,-0.1) .. controls (0.3+5,-0.5) and (0.3+5,0.5) ..  (-0.2+5,0.1);
\node [align=center,left] at (-0.45+5,0.5) {$h_1$};
\node [align=center,left] at (-0.45+5,-0.5) {$h_2$};
\node [align=center,right] at (5.45,0.5) {$h_5$};
\node [align=center,right] at (5.45,0) {$h_4$};
\node [align=center,right] at (5.45,-0.5) {$h_3$};
\node [align=center,align=center] at (6.45,0) {$\bigg)=$};
\draw [fill] (0+7+1.1,0) circle [radius=0.055];
\draw (1.5+6+1.3,0) circle [radius=0.07];
\draw [thick] (0+7+1.1,0) -- (-0.5+7+1.1,0.5);
\draw [thick] (0+7+1.1,0) -- (-0.5+7+1.1,-0.5);
\draw [thick] (1.5+6+1.35,0.05) -- (2+6+1.3,0.5);
\draw [thick] (1.5+6+1.35,-0.05) -- (2+6+1.3,-0.5);
\draw [thick] (1.5+6+1.37,0) -- (2+6+1.3,0);
\draw [->,>=stealth] (-0.2+7+1.1,-0.1) .. controls (0.3+7+1.1,-0.5) and (0.3+7+1.1,0.5) ..  (-0.2+7+1.1,0.1);
\node [align=center,left] at (-0.45+7+1.1,0.5) {$h_1$};
\node [align=center,left] at (-0.45+7+1.1,-0.5) {$h_2$};
\node [align=center,right] at (2+6+1.25,0) {$h_4$};
\node [align=center,right] at (2+6+1.25,0.5) {$h_5$};
\node [align=center,right] at (2+6+1.25,-0.5) {$h_3$};
\draw [->,>=stealth] (1.3+6+1.3,-0.1) .. controls (1.8+6+1.3,-0.5) and (1.8+6+1.3,0.5) ..  (1.3+6+1.3,0.1);
\draw [thick] (0+7+1.1,0) -- (1.5+6+1.23,0);
\node [align=center,align=center] at (10.2,0) {$+$};
\draw [fill] (0+7+1.1+3.5,0) circle [radius=0.055];
\draw (1.5+6+1.3+3.5,0) circle [radius=0.07];
\draw [thick] (0+7+1.1+3.5,0) -- (-0.5+7+1.1+3.5,0.5);
\draw [thick] (0+7+1.1+3.5,0) -- (-0.5+7+1.1+3.5,-0.5);
\draw [thick] (1.5+6+1.35+3.5,0.05) -- (2+6+1.3+3.5,0.5);
\draw [thick] (1.5+6+1.35+3.5,-0.05) -- (2+6+1.3+3.5,-0.5);
\draw [thick] (1.5+6+1.37+3.5,0) -- (2+6+1.3+3.5,0);
\draw [->,>=stealth] (-0.2+7+1.1+3.5,-0.1) .. controls (0.3+7+1.1+3.5,-0.5) and (0.3+7+1.1+3.5,0.5) ..  (-0.2+7+1.1+3.5,0.1);
\node [align=center,left] at (-0.45+7+1.1+3.5,0.5) {$h_2$};
\node [align=center,left] at (-0.45+7+1.1+3.5,-0.5) {$h_3$};
\node [align=center,right] at (2+6+1.25+3.5,0) {$h_5$};
\node [align=center,right] at (2+6+1.25+3.5,0.5) {$h_1$};
\node [align=center,right] at (2+6+1.25+3.5,-0.5) {$h_4$};
\draw [->,>=stealth] (1.3+6+1.3+3.5,-0.1) .. controls (1.8+6+1.3+3.5,-0.5) and (1.8+6+1.3+3.5,0.5) ..  (1.3+6+1.3+3.5,0.1);
\draw [thick] (0+7+1.1+3.5,0) -- (1.5+6+1.23+3.5,0);
\end{tikzpicture}
\\
&\quad
\begin{tikzpicture}[scale=1.125]
\node [align=center,align=center] at (6.7,0) {$+$};
\draw [fill] (0+7+1.1,0) circle [radius=0.055];
\draw (1.5+6+1.3,0) circle [radius=0.07];
\draw [thick] (0+7+1.1,0) -- (-0.5+7+1.1,0.5);
\draw [thick] (0+7+1.1,0) -- (-0.5+7+1.1,-0.5);
\draw [thick] (1.5+6+1.35,0.05) -- (2+6+1.3,0.5);
\draw [thick] (1.5+6+1.35,-0.05) -- (2+6+1.3,-0.5);
\draw [thick] (1.5+6+1.37,0) -- (2+6+1.3,0);
\draw [->,>=stealth] (-0.2+7+1.1,-0.1) .. controls (0.3+7+1.1,-0.5) and (0.3+7+1.1,0.5) ..  (-0.2+7+1.1,0.1);
\node [align=center,left] at (-0.45+7+1.1,0.5) {$h_3$};
\node [align=center,left] at (-0.45+7+1.1,-0.5) {$h_4$};
\node [align=center,right] at (2+6+1.25,0) {$h_1$};
\node [align=center,right] at (2+6+1.25,0.5) {$h_2$};
\node [align=center,right] at (2+6+1.25,-0.5) {$h_5$};
\draw [->,>=stealth] (1.3+6+1.3,-0.1) .. controls (1.8+6+1.3,-0.5) and (1.8+6+1.3,0.5) ..  (1.3+6+1.3,0.1);
\draw [thick] (0+7+1.1,0) -- (1.5+6+1.23,0);
\node [align=center,align=center] at (10.2,0) {$+$};
\draw [fill] (0+7+1.1+3.5,0) circle [radius=0.055];
\draw (1.5+6+1.3+3.5,0) circle [radius=0.07];
\draw [thick] (0+7+1.1+3.5,0) -- (-0.5+7+1.1+3.5,0.5);
\draw [thick] (0+7+1.1+3.5,0) -- (-0.5+7+1.1+3.5,-0.5);
\draw [thick] (1.5+6+1.35+3.5,0.05) -- (2+6+1.3+3.5,0.5);
\draw [thick] (1.5+6+1.35+3.5,-0.05) -- (2+6+1.3+3.5,-0.5);
\draw [thick] (1.5+6+1.37+3.5,0) -- (2+6+1.3+3.5,0);
\draw [->,>=stealth] (-0.2+7+1.1+3.5,-0.1) .. controls (0.3+7+1.1+3.5,-0.5) and (0.3+7+1.1+3.5,0.5) ..  (-0.2+7+1.1+3.5,0.1);
\node [align=center,left] at (-0.45+7+1.1+3.5,0.5) {$h_4$};
\node [align=center,left] at (-0.45+7+1.1+3.5,-0.5) {$h_5$};
\node [align=center,right] at (2+6+1.25+3.5,0) {$h_2$};
\node [align=center,right] at (2+6+1.25+3.5,0.5) {$h_3$};
\node [align=center,right] at (2+6+1.25+3.5,-0.5) {$h_1$};
\draw [->,>=stealth] (1.3+6+1.3+3.5,-0.1) .. controls (1.8+6+1.3+3.5,-0.5) and (1.8+6+1.3+3.5,0.5) ..  (1.3+6+1.3+3.5,0.1);
\draw [thick] (0+7+1.1+3.5,0) -- (1.5+6+1.23+3.5,0);
\node [align=center,align=center] at (10.2+3.5,0) {$+$};
\draw [fill] (0+7+1.1+3.5+3.5,0) circle [radius=0.055];
\draw (1.5+6+1.3+3.5+3.5,0) circle [radius=0.07];
\draw [thick] (0+7+1.1+3.5+3.5,0) -- (-0.5+7+1.1+3.5+3.5,0.5);
\draw [thick] (0+7+1.1+3.5+3.5,0) -- (-0.5+7+1.1+3.5+3.5,-0.5);
\draw [thick] (1.5+6+1.35+3.5+3.5,0.05) -- (2+6+1.3+3.5+3.5,0.5);
\draw [thick] (1.5+6+1.35+3.5+3.5,-0.05) -- (2+6+1.3+3.5+3.5,-0.5);
\draw [thick] (1.5+6+1.37+3.5+3.5,0) -- (2+6+1.3+3.5+3.5,0);
\draw [->,>=stealth] (-0.2+7+1.1+3.5+3.5,-0.1) .. controls (0.3+7+1.1+3.5+3.5,-0.5) and (0.3+7+1.1+3.5+3.5,0.5) ..  (-0.2+7+1.1+3.5+3.5,0.1);
\node [align=center,left] at (-0.45+7+1.1+3.5+3.5,0.5) {$h_5$};
\node [align=center,left] at (-0.45+7+1.1+3.5+3.5,-0.5) {$h_1$};
\node [align=center,right] at (2+6+1.25+3.5+3.5,0) {$h_3$};
\node [align=center,right] at (2+6+1.25+3.5+3.5,0.5) {$h_4$};
\node [align=center,right] at (2+6+1.25+3.5+3.5,-0.5) {$h_2$};
\draw [->,>=stealth] (1.3+6+1.3+3.5+3.5,-0.1) .. controls (1.8+6+1.3+3.5+3.5,-0.5) and (1.8+6+1.3+3.5+3.5,0.5) ..  (1.3+6+1.3+3.5+3.5,0.1);
\draw [thick] (0+7+1.1+3.5+3.5,0) -- (1.5+6+1.23+3.5+3.5,0);
\node [align=center,align=center] at (16.9,-0.3) {$,$};
\end{tikzpicture}
\end{split}&&
\end{flalign*}
\begin{flalign*}
\begin{split}
&\begin{tikzpicture}[scale=1.075]
\node [align=center,align=center] at (0.25,0) {$\cS_{3,3}\bigg($};
\draw [fill] (2-0.3+0.15,0) circle [radius=0.06];
\draw [thick] (0+2-0.3+0.15,0) -- (-0.5+2-0.3+0.15,0.5);
\draw [thick] (0+2-0.3+0.15,0) -- (-0.5+2-0.3+0.15,-0.5);
\draw [thick] (1.7+0.15,0) -- (2.35+0.15,0);
\node [align=center,left] at (-0.5+2-0.25+0.15,0.5) {$h_1$};
\node [align=center,left] at (-0.5+2-0.25+0.15,-0.5) {$h_2$};
\node [align=center,right] at (2.3+0.15,0) {$h_3$};
\draw [->,>=stealth] (-0.2+1.7+0.15,-0.1) .. controls (0.3+1.7+0.15,-0.5) and (0.3+1.7+0.15,0.5) ..  (-0.2+1.7+0.15,0.1);
\node [align=center,align=center] at (3.35,0) {$\bigg)=$};
\draw [fill] (5,0) circle [radius=0.06];
\draw (5.5,0.5) circle [radius=0.07];
\draw [thick] (0+5,0) -- (-0.5+5,0.5);
\draw [thick] (0+5,0) -- (-0.5+5,-0.5);
\draw [thick] (5,0) -- (5.45,0.45);
\draw [thick] (5,0) -- (5.5,-0.5);
\draw [->,>=stealth] (-0.2+5,-0.1) .. controls (0.3+5,-0.5) and (0.3+5,0.5) ..  (-0.2+5,0.1);
\node [align=center,left] at (-0.45+5,0.5) {$h_1$};
\node [align=center,left] at (-0.45+5,-0.5) {$h_2$};
\node [align=center,right] at (5.45,-0.5) {$h_3$};
\node [align=center,align=center] at (6.2,0) {$+$};
\draw [fill] (5+2.6,0) circle [radius=0.06];
\draw (5.5+2.6,0.5) circle [radius=0.07];
\draw [thick] (0+5+2.6,0) -- (-0.5+5+2.6,0.5);
\draw [thick] (0+5+2.6,0) -- (-0.5+5+2.6,-0.5);
\draw [thick] (5+2.6,0) -- (5.45+2.6,0.45);
\draw [thick] (5+2.6,0) -- (5.5+2.6,-0.5);
\draw [->,>=stealth] (-0.2+5+2.6,-0.1) .. controls (0.3+5+2.6,-0.5) and (0.3+5+2.6,0.5) ..  (-0.2+5+2.6,0.1);
\node [align=center,left] at (-0.45+5+2.6,0.5) {$h_2$};
\node [align=center,left] at (-0.45+5+2.6,-0.5) {$h_3$};
\node [align=center,right] at (5.45+2.6,-0.5) {$h_1$};
\node [align=center,align=center] at (6.2+2.6,0) {$+$};
\draw [fill] (5+2.6+2.6,0) circle [radius=0.06];
\draw (5.5+2.6+2.6,0.5) circle [radius=0.07];
\draw [thick] (0+5+2.6+2.6,0) -- (-0.5+5+2.6+2.6,0.5);
\draw [thick] (0+5+2.6+2.6,0) -- (-0.5+5+2.6+2.6,-0.5);
\draw [thick] (5+2.6+2.6,0) -- (5.45+2.6+2.6,0.45);
\draw [thick] (5+2.6+2.6,0) -- (5.5+2.6+2.6,-0.5);
\draw [->,>=stealth] (-0.2+5+2.6+2.6,-0.1) .. controls (0.3+5+2.6+2.6,-0.5) and (0.3+5+2.6+2.6,0.5) ..  (-0.2+5+2.6+2.6,0.1);
\node [align=center,left] at (-0.45+5+2.6+2.6,0.5) {$h_3$};
\node [align=center,left] at (-0.45+5+2.6+2.6,-0.5) {$h_1$};
\node [align=center,right] at (5.45+2.6+2.6,-0.5) {$h_2$};
\node [align=center,align=center] at (11.25,-0.3) {$,$};
\end{tikzpicture}\\
\end{split}&&
\end{flalign*}
\begin{flalign*}
\begin{split}
&\begin{tikzpicture}[scale=1.075]
\node [align=center,align=center] at (0.25,0) {$\cS_{3,0}\bigg($};
\draw [fill] (2-0.3+0.15,0) circle [radius=0.06];
\draw [thick] (0+2-0.3+0.15,0) -- (-0.5+2-0.3+0.15,0.5);
\draw [thick] (0+2-0.3+0.15,0) -- (-0.5+2-0.3+0.15,-0.5);
\draw [thick] (1.7+0.15,0) -- (2.35+0.15,0);
\node [align=center,left] at (-0.5+2-0.25+0.15,0.5) {$h_1$};
\node [align=center,left] at (-0.5+2-0.25+0.15,-0.5) {$h_2$};
\node [align=center,right] at (2.3+0.15,0) {$h_3$};
\draw [->,>=stealth] (-0.2+1.7+0.15,-0.1) .. controls (0.3+1.7+0.15,-0.5) and (0.3+1.7+0.15,0.5) ..  (-0.2+1.7+0.15,0.1);
\node [align=center,align=center] at (3.35,0) {$\bigg)=$};
\draw (5,0) circle [radius=0.07];
\draw [fill] (5.5,0.5) circle [radius=0.065];
\draw [thick] (0+5-0.05,0.05) -- (-0.5+5,0.5);
\draw [thick] (0+5-0.05,-0.05) -- (-0.5+5,-0.5);
\draw [thick] (5+0.05,0.05) -- (5.5,0.5);
\draw [thick] (5+0.05,-0.05) -- (5.5,-0.5);
\draw [->,>=stealth] (-0.2+5,-0.1) .. controls (0.3+5,-0.5) and (0.3+5,0.5) ..  (-0.2+5,0.1);
\node [align=center,left] at (-0.45+5,0.5) {$h_1$};
\node [align=center,left] at (-0.45+5,-0.5) {$h_2$};
\node [align=center,right] at (5.45,-0.5) {$h_3$};
\node [align=center,align=center] at (6.2,0) {$+$};
\draw (5+2.6,0) circle [radius=0.07];
\draw [fill] (5.5+2.6,0.5) circle [radius=0.065];
\draw [thick] (0+5+2.6-0.05,0.05) -- (-0.5+5+2.6,0.5);
\draw [thick] (0+5+2.6-0.05,-0.05) -- (-0.5+5+2.6,-0.5);
\draw [thick] (5+2.6+0.05,0.05) -- (5.5+2.6,0.5);
\draw [thick] (5+2.6+0.05,-0.05) -- (5.5+2.6,-0.5);
\draw [->,>=stealth] (-0.2+5+2.6,-0.1) .. controls (0.3+5+2.6,-0.5) and (0.3+5+2.6,0.5) ..  (-0.2+5+2.6,0.1);
\node [align=center,left] at (-0.45+5+2.6,0.5) {$h_2$};
\node [align=center,left] at (-0.45+5+2.6,-0.5) {$h_3$};
\node [align=center,right] at (5.45+2.6,-0.5) {$h_1$};
\node [align=center,align=center] at (6.2+2.6,0) {$+$};
\draw (5+2.6+2.6,0) circle [radius=0.07];
\draw [fill] (5.5+2.6+2.6,0.5) circle [radius=0.065];
\draw [thick] (0+5+2.6+2.6-0.05,0.05) -- (-0.5+5+2.6+2.6,0.5);
\draw [thick] (0+5+2.6+2.6-0.05,-0.05) -- (-0.5+5+2.6+2.6,-0.5);
\draw [thick] (5+2.6+2.6+0.05,0.05) -- (5.5+2.6+2.6,0.5);
\draw [thick] (5+2.6+2.6+0.05,-0.05) -- (5.5+2.6+2.6,-0.5);
\draw [->,>=stealth] (-0.2+5+2.6+2.6,-0.1) .. controls (0.3+5+2.6+2.6,-0.5) and (0.3+5+2.6+2.6,0.5) ..  (-0.2+5+2.6+2.6,0.1);
\node [align=center,left] at (-0.45+5+2.6+2.6,0.5) {$h_3$};
\node [align=center,left] at (-0.45+5+2.6+2.6,-0.5) {$h_1$};
\node [align=center,right] at (5.45+2.6+2.6,-0.5) {$h_2$};
\node [align=center,align=center] at (11.25,-0.3) {$.$};
\end{tikzpicture}\\
\end{split}&&
\end{flalign*}

\end{Example}

\begin{Remark}
The vertex-splitting operators are `inverses' of the edge-contraction procedures
in the following sense.
For every graph (with a new hollow vertex) appearing on the right-hand side
of the above equalities,
one can regard the new hollow vertex as the special vertex $v_1$
discussed in \S \ref{sec-fatedge-contr},
then the original graph can be obtained from this new graph by contracting along
a half-edge attached to this new hollow vertex.
\end{Remark}

Finally,
let us define a family of operators $\cJ_{m,n}$,
where $m,n\in\bZ_{>0}$ are a pair of positive integers.
Similar to the vertex-splitting operators $\cS_{n,k}$ defined above,
the operators $\cJ_{m,n}$ are also `inverses' of edge-contraction operators.
The operators $\cS_{n,k}$ inverse the contraction of an internal edge
which is not a loop,
while $\cJ_{m,n}$ inverse the contraction of a loop
(see Definition \ref{def-abs-opr}).

\begin{Definition}
\label{def-abs-opr-J}
Given $m,n\in\bZ_{\geq 0}$,
the operator $\cJ_{m,n}$ is defined to be the linear operator on $\cV^{ho}$
such that
\be
\cJ_{m,n}(\Gamma)=
\sum_{\substack{v_1\in V_m(\Gamma)\\v_2\in V_n(\Gamma)\\ v_1\not= v_2}}
\sum_{\substack{h_1\in H(v_1)\\h_2\in H(v_2)}}
\cJ(\Gamma,h_1,h_2)
\ee
for every graph $\Gamma$,
where the graph $\cJ(\Gamma,h_1,h_2)$ is defined to be the following graph:

\begin{itemize}
\item[1)]
First, we attach a new half-edge $h'$ at $v_1$ such that
$h'$ is right before $h_1$ w.r.t. the cyclic order;
and attach a new half-edge $h''$ at $v_2$ such that
$h''$ is right before $h_2$ w.r.t. the cyclic order.
\item[2)]
Then we merge $v_1$ and $v_2$ together to obtain a new hollow vertex
of valence $m+n+2$,
such that the cyclic orders are defined as follows:
\begin{equation*}
h'\quad\to\quad\text{half-edges on $v_1$}\quad\to\quad
h''\quad\to\quad\text{half-edges on $v_2$}.
\end{equation*}
\item[3)]
Finally,
we glue the half-edges $h'$ and $h''$ together to obtain a loop.
\end{itemize}

\end{Definition}

\begin{Example}
Let us give some examples of the operator $\cJ_{m,n}$.
Similar to Example \ref{eg-fat-vertexsplit},
we only draw two vertices together with half-edges attached to them
instead of drawing the entire graphs:
\begin{flalign*}
\begin{split}
&\begin{tikzpicture}[scale=1.02]
\node [align=center,align=center] at (0.25,0) {$\cJ_{3,2}\bigg($};
\draw [fill] (2-0.3+0.15,0) circle [radius=0.06];
\draw [thick] (0+2-0.3+0.15,0) -- (-0.5+2-0.3+0.15,0.5);
\draw [thick] (0+2-0.3+0.15,0) -- (-0.5+2-0.3+0.15,-0.5);
\draw [thick] (1.7+0.15,0) -- (2.35+0.15,0);
\node [align=center,left] at (-0.5+2-0.25+0.15,0.5) {$h_1$};
\node [align=center,left] at (-0.5+2-0.25+0.15,-0.5) {$h_2$};
\node [align=center,right] at (2.3+0.15,0) {$h_3$};
\draw [->,>=stealth] (-0.2+1.7+0.15,-0.1) .. controls (0.3+1.7+0.15,-0.5) and (0.3+1.7+0.15,0.5) ..  (-0.2+1.7+0.15,0.1);
\node [align=center,align=center] at (4.6,0) {$\bigg)=$};
\draw [fill] (3.6,0) circle [radius=0.06];
\draw [thick] (3.6,-0.55) -- (3.6,0.55);
\node [align=center,right] at (3.6,0.5) {$h_4$};
\node [align=center,right] at (3.6,-0.5) {$h_5$};
\draw [->,>=stealth] (-0.2+1.7+0.15+1.75,-0.1) .. controls (0.3+1.7+0.15+1.75,-0.5) and (0.3+1.7+0.15+1.75,0.5) ..  (-0.2+1.7+0.15+1.75,0.1);
\draw (0+7-0.8,0) circle [radius=0.07];
\draw [thick] (7.4,0) arc (0:174.5:0.6);
\draw [thick] (7.4,0) arc (360:185.5:0.6);
\draw [thick] (-0.5+7-0.8,0) -- (0+7-0.07-0.8,0);
\draw [thick] (-0.5+7-0.8,0.6) -- (0+7-0.05-0.8,0.05);
\draw [thick] (-0.5+7-0.8,-0.6) -- (0+7-0.05-0.8,-0.05);
\draw [thick] (0.05+7-0.8,0.05) -- (0.4+7-0.8,0.3);
\draw [thick] (0.05+7-0.8,-0.05) -- (0.4+7-0.8,-0.3);
\node [align=center,left] at (-0.45+7-0.8,0) {$h_2$};
\node [align=center,left] at (-0.45+7-0.8,0.6) {$h_1$};
\node [align=center,left] at (-0.45+7-0.8,-0.6) {$h_3$};
\node [align=center,right] at (0.35+7-0.8,0.3) {$h_4$};
\node [align=center,right] at (0.35+7-0.8,-0.3) {$h_5$};
\draw [->,>=stealth] (-0.2+7-0.8,-0.1) .. controls (0.3+7-0.8,-0.5) and (0.3+7-0.8,0.5) ..  (-0.2+7-0.8,0.1);
\node [align=center,align=center] at (7.9,0) {$+$};
\draw (0+7-0.8+3,0) circle [radius=0.07];
\draw [thick] (7.4+3,0) arc (0:174.5:0.6);
\draw [thick] (7.4+3,0) arc (360:185.5:0.6);
\draw [thick] (-0.5+7-0.8+3,0) -- (0+7-0.07-0.8+3,0);
\draw [thick] (-0.5+7-0.8+3,0.6) -- (0+7-0.05-0.8+3,0.05);
\draw [thick] (-0.5+7-0.8+3,-0.6) -- (0+7-0.05-0.8+3,-0.05);
\draw [thick] (0.05+7-0.8+3,0.05) -- (0.4+7-0.8+3,0.3);
\draw [thick] (0.05+7-0.8+3,-0.05) -- (0.4+7-0.8+3,-0.3);
\node [align=center,left] at (-0.45+7-0.8+3,0) {$h_2$};
\node [align=center,left] at (-0.45+7-0.8+3,0.6) {$h_1$};
\node [align=center,left] at (-0.45+7-0.8+3,-0.6) {$h_3$};
\node [align=center,right] at (0.35+7-0.8+3,0.3) {$h_5$};
\node [align=center,right] at (0.35+7-0.8+3,-0.3) {$h_4$};
\draw [->,>=stealth] (-0.2+7-0.8+3,-0.1) .. controls (0.3+7-0.8+3,-0.5) and (0.3+7-0.8+3,0.5) ..  (-0.2+7-0.8+3,0.1);
\end{tikzpicture}\\
&\quad
\begin{tikzpicture}[scale=1.02]
\node [align=center,align=center] at (7.9-3,0) {$+$};
\draw (0+7-0.8,0) circle [radius=0.07];
\draw [thick] (7.4,0) arc (0:174.5:0.6);
\draw [thick] (7.4,0) arc (360:185.5:0.6);
\draw [thick] (-0.5+7-0.8,0) -- (0+7-0.07-0.8,0);
\draw [thick] (-0.5+7-0.8,0.6) -- (0+7-0.05-0.8,0.05);
\draw [thick] (-0.5+7-0.8,-0.6) -- (0+7-0.05-0.8,-0.05);
\draw [thick] (0.05+7-0.8,0.05) -- (0.4+7-0.8,0.3);
\draw [thick] (0.05+7-0.8,-0.05) -- (0.4+7-0.8,-0.3);
\node [align=center,left] at (-0.45+7-0.8,0) {$h_3$};
\node [align=center,left] at (-0.45+7-0.8,0.6) {$h_2$};
\node [align=center,left] at (-0.45+7-0.8,-0.6) {$h_1$};
\node [align=center,right] at (0.35+7-0.8,0.3) {$h_4$};
\node [align=center,right] at (0.35+7-0.8,-0.3) {$h_5$};
\draw [->,>=stealth] (-0.2+7-0.8,-0.1) .. controls (0.3+7-0.8,-0.5) and (0.3+7-0.8,0.5) ..  (-0.2+7-0.8,0.1);
\node [align=center,align=center] at (7.9,0) {$+$};
\draw (0+7-0.8+3,0) circle [radius=0.07];
\draw [thick] (7.4+3,0) arc (0:174.5:0.6);
\draw [thick] (7.4+3,0) arc (360:185.5:0.6);
\draw [thick] (-0.5+7-0.8+3,0) -- (0+7-0.07-0.8+3,0);
\draw [thick] (-0.5+7-0.8+3,0.6) -- (0+7-0.05-0.8+3,0.05);
\draw [thick] (-0.5+7-0.8+3,-0.6) -- (0+7-0.05-0.8+3,-0.05);
\draw [thick] (0.05+7-0.8+3,0.05) -- (0.4+7-0.8+3,0.3);
\draw [thick] (0.05+7-0.8+3,-0.05) -- (0.4+7-0.8+3,-0.3);
\node [align=center,left] at (-0.45+7-0.8+3,0) {$h_3$};
\node [align=center,left] at (-0.45+7-0.8+3,0.6) {$h_2$};
\node [align=center,left] at (-0.45+7-0.8+3,-0.6) {$h_1$};
\node [align=center,right] at (0.35+7-0.8+3,0.3) {$h_5$};
\node [align=center,right] at (0.35+7-0.8+3,-0.3) {$h_4$};
\draw [->,>=stealth] (-0.2+7-0.8+3,-0.1) .. controls (0.3+7-0.8+3,-0.5) and (0.3+7-0.8+3,0.5) ..  (-0.2+7-0.8+3,0.1);
\node [align=center,align=center] at (7.9-3+6,0) {$+$};
\draw (0+7-0.8+6,0) circle [radius=0.07];
\draw [thick] (7.4+6,0) arc (0:174.5:0.6);
\draw [thick] (7.4+6,0) arc (360:185.5:0.6);
\draw [thick] (-0.5+7-0.8+6,0) -- (0+7-0.07-0.8+6,0);
\draw [thick] (-0.5+7-0.8+6,0.6) -- (0+7-0.05-0.8+6,0.05);
\draw [thick] (-0.5+7-0.8+6,-0.6) -- (0+7-0.05-0.8+6,-0.05);
\draw [thick] (0.05+7-0.8+6,0.05) -- (0.4+7-0.8+6,0.3);
\draw [thick] (0.05+7-0.8+6,-0.05) -- (0.4+7-0.8+6,-0.3);
\node [align=center,left] at (-0.45+7-0.8+6,0) {$h_1$};
\node [align=center,left] at (-0.45+7-0.8+6,0.6) {$h_3$};
\node [align=center,left] at (-0.45+7-0.8+6,-0.6) {$h_2$};
\node [align=center,right] at (0.35+7-0.8+6,0.3) {$h_4$};
\node [align=center,right] at (0.35+7-0.8+6,-0.3) {$h_5$};
\draw [->,>=stealth] (-0.2+7-0.8+6,-0.1) .. controls (0.3+7-0.8+6,-0.5) and (0.3+7-0.8+6,0.5) ..  (-0.2+7-0.8+6,0.1);
\node [align=center,align=center] at (7.9+6,0) {$+$};
\draw (0+7-0.8+3+6,0) circle [radius=0.07];
\draw [thick] (7.4+3+6,0) arc (0:174.5:0.6);
\draw [thick] (7.4+3+6,0) arc (360:185.5:0.6);
\draw [thick] (-0.5+7-0.8+3+6,0) -- (0+7-0.07-0.8+3+6,0);
\draw [thick] (-0.5+7-0.8+3+6,0.6) -- (0+7-0.05-0.8+3+6,0.05);
\draw [thick] (-0.5+7-0.8+3+6,-0.6) -- (0+7-0.05-0.8+3+6,-0.05);
\draw [thick] (0.05+7-0.8+3+6,0.05) -- (0.4+7-0.8+3+6,0.3);
\draw [thick] (0.05+7-0.8+3+6,-0.05) -- (0.4+7-0.8+3+6,-0.3);
\node [align=center,left] at (-0.45+7-0.8+3+6,0) {$h_1$};
\node [align=center,left] at (-0.45+7-0.8+3+6,0.6) {$h_3$};
\node [align=center,left] at (-0.45+7-0.8+3+6,-0.6) {$h_2$};
\node [align=center,right] at (0.35+7-0.8+3+6,0.3) {$h_5$};
\node [align=center,right] at (0.35+7-0.8+3+6,-0.3) {$h_4$};
\draw [->,>=stealth] (-0.2+7-0.8+3+6,-0.1) .. controls (0.3+7-0.8+3+6,-0.5) and (0.3+7-0.8+3+6,0.5) ..  (-0.2+7-0.8+3+6,0.1);
\node [align=center,align=center] at (16.65,-0.35) {$.$};
\end{tikzpicture}\\
\end{split}&&
\end{flalign*}
Notice that in the above definition we also allow the two special cases $m=0$ or $n=0$.
For example:
\begin{flalign*}
\begin{tikzpicture}[scale=1.475]
\node [align=center,align=center] at (2.745,0) {$\cJ_{2,0}\bigg($};
\draw [fill] (3.6,0) circle [radius=0.05];
\draw [thick] (3.6,-0.55) -- (3.6,0.55);
\node [align=center,right] at (3.6,0.5) {$h_1$};
\node [align=center,right] at (3.6,-0.5) {$h_2$};
\draw [->,>=stealth] (-0.2+1.7+0.15+1.75,-0.1) .. controls (0.3+1.7+0.15+1.75,-0.5) and (0.3+1.7+0.15+1.75,0.5) ..  (-0.2+1.7+0.15+1.75,0.1);
\draw [fill] (4.75,0) circle [radius=0.05];
\node [align=center,align=center] at (5.35,0) {$\bigg)=$};
\draw (7-0.1,0) circle [radius=0.05];
\draw [thick] (6.45-0.1,0.35) -- (6.97-0.1,0.03);
\draw [thick] (6.45-0.1,-0.35) -- (6.97-0.1,-0.03);
\node [align=center,left] at (6.5-0.1,0.35) {$h_1$};
\node [align=center,left] at (6.5-0.1,-0.35) {$h_2$};
\draw [->,>=stealth] (-0.2+1.7+0.15+1.75+3.4-0.1,-0.1) .. controls (0.3+1.7+0.15+1.75+3.4-0.1,-0.5) and (0.3+1.7+0.15+1.75+3.4-0.1,0.5) ..  (-0.2+1.7+0.15+1.75+3.4-0.1,0.1);
\draw [thick] (7.8-0.1,0) arc (0:173:0.4);
\draw [thick] (7.8-0.1,0) arc (360:188:0.4);
\node [align=center,align=center] at (8.25-0.1,0) {$+$};
\draw (7+2.8-0.2,0) circle [radius=0.05];
\draw [thick] (6.45+2.8-0.2,0.35) -- (6.97+2.8-0.2,0.03);
\draw [thick] (6.45+2.8-0.2,-0.35) -- (6.97+2.8-0.2,-0.03);
\node [align=center,left] at (6.5+2.8-0.2,0.35) {$h_2$};
\node [align=center,left] at (6.5+2.8-0.2,-0.35) {$h_1$};
\draw [->,>=stealth] (-0.2+1.7+0.15+1.75+3.4+2.8-0.2,-0.1) .. controls (0.3+1.7+0.15+1.75+3.4+2.8-0.2,-0.5) and (0.3+1.7+0.15+1.75+3.4+2.8-0.2,0.5) ..  (-0.2+1.7+0.15+1.75+3.4+2.8-0.2,0.1);
\draw [thick] (7.8+2.8-0.2,0) arc (0:173:0.4);
\draw [thick] (7.8+2.8-0.2,0) arc (360:188:0.4);
\node [align=center,align=center] at (10.85-0.25,-0.3) {$.$};
\end{tikzpicture}&&
\end{flalign*}

\end{Example}

\begin{Remark}
In the above equalities,
the original graph can be obtained from  a graph on the right-hand side
by contracting a loop attached to the new hollow vertex.
\end{Remark}

\subsection{Abstract Virasoro operators and abstract Virasoro constraints}
\label{sec-fatabs-virasoro}

In this subsection,
we define a sequence of operators $\{\cL_m\}_{m\geq -1}$ on fat graphs
which annihilates the abstract partition function $ \cZ $.
Moreover,
we show that these operators satisfy $[\cL_m,\cL_n]=(m-n)\cL_{m+n}$
with respect to a Lie bracket $[-,-]$.
In other words,
we construct the `abstract Virasoro constraints' for the abstract QFT
for fat graphs.

\begin{Definition}
\label{def-abs-Viropr}
We define a sequence of operators $\{\cL_m\}_{m\geq -1}$ as follows:
\begin{itemize}
\item[1)]
The operator $\cL_{-1}$ is defined to be:
\ben
\cL_{-1}:=-\pd_1+\sum_{n\geq 1}\cS_{n,n}
+g_s^{-2}\cdot\gamma_{-1},
\een
where $g_s$ is a formal variable
(see \S \ref{sec-abstractpartition}),
and the operator $\gamma_{-1}$ is multiplying by
(i.e., taking disjoint union with) $\Gamma_{-1}$,
where $\Gamma_{-1}$ is the following graph:
\ben
\begin{tikzpicture}[scale=1.2]
\node [align=center,align=center] at (-0.9,0) {$\Gamma_{-1}:=$};
\draw (0,0) circle [radius=0.06];
\draw [thick] (0.06,0) -- (1.5,0);
\draw [fill] (1.5,0) circle [radius=0.06];
\node [align=center,align=center] at (1.8,-0.25) {$.$};
\end{tikzpicture}
\een

\item[2)]
The operator $\cL_{0}$ is defined to be:
\ben
\cL_{0}:=-2\pd_2+\sum_{n\geq 1}\cS_{n,n-1}
+g_s^{-2}\cdot\gamma_0,
\een
where the operator $\gamma_0$ is multiplying by the graph $\Gamma_0$:
\ben
\begin{tikzpicture}[scale=1.2]
\node [align=center,align=center] at (-0.85,0) {$\Gamma_{0}:=$};
\draw (0,0) circle [radius=0.06];
\draw [thick] (0.04,-0.04) .. controls (1.5,-1.2) and (1.5,1.2) ..  (0.04,0.04);
\draw [->,>=stealth] (-0.2,-0.1) .. controls (0.3,-0.5) and (0.3,0.5) ..  (-0.2,0.1);
\node [align=center,align=center] at (1.5,-0.25) {$.$};
\end{tikzpicture}
\een

\item[3)]
The operator $\cL_{1}$ is defined to be:
\ben
\cL_{1}:=-3\pd_3+\sum_{n\geq 1}\cS_{n+1,n-1}
+2\cJ_{1,0}(-\sqcup\Gamma_{dot}),
\een
where $\Gamma_{dot}$ is the graph consisting of one single solid vertex
of valence zero.

\item[4)]
The operator $\cL_{m}$ ($m\geq 2$) is defined to be:
\begin{equation*}
\begin{split}
\cL_{m}:=&-(m+2)\pd_{m+2}+\sum_{n\geq 1}\cS_{n+m,n-1}
+g_s^2\cdot \sum_{n=1}^{m-1}\cJ_{n,m-n}\\
&+2\cJ_{m,0}(-\sqcup\Gamma_{dot}).
\end{split}
\end{equation*}

\end{itemize}
We will call $\{\cL_m\}_{m\geq -1}$ the
`abstract Virasoro operators' for fat graphs.
\end{Definition}

It is clear that the operator $\cL_m$ ($m\geq -1$) will produce a new hollow vertex
of valence $m+2$ when applied to a fat graph.
Our first main theorem in this subsection is the following
`abstract Virasoro constraints':

\begin{Theorem}
\label{thm-abstract-vir1}
We have:
\be
\cL_m  (\cZ)  =0,\qquad \forall m\geq -1,
\ee
where $\cZ$ is the abstract partition function
(see Definition \ref{def-fat-abs-par}).
\end{Theorem}

\begin{Remark}
Notice that the actions of operators $\cL_n$ on $\cZ$ are well-defined
even though $\cZ$ is a formal infinite summation,
since it is easy to see that
each graph can be obtained only in a finite number of ways
by applying operations in $\cL_n$ to graphs appearing in the expression of $\cZ$.

\end{Remark}

\begin{proof}
We only prove the first equation $\cL_{-1} \cZ =0$ here,
since one can prove other equations using exactly the same method.

The equation $\cL_{-1} \cZ =0$ is equivalent to the following:
\be
\pd_1\cF=
\sum_{n\geq 1}\cS_{n,n}\cF+g_s^{-2}\cdot \Gamma_{-1},
\ee
where $\cF=\log(\cZ)$ is the abstract free energy
(see Definition \ref{def-fat-abs-fe})
and the graph $\Gamma_{-1}$ is defined in Definition \ref{def-abs-Viropr}.
Comparing the coefficients of $g_s$ in both sides of the above equation,
we know that it suffices to show the following sequence of equations:
\be
\label{eq-pf-virasoro1}
\pd_1\cF_g=
\sum_{n\geq 1}\cS_{n,n} \cF_g
+\delta_{g,0}\cdot \Gamma_{-1},
\qquad
\forall g\geq 0.
\ee

Now let us recall the abstract quadratic recursion relations
in Theorem \ref{thm-abstract-rec}.
Take $\mu_1=1$ in this theorem,
i.e., we contract a vertex of valence one in a connected fat graph
of type $\mu=(1,\mu_2,\cdots,\mu_n)$,
then this theorem tells us the resulting graphs must be
connected and of type $(\mu_j-1,\mu_{[n]\backslash\{1,j\}})$ for some $j$.
Moreover,
we can decompose the edge-contraction procedure into two steps:
first we change this vertex of valence one into a hollow one,
then we contract the internal edge and merge this hollow vertex
with an adjacent solid vertex into a new solid vertex.
Then the one-to-one correspondence constructed in the proof
of Theorem \ref{thm-abstract-rec} simply gives us the one-to-one correspondence
we need in proving \eqref{eq-pf-virasoro1},
except for the special graph $\Gamma_{-1}$ of genus $0$
which needs to be added additionally.

All equations $\cL_m \cZ =0$ ($m\geq -1$) are proved in the same way.
In other words,
this theorem is just a reformulation of the
quadratic recursion relation \eqref{eq-abstract-rec}
in terms of `generating series' $\cF_g, \cF,  \cZ $
of the abstract correlators $\cF_g^\mu$.
\end{proof}

Now we have already known that the abstract Virasoro operators $\{\cL_m\}_{m\geq -1}$
annihilate the abstract partition function $\cZ$.
In order to justify the name `abstract Virasoro constraints',
we are now supposed to construct the Virasoro commutation relations
among these operators.
In what follows,
we will construct a Lie bracket $[-,-]$,
and show that the operators $\{\cL_{m}\}_{m\geq -1}$
satisfy $[\cL_m,\cL_n]=(m-n)\cL_{m+n}$ for every $m,n\geq -1$.
We will see that the Lie bracket we need is `almost' a commutator--
it is the composition of the commutator and an additional edge-contraction.
Before giving the specific definition of this bracket,
we need the following lemma:

\begin{Lemma}
\label{lemma-adjacent}
Let $\Gamma$ be a fat graph,
and $m,n\geq -1$ be two integers.
Then for every graph $\Gamma'$ appearing in the expression of
$(\cL_m\circ\cL_n-\cL_n\circ\cL_m)(\Gamma)$ with nonzero coefficient,
the two new hollow vertices (of valences $m+2$ and $n+2$ respectively)
are adjacent,
i.e.,
there is an internal edge (which is not a loop) connecting them.
\end{Lemma}

\begin{proof}
The operator $\cL_k$ produces a new hollow vertex of valence $k+2$,
and this hollow vertex must `come from' some solid vertices of the original graph
or an additional connected component in the following sense:
\begin{itemize}
\item[1)]
It may come from a solid vertex by simply changing this vertex
into a hollow one with the same valence
(i.e., come from the operator $\pd_{k+2}$);

\item[2)]
It may come from a solid vertex by slitting this vertex into
a hollow one and a solid one
(i.e., come from the operator $\cS_{n,n-k-1}$ for some $n$);

\item[3)]
It may come from two solid vertices by merging them together and adding a loop,
or from one solid vertex by adding a loop
(i.e., come from the operators $\cJ_{n,k-n}$ for some $n$);

\item[4)]
It may come from an additional connected component
(i.e., come from the operator $\gamma_{-1}$ or $\gamma_0$).

\end{itemize}
The inverses of above procedures are clear:
in the first case,
the inverse procedure is to simply change the hollow vertex into a solid one;
and in the other cases,
the inverse procedures are just the edge-contraction procedures
discussed in \S \ref{sec-absqft}.

Now let $\Gamma'$ be a fat graph with hollow vertices
and assume that $\Gamma'$ appears in
the expression of $\cL_m\cL_n(\Gamma)$ or $\cL_n\cL_m(\Gamma)$.
Then there are two hollow vertices of valences $m+2$ and $n+2$ respectively
in $\Gamma'$,
created by the two operators $\cL_m$ and $\cL_n$ respectively.
Now let us compare the coefficients of $\Gamma'$
in the expressions of $\cL_m\cL_n(\Gamma)$ and $\cL_n\cL_m(\Gamma)$.

If these two hollow vertices in $\Gamma'$ are not adjacent,
then we can recover the graph $\Gamma$ by applying two successive
vertex-changing or edge-contraction procedures
to $\Gamma'$ that inverse two above procedures in the definition of $\cL_m$ and $\cL_n$.
It is easy to see that in these cases
two such vertex-changing or edge-contraction procedures commute,
and their inverses also commute.
In other words,
the actions of $\cL_m$ and $\cL_n$ commute `locally':
their actions commute on these specific solid vertices or components.
Therefore the coefficients of $\Gamma'$ in
$(\cL_m\circ\cL_n)(\Gamma)$ and in $(\cL_n\circ\cL_m)(\Gamma)$
must be equal.

Therefore,
the only possibility for a graph $\Gamma'$ to have a nonzero coefficient
in $(\cL_m\cL_n-\cL_n\cL_m)(\Gamma)$
is the case that
the two vertex-changing or edge-contraction procedures `interplay' with each other,
and so do their inverses.
For example,
one can first apply a vertex splitting operator to produce a new hollow vertex
together with a solid vertex,
and then apply another operator to this new solid vertex to create another new hollow vertex
adjacent to the first one.
Only in such cases can two operators have non-trivial commutator,
thus the two new hollow vertices in $\Gamma'$ must be adjacent
(notice that we cannot contract the internal edge that connects these two hollow vertices).
\end{proof}

Now let us consider a special case of edge-contraction.
The above lemma tells us that the two new hollow vertices
(of valences $m+2$ and $n+2$ respectively)
created by $\cL_m\circ\cL_n-\cL_n\circ\cL_m$
are adjacent in every resulting graph.
Therefore,
we can contract the internal edge connecting them in each graph,
and denote by $K^{ho}$ (here `ho' for `hollow')
such an edge-contraction operator
(the cyclic order on the resulting vertex is given
similarly as in Definition \ref{def-abs-opr}).
Notice that $K^{ho}$ may not make sense when it appears alone,
since there might not be a pair of adjacent hollow vertices in general.
In what follows,
we will always consider compositions of operators in the following form:
\be
K^{ho}\circ\big(\cL_m\circ\cL_n-\cL_n\circ\cL_m\big).
\ee

Now we can describe the structure of Lie algebra
on the space spanned by the operators $\{\cL_{m}\}_{m\geq -1}$.
The Lie bracket we need is just the composition of edge-contraction $K^{ho}$
with the common commutators of these abstract Virasoro operators,
i.e.,
we define the Lie bracket by:
\be
\label{eq-def-Liebracket-Virasoro}
[\cL_m,\cL_n]:=K^{ho}\circ\big(
\cL_m\circ\cL_n-\cL_n\circ\cL_m
\big).
\ee
Then we have the following:

\begin{Theorem}
\label{thm-fat-abs-virasorocomm}
We have:
\be
[\cL_m,\cL_n]=(m-n)\cL_{m+n},\qquad
\forall  m,n\geq -1,
\ee
where $[-,-]$ is defined by \eqref{eq-def-Liebracket-Virasoro}.
This gives a structure of Lie algebra on the space
spanned by the abstract Virasoro operators $\{\cL_m\}_{m\geq -1}$.
\end{Theorem}

\begin{proof}

Assume $\Gamma$ is a fat graph.
Let us check the following relation:
\be
\bigg(K^{ho}\circ\big(\cL_m\circ\cL_n-\cL_n\circ\cL_m\big)\bigg)(\Gamma)
=(m-n)\cL_{m+n}(\Gamma).
\ee

By Lemma \ref{lemma-adjacent},
we only need to consider some special graphs where
the two new hollow vertices are adjacent.
As pointed out in the proof of Lemma \ref{lemma-adjacent},
there are four types of operators that produce new hollow vertices:
$\pd_k$, $\cS_{k,l}$, $\cJ_{k,l}$, and $\gamma_k$.
We can simply analyse all the possible cases of composition of
two of these operators
and compute their commutators,
then compare the results with $(m-n)\cL_{m+n}(\Gamma)$.
Here let us work out the details only for some special cases,
and omit other cases since they are all in the same fashion.

First,
let us consider the composition of two operators of the form $\pd_k$
(acting on two adjacent solid vertices).
It is clear that the actions of $\pd_{m+2}$ and $\pd_{n+2}$ commute
(see Lemma \ref{eq-pd-commute}),
thus this case is trivial.

Now let us consider the composition of an operator $\pd_k$ and an
vertex-splitting operator $\cS_{k,l}$.
It is not hard to find out that only the operators
$-(m+2)\pd_{m+2}\circ\cS_{m+n+2,m+1}$ and $-(n+2)\pd_{n+2}\circ\cS_{m+n+2,n+1}$
can provide us nontrivial commutators
(i.e., graphs with adjacent hollow vertices)
in $(\cL_m\cL_n-\cL_n\cL_m)(\Gamma)$:
the only case we need to consider is that we first split a solid vertex of valence $m+n+2$
into a hollow one and a solid one,
and then change the solid one into a hollow one using $\pd_{m+2}$ or $\pd_{n+2}$.
For example,
$\pd_{m+2}\circ\cS_{m+n+2,m+1}$ gives the following:
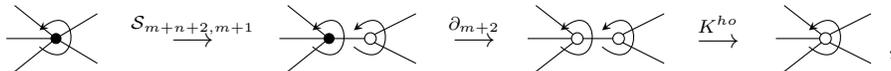
\begin{equation*}
\begin{tikzpicture}[scale=1.1]
\draw [fill] (0+0.2,0) circle [radius=0.06];
\draw [->,>=stealth] (-0.2+0.2,-0.1) .. controls (0.3+0.2,-0.5) and (0.3+0.2,0.5) ..  (-0.2+0.2,0.1);
\draw (-0.6+0.2,0) -- (0+0.2,0);
\draw (-0.5+0.2,0.4) -- (0+0.2,0);
\draw (-0.5+0.2,-0.4) -- (0+0.2,0);
\draw (0.5+0.2,0.3) -- (0+0.2,0);
\draw (0.5+0.2,-0.3) -- (0+0.2,0);
\node [align=center,align=center] at (1.85,0.1) {$\stackrel{\cS_{m+n+2,m+1}}{\longrightarrow}$};
\draw [fill] (0+3.5,0) circle [radius=0.06];
\draw [->,>=stealth] (-0.2+3.5,-0.1) .. controls (0.3+3.5,-0.5) and (0.3+3.5,0.5) ..  (-0.2+3.5,0.1);
\draw (0+4,0) circle [radius=0.07];
\draw [->,>=stealth] (-0.2+4,-0.1) .. controls (0.3+4,-0.5) and (0.3+4,0.5) ..  (-0.2+4,0.1);
\draw (3.5,0) -- (3.93,0);
\draw (-0.6+3.5,0) -- (0+3.5,0);
\draw (-0.5+3.5,0.4) -- (0+3.5,0);
\draw (-0.5+3.5,-0.4) -- (0+3.5,0);
\draw (0.55+4,0.3) -- (0+4.05,0.05);
\draw (0.55+4,-0.3) -- (0+4.05,-0.05);
\node [align=center,align=center] at (5.25,0.1) {$\stackrel{\pd_{m+2}}{\longrightarrow}$};
\draw (0+3.5+3,0) circle [radius=0.07];
\draw [->,>=stealth] (-0.2+3.5+3,-0.1) .. controls (0.3+3.5+3,-0.5) and (0.3+3.5+3,0.5) ..  (-0.2+3.5+3,0.1);
\draw (0+4+3,0) circle [radius=0.07];
\draw [->,>=stealth] (-0.2+4+3,-0.1) .. controls (0.3+4+3,-0.5) and (0.3+4+3,0.5) ..  (-0.2+4+3,0.1);
\draw (3.5+3.07,0) -- (3.93+3,0);
\draw (-0.6+3.5+3,0) -- (0+3.5+3-0.07,0);
\draw (-0.5+3.5+3,0.4) -- (0+3.5+2.95,0.05);
\draw (-0.5+3.5+3,-0.4) -- (0+3.5+2.95,-0.05);
\draw (0.55+4+3,0.3) -- (0+4.05+3,0.05);
\draw (0.55+4+3,-0.3) -- (0+4.05+3,-0.05);
\node [align=center,align=center] at (8.2,0.1) {$\stackrel{K^{ho}}{\longrightarrow}$};
\draw (0+3.5+3+3,0) circle [radius=0.07];
\draw [->,>=stealth] (-0.2+3.5+3+3,-0.1) .. controls (0.3+3.5+3+3,-0.5) and (0.3+3.5+3+3,0.5) ..  (-0.2+3.5+3+3,0.1);
\draw (-0.6+3.5+3+3,0) -- (0+3.5+3-0.07+3,0);
\draw (-0.5+3.5+3+3,0.4) -- (0+3.5+2.95+3,0.05);
\draw (-0.5+3.5+3+3,-0.4) -- (0+3.5+2.95+3,-0.05);
\draw (0.55+4+3+2.5,0.3) -- (0+4.05+3+2.5,0.05);
\draw (0.55+4+3+2.5,-0.3) -- (0+4.05+3+2.5,-0.05);
\node [align=center,align=center] at (10.3,-0.2) {$,$};
\end{tikzpicture}
\end{equation*}
and $\pd_{n+2}\circ\cS_{m+n+2,n+1}$ gives the following:
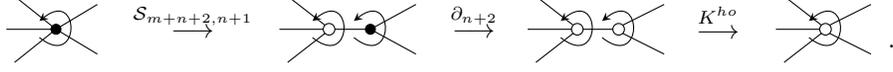
\begin{equation*}
\begin{tikzpicture}[scale=1.1]
\draw [fill] (0+0.2,0) circle [radius=0.06];
\draw [->,>=stealth] (-0.2+0.2,-0.1) .. controls (0.3+0.2,-0.5) and (0.3+0.2,0.5) ..  (-0.2+0.2,0.1);
\draw (-0.6+0.2,0) -- (0+0.2,0);
\draw (-0.5+0.2,0.4) -- (0+0.2,0);
\draw (-0.5+0.2,-0.4) -- (0+0.2,0);
\draw (0.5+0.2,0.3) -- (0+0.2,0);
\draw (0.5+0.2,-0.3) -- (0+0.2,0);
\node [align=center,align=center] at (1.85,0.1) {$\stackrel{\cS_{m+n+2,n+1}}{\longrightarrow}$};
\draw (0+3.5,0) circle [radius=0.07];
\draw [fill] (0+4,0) circle [radius=0.06];
\draw [->,>=stealth] (-0.2+3.5,-0.1) .. controls (0.3+3.5,-0.5) and (0.3+3.5,0.5) ..  (-0.2+3.5,0.1);
\draw [->,>=stealth] (-0.2+4,-0.1) .. controls (0.3+4,-0.5) and (0.3+4,0.5) ..  (-0.2+4,0.1);
\draw (3.57,0) -- (4,0);
\draw (-0.6+3.5,0) -- (0+3.43,0);
\draw (-0.5+3.5,0.4) -- (0+3.45,0.05);
\draw (-0.5+3.5,-0.4) -- (0+3.45,-0.05);
\draw (0.55+4,0.3) -- (0+4,0);
\draw (0.55+4,-0.3) -- (0+4,0);
\node [align=center,align=center] at (5.25,0.1) {$\stackrel{\pd_{n+2}}{\longrightarrow}$};
\draw (0+3.5+3,0) circle [radius=0.07];
\draw (0+4+3,0) circle [radius=0.07];
\draw [->,>=stealth] (-0.2+3.5+3,-0.1) .. controls (0.3+3.5+3,-0.5) and (0.3+3.5+3,0.5) ..  (-0.2+3.5+3,0.1);
\draw [->,>=stealth] (-0.2+4+3,-0.1) .. controls (0.3+4+3,-0.5) and (0.3+4+3,0.5) ..  (-0.2+4+3,0.1);
\draw (3.5+3.07,0) -- (3.93+3,0);
\draw (-0.6+3.5+3,0) -- (0+3.5+3-0.07,0);
\draw (-0.5+3.5+3,0.4) -- (0+3.5+2.95,0.05);
\draw (-0.5+3.5+3,-0.4) -- (0+3.5+2.95,-0.05);
\draw (0.55+4+3,0.3) -- (0+4.05+3,0.05);
\draw (0.55+4+3,-0.3) -- (0+4.05+3,-0.05);
\node [align=center,align=center] at (8.2,0.1) {$\stackrel{K^{ho}}{\longrightarrow}$};
\draw (0+3.5+3+3,0) circle [radius=0.07];
\draw [->,>=stealth] (-0.2+3.5+3+3,-0.1) .. controls (0.3+3.5+3+3,-0.5) and (0.3+3.5+3+3,0.5) ..  (-0.2+3.5+3+3,0.1);
\draw (-0.6+3.5+3+3,0) -- (0+3.5+3-0.07+3,0);
\draw (-0.5+3.5+3+3,0.4) -- (0+3.5+2.95+3,0.05);
\draw (-0.5+3.5+3+3,-0.4) -- (0+3.5+2.95+3,-0.05);
\draw (0.55+4+3+2.5,0.3) -- (0+4.05+3+2.5,0.05);
\draw (0.55+4+3+2.5,-0.3) -- (0+4.05+3+2.5,-0.05);
\node [align=center,align=center] at (10.3,-0.2) {$.$};
\end{tikzpicture}
\end{equation*}
Thus if we apply $K^{ho}$ to the difference of
$-(m+2)\pd_{m+2}\circ\cS_{m+n+2,m+1}(\Gamma)$ and
$-(n+2)\pd_{n+2}\circ\cS_{m+n+2,n+1}(\Gamma)$,
we may simply obtain:
\begin{equation*}
\begin{split}
&\big((-m-2)-(-n-2)\big)(m+n+2)\pd_{m+n+2}(\Gamma)\\
=&(n-m)(m+n+2)\pd_{m+n+2}(\Gamma),\\
\end{split}
\end{equation*}
where the coefficient $(m+n+2)$ is the number of ways
to choose $m$ (or $n$) adjacent half-edges
when applying $\cS_{m+n+2,n+1}$ or $\cS_{m+n+2,m+1}$.
This term is exactly the first term in $(m-n)\cL_{m+n}(\Gamma)$.

Similarly,
one can analyse all possible commutators of the four operators
$\pd_k$, $\cS_{k,l}$, $\cJ_{k,l}$, and $\gamma_k$ using the same method,
and here we omit the details.
Then the conclusion of this theorem will follow from these results.
In particular,
the term $(m-n)\cS_{k+m+n,k-1}$ in the expression of $(m-n)\cL_{m+n}$
comes from the commutator of the two vertex-splitting operators $\cS_{k,l}$
in $\cL_m$ and $\cL_n$.
\end{proof}

\begin{Example}
Let us give an example of the above theorem.
Let $\Gamma$ be a fat graph consisting of two solid vertices
and an internal edge connecting them,
and let us check the following relation:
\ben
[\cL_{-1},\cL_{0}](\Gamma)=-\cL_{-1}(\Gamma).
\een
In this example we omit the cyclic orders at the vertices,
since all the vertices are of valence $\leq 2$ and hence there is no confusion.
It is easy to compute:
\begin{flalign*}
\begin{tikzpicture}[scale=0.95]
\node [align=center,align=center] at (0,0) {$\cL_0\bigg($};
\draw [fill] (0.6,0) circle [radius=0.058];
\draw [fill] (1.2,0) circle [radius=0.058];
\draw (0.6,0) -- (1.2,0);
\node [align=center,align=center] at (1.9,0) {$\bigg)=2$};
\draw [fill] (2.7,0) circle [radius=0.058];
\draw (3.3,0) circle [radius=0.07];
\draw [fill] (3.9,0) circle [radius=0.058];
\draw (2.65,0) -- (3.23,0);
\draw (3.37,0) -- (3.85,0);
\node [align=center,align=center] at (4.8,0) {$+g_s^{-2}\bigg($};
\draw [fill] (5.6,0) circle [radius=0.058];
\draw [fill] (6.2,0) circle [radius=0.058];
\draw (5.6,0) -- (6.2,0);
\draw (6.6,0) circle [radius=0.07];
\draw (6.65,0.05) .. controls (7.3,0.55) and (7.3,-0.55) ..  (6.65,-0.05);
\node [align=center,align=center] at (7.5,0) {$\bigg),$};
\end{tikzpicture}&&
\end{flalign*}
\begin{flalign*}
\begin{tikzpicture}[scale=0.95]
\node [align=center,align=center] at (-0.1,0) {$\cL_{-1}\bigg($};
\draw [fill] (0.6,0) circle [radius=0.058];
\draw [fill] (1.2,0) circle [radius=0.058];
\draw (0.6,0) -- (1.2,0);
\node [align=center,align=center] at (2.05,0) {$\bigg)=-2$};
\draw [fill] (3,0) circle [radius=0.058];
\draw (3.6,0) circle [radius=0.07];
\draw (3,0) -- (3.53,0);
\node [align=center,align=center] at (4.2,0) {$+2$};
\draw [fill] (4.8,0) circle [radius=0.058];
\draw [fill] (5.4,0) circle [radius=0.058];
\draw (6,0) circle [radius=0.07];
\draw (4.8,0) -- (5.93,0);
\node [align=center,align=center] at (6.9,0) {$+g_s^{-2}\bigg($};
\draw [fill] (7.7,0) circle [radius=0.058];
\draw [fill] (8.3,0) circle [radius=0.058];
\draw [fill] (8.7,0) circle [radius=0.058];
\draw (9.3,0) circle [radius=0.07];
\draw (7.7,0) -- (8.3,0);
\draw (8.7,0) -- (9.23,0);
\node [align=center,align=center] at (9.7,0) {$\bigg).$};
\end{tikzpicture}&&
\end{flalign*}
And then:
\begin{flalign*}
\begin{split}
&
\begin{tikzpicture}[scale=0.9]
\node [align=center,align=center] at (-0.5-0.05,0) {$\cL_{-1}\circ\cL_0\bigg($};
\draw [fill] (0.6-0.05,0) circle [radius=0.06];
\draw [fill] (1.2-0.05,0) circle [radius=0.06];
\draw (0.6-0.05,0) -- (1.2-0.05,0);
\node [align=center,align=center] at (2.05-0.05,0) {$\bigg)=-4$};
\draw [fill] (2.7+0.25,0) circle [radius=0.06];
\draw (3.3+0.25,0) circle [radius=0.07];
\draw (3.9+0.25,0) circle [radius=0.07];
\draw (2.65+0.25,0) -- (3.23+0.25,0);
\draw (3.37+0.25,0) -- (3.83+0.25,0);
\node [align=center,align=center] at (5.1,0) {$-2g_s^{-2}\bigg($};
\draw [fill] (5.6+0.35,0) circle [radius=0.06];
\draw (6.2+0.35,0) circle [radius=0.07];
\draw (5.6+0.35,0) -- (6.13+0.35,0);
\draw (6.6+0.35,0) circle [radius=0.07];
\draw (6.65+0.35,0.05) .. controls (7.3+0.35,0.55) and (7.3+0.35,-0.55) ..  (6.65+0.35,-0.05);
\node [align=center,align=center] at (7.5+0.55,0) {$\bigg)+4$};
\draw [fill] (8.8,0) circle [radius=0.06];
\draw (9.4,0) circle [radius=0.07];
\draw [fill] (10,0) circle [radius=0.06];
\draw (10.6,0) circle [radius=0.07];
\draw (8.8,0) -- (9.33,0);
\draw (9.47,0) -- (10.53,0);
\end{tikzpicture}\\
&\quad
\begin{tikzpicture}[scale=0.9]
\node [align=center,align=center] at (4.5+0.05,0) {$+2g_s^{-2}\bigg($};
\draw [fill] (5+0.35+0.05,0) circle [radius=0.06];
\draw [fill] (5.6+0.35+0.05,0) circle [radius=0.06];
\draw (6.2+0.35+0.05,0) circle [radius=0.07];
\draw (5+0.35+0.05,0) -- (6.13+0.35+0.05,0);
\draw (6.6+0.35+0.05,0) circle [radius=0.07];
\draw (6.65+0.35+0.05,0.05) .. controls (7.3+0.35+0.05,0.55) and (7.3+0.35+0.05,-0.55) ..  (6.65+0.35+0.05,-0.05);
\node [align=center,align=center] at (8.55,0) {$\bigg)+2g_s^{-2}\bigg($};
\draw [fill] (9.7,0) circle [radius=0.06];
\draw (10.3,0) circle [radius=0.07];
\draw [fill] (10.9,0) circle [radius=0.06];
\draw [fill] (11.3,0) circle [radius=0.06];
\draw (11.9,0) circle [radius=0.07];
\draw (9.7,0) -- (10.23,0);
\draw (10.37,0) -- (10.9,0);
\draw (11.3,0) -- (11.83,0);
\node [align=center,align=center] at (12.9,0) {$\bigg)+g_s^{-4}\bigg($};
\draw [fill] (4.6+0.35+0.6+9,0) circle [radius=0.06];
\draw [fill] (4+0.35+0.6+9,0) circle [radius=0.07];
\draw (4.35+0.6+9,0) -- (4.88+0.6+9,0);
\draw [fill] (5.6+0.35+9,0) circle [radius=0.06];
\draw (6.2+0.35+9,0) circle [radius=0.07];
\draw (5.6+0.35+9,0) -- (6.13+0.35+9,0);
\draw (6.6+0.35+9,0) circle [radius=0.07];
\draw (6.65+0.35+9,0.05) .. controls (7.3+0.35+9,0.55) and (7.3+0.35+9,-0.55) ..  (6.65+0.35+9,-0.05);
\node [align=center,align=center] at (7.8+9,0) {$\bigg),$};
\end{tikzpicture}
\end{split}&&
\end{flalign*}

\begin{flalign*}
\begin{split}
&
\begin{tikzpicture}[scale=0.9]
\node [align=center,align=center] at (-0.65,0) {$\cL_{0}\circ\cL_{-1}\bigg($};
\draw [fill] (0.6-0.1,0) circle [radius=0.06];
\draw [fill] (1.2-0.1,0) circle [radius=0.06];
\draw (0.6-0.1,0) -- (1.2-0.1,0);
\node [align=center,align=center] at (2.05-0.1,0) {$\bigg)=-6$};
\draw [fill] (2.7+0.25,0) circle [radius=0.06];
\draw (3.3+0.25,0) circle [radius=0.07];
\draw (3.9+0.25,0) circle [radius=0.07];
\draw (2.65+0.25,0) -- (3.23+0.25,0);
\draw (3.37+0.25,0) -- (3.83+0.25,0);
\node [align=center,align=center] at (4.8,0) {$+4$};
\draw [fill] (5.8-0.35,0) circle [radius=0.06];
\draw (6.4-0.35,0) circle [radius=0.07];
\draw [fill] (7-0.35,0) circle [radius=0.06];
\draw (7.6-0.35,0) circle [radius=0.07];
\draw (5.8-0.35,0) -- (6.33-0.35,0);
\draw (6.47-0.35,0) -- (7.53-0.35,0);
\node [align=center,align=center] at (7.9,0) {$+2$};
\draw [fill] (5.8-0.35+3.1,0) circle [radius=0.06];
\draw [fill] (6.4-0.35+3.1,0) circle [radius=0.06];
\draw (7-0.35+3.1,0) circle [radius=0.07];
\draw (7.6-0.35+3.1,0) circle [radius=0.07];
\draw (5.8-0.35+3.1,0) -- (7-0.35+3.03,0);
\draw (7-0.35+3.17,0) -- (7.6-0.35+3.03,0);
\end{tikzpicture}\\
&\quad
\begin{tikzpicture}[scale=0.9]
\node [align=center,align=center] at (4.5+0.05-0.05,0) {$+2g_s^{-2}\bigg($};
\draw [fill] (5+0.35+0.05,0) circle [radius=0.06];
\draw [fill] (5.6+0.35+0.05,0) circle [radius=0.06];
\draw (6.2+0.35+0.05,0) circle [radius=0.07];
\draw (5+0.35+0.05,0) -- (6.13+0.35+0.05,0);
\draw (6.6+0.35+0.05,0) circle [radius=0.07];
\draw (6.65+0.35+0.05,0.05) .. controls (7.3+0.35+0.05,0.55) and (7.3+0.35+0.05,-0.55) ..  (6.65+0.35+0.05,-0.05);
\node [align=center,align=center] at (8.55,0) {$\bigg)+2g_s^{-2}\bigg($};
\draw [fill] (9.7,0) circle [radius=0.06];
\draw (10.3,0) circle [radius=0.07];
\draw [fill] (10.9,0) circle [radius=0.06];
\draw [fill] (11.3,0) circle [radius=0.06];
\draw (11.9,0) circle [radius=0.07];
\draw (9.7,0) -- (10.23,0);
\draw (10.37,0) -- (10.9,0);
\draw (11.3,0) -- (11.83,0);
\node [align=center,align=center] at (8.55+4.4,0) {$\bigg)+g_s^{-2}\bigg($};
\draw [fill] (9.7+4.3,0) circle [radius=0.06];
\draw (10.3+4.3,0) circle [radius=0.07];
\draw (10.9+4.3,0) circle [radius=0.07];
\draw [fill] (11.3+4.3,0) circle [radius=0.06];
\draw [fill] (11.9+4.3,0) circle [radius=0.06];
\draw (9.7+4.3,0) -- (10.23+4.3,0);
\draw (10.37+4.3,0) -- (10.83+4.3,0);
\draw (11.3+4.3,0) -- (11.9+4.3,0);
\node [align=center,align=center] at (16.55,0) {$\bigg)$};
\end{tikzpicture}\\
&\quad
\begin{tikzpicture}[scale=0.9]
\node [align=center,align=center] at (5.1+4.6-0.05,0) {$-2g_s^{-2}\bigg($};
\draw [fill] (5.6+0.35+4.6,0) circle [radius=0.06];
\draw (6.2+0.35+4.6,0) circle [radius=0.07];
\draw (5.6+0.35+4.6,0) -- (6.13+0.35+4.6,0);
\draw (6.6+0.35+4.6,0) circle [radius=0.07];
\draw (6.65+0.35+4.6,0.05) .. controls (7.3+0.35+4.6,0.55) and (7.3+0.35+4.6,-0.55) ..  (6.65+0.35+4.6,-0.05);
\node [align=center,align=center] at (13.05,0) {$\bigg)+g_s^{-4}\bigg($};
\draw [fill] (4.6+0.35+0.6+9.2,0) circle [radius=0.06];
\draw [fill] (4+0.35+0.6+9.2,0) circle [radius=0.06];
\draw (4.35+0.6+9.2,0) -- (4.88+0.6+9.2,0);
\draw [fill] (5.6+0.35+9.2,0) circle [radius=0.06];
\draw (6.2+0.35+9.2,0) circle [radius=0.07];
\draw (5.6+0.35+9.2,0) -- (6.13+0.35+9.2,0);
\draw (6.6+0.35+9.2,0) circle [radius=0.07];
\draw (6.65+0.35+9.2,0.05) .. controls (7.3+0.35+9.2,0.55) and (7.3+0.35+9.2,-0.55) ..  (6.65+0.35+9.2,-0.05);
\node [align=center,align=center] at (7.8+9.2,0) {$\bigg),$};
\end{tikzpicture}\\
\end{split}&&
\end{flalign*}
Therefore we have:
\begin{flalign*}
\begin{tikzpicture}[scale=0.9]
\node [align=center,align=center] at (-1.4+0.05,0) {$(\cL_{-1}\cL_{0}-\cL_{0}\cL_{-1})\bigg($};
\draw [fill] (0.6-0.1,0) circle [radius=0.06];
\draw [fill] (1.2-0.2,0) circle [radius=0.06];
\draw (0.6-0.1,0) -- (1.2-0.2,0);
\node [align=center,align=center] at (2.05-0.35,0) {$\bigg)=2$};
\draw [fill] (2.7+0.25-0.4,0) circle [radius=0.06];
\draw (3.3+0.25-0.5,0) circle [radius=0.07];
\draw (3.9+0.25-0.6,0) circle [radius=0.07];
\draw (2.65+0.25-0.3,0) -- (3.23+0.25-0.5,0);
\draw (3.37+0.25-0.5,0) -- (3.83+0.25-0.6,0);
\node [align=center,align=center] at (4.8-0.6,0) {$-2$};
\draw [fill] (5.8-0.35-0.65,0) circle [radius=0.06];
\draw [fill] (6.4-0.35-0.75,0) circle [radius=0.06];
\draw (7-0.35-0.85,0) circle [radius=0.07];
\draw (7.6-0.35-0.95,0) circle [radius=0.07];
\draw (4.8,0) -- (5.73,0);
\draw (5.87,0) -- (6.23,0);
\node [align=center,align=center] at (8.55+4.4-5.7,0) {$-g_s^{-2}\bigg($};
\draw [fill] (9.7+4.3-5.9,0) circle [radius=0.06];
\draw (10.3+4.3-6,0) circle [radius=0.07];
\draw (10.9+4.3-6.1,0) circle [radius=0.07];
\draw [fill] (11.3+4.3-6.1,0) circle [radius=0.06];
\draw [fill] (11.9+4.3-6.2,0) circle [radius=0.06];
\draw (9.7+4.3-5.9,0) -- (10.23+4.3-6,0);
\draw (10.37+4.3-6,0) -- (10.83+4.3-6.1,0);
\draw (11.3+4.3-6.1,0) -- (11.9+4.3-6.2,0);
\node [align=center,align=center] at (16.55-6.15,0) {$\bigg).$};
\end{tikzpicture}&&
\end{flalign*}
Now it is clear that:
\begin{flalign*}
\begin{tikzpicture}
\node [align=center,align=center] at (-2,0) {$K^{ho}\circ(\cL_{-1}\circ\cL_{0}-\cL_{0}\circ\cL_{-1})\bigg($};
\draw [fill] (0.6-0.1,0) circle [radius=0.06];
\draw [fill] (1.2-0.2,0) circle [radius=0.06];
\draw (0.6-0.1,0) -- (1.2-0.2,0);
\node [align=center,align=center] at (2.05-0.35,0) {$\bigg)=-$};
\node [align=center,align=center] at (-0.1+2.8,0) {$\cL_{-1}\bigg($};
\draw [fill] (0.6+2.75,0) circle [radius=0.055];
\draw [fill] (1.2+2.75-0.1,0) circle [radius=0.055];
\draw (0.6+2.75,0) -- (1.2+2.75-0.1,0);
\node [align=center,align=center] at (1.9+2.45-0.1,0) {$\bigg).$};
\end{tikzpicture}
\end{flalign*}

\end{Example}

\subsection{Abstract cut-and-join type representation for $\cZ$}

\label{sec-abs-cur&join}

Notice that in the above subsections,
we have formulated the abstract Virasoro constraints in terms of
fat graphs with one hollow vertex.
In this subsection we show that one can recover the graphs without hollow vertices
using the graphs with one hollow vertex.
By doing this,
we will obtain a cut-and-join type representation for $\cZ$.

First recall that $\cZ$ is an element in the following space
of formal infinite summations of (not necessarily connected) fat graphs
without hollow vertices:
\be
\cV^{0}:=
\bQ\cdot 1\times
\prod_{g\in \bZ}\bigg(
g_s^{2g-2}\cdot\prod_{\Gamma\in \mathfrak{Fat}_{g}}
\bQ\Gamma\bigg).
\ee
Denote:
\be
\cV^{1}:=
\bQ\cdot 1\times
\prod_{g\in \bZ}\bigg(
g_s^{2g-2}\cdot\prod_{\Gamma\in \mathfrak{Fat}_{g}^1}
\bQ\Gamma\bigg),
\ee
where $\mathfrak{Fat}_{g}^1$ is the set of all (not necessarily connected) fat graphs
with one hollow vertex (and an arbitrary number of solid vertices).
Then the operators $\{\pd_n\}_{n\geq 1}$ and the abstract Virasoro operators $\{\cL_m\}_{m\geq -1}$
are linear maps from $\cV^0$ to $\cV^1$.

Now let $\cI: \cV^1 \to \cV^0$ be the linear map
which takes a graph (with a hollow vertex) to the graph obtained by
simply changing this hollow vertex to a solid one.
Then we have:
\begin{Lemma}
\label{lem-rescaling}
The compositions $\cI\circ(\sum_{n\geq 1}\pd_n)$ and $\cI\circ(\sum_{n\geq 1}n\pd_n)$
are given by the following rescalings:
\be
\begin{split}
&\cI\circ \bigg(
\sum_{n\geq 1} \pd_n
\bigg)
(\Gamma) = |V(\Gamma)| \cdot \Gamma,\\
&\cI\circ \bigg(
\sum_{n\geq 1} n\pd_n
\bigg)
(\Gamma) = 2|E(\Gamma)| \cdot \Gamma,
\end{split}
\ee
where $|V(\Gamma)|$ and $|E(\Gamma)|$ are the number of vertices and edges in $\Gamma$
respectively.
\end{Lemma}
\begin{proof}
The first relation follows from the definitions of these operators.
The second relation holds since the left-hand side simply counts the summation of the valences
of all vertices which is exactly $2|E(\Gamma)|$.
\end{proof}

Define a family of linear operators
\be
\cL_m': \cV^0 \to \cV^1,
\qquad
m\geq -1,
\ee
to be:
\be
\label{eq-Vir-opr-ref}
\begin{split}
&
\cL_{-1}':=\sum_{n\geq 1}\cS_{n,n}
+g_s^{-2}\cdot\gamma_{-1},\\
&
\cL_{0}':=\sum_{n\geq 1}\cS_{n,n-1}
+g_s^{-2}\cdot\gamma_0,\\
&
\cL_{1}':=\sum_{n\geq 1}\cS_{n+1,n-1}
+2\cJ_{1,0}(-\sqcup\Gamma_{dot}),\\
&\cL_{m}':=\sum_{n\geq 1}\cS_{n+m,n-1}
+g_s^2\cdot \sum_{n=1}^{m-1}\cJ_{n,m-n}
+2\cJ_{m,0}(-\sqcup\Gamma_{dot}),
\quad
m\geq 2,
\end{split}
\ee
i.e., define $\cL_m' = (m+2)\pd_{m+2} +\cL_m$ for every $m\geq -1$.
Then the abstract Virasoro constraints $\cL_m(\cZ)=0$ becomes:
\be
\label{eq-abs-Virasoro-ref}
\cL_m'(\cZ) = (m+2) \pd_{m+2} (\cZ),
\qquad
m\geq -1.
\ee
Now apply the similar method used in \cite{al},
one gets the following:
\begin{Theorem}
\label{thn-abs-cut&join}
Define a linear map $\cM: \cV^0 \to \cV^0$ by:
\be
\label{eq-def-cM}
\cM := \half \sum_{m\geq -1}\cI \circ \cL_m',
\ee
then we have:
\be
\label{eq-abs-cut&join}
\cZ = e^\cM (1),
\ee
where $1$ is the empty graph.
\end{Theorem}
\begin{proof}
Define a grading on $\cV^0$ by the number of edges,
i.e.,
define $\deg(\Gamma):= |E(\Gamma)|$ for every graph $\Gamma$.
Then the abstract partition function $\cZ$ can be decomposed as:
\be
\cZ = \sum_{k\geq 0} \cZ_{(k)},
\ee
where $\deg(\cZ_{(k)})=k$, and $\cZ_{(0)}=1$ is the empty graph.
Using the abstract Virasoro constraints \eqref{eq-abs-Virasoro-ref},
one gets:
\begin{equation*}
\cM (\cZ) = \sum_{m\geq -1} \half \cI\circ \cL_m' (\cZ)
= \sum_{n\geq 1} \half \cI\circ n\pd_n (\cZ)
=\half \sum_{k\geq 0} \sum_{n\geq 1} n\cI\circ \pd_n (\cZ_{(k)}).
\end{equation*}
And Lemma \ref{lem-rescaling} one knows that:
\begin{equation*}
\sum_{n\geq 1} n\cI\circ \pd_n (\cZ_{(k)})= 2k \cdot \cZ_{(k)},
\end{equation*}
thus:
\be
\label{eq-deg-Z}
\cM (\cZ) = \sum_{k\geq 0} k\cdot \cZ_{(k)}.
\ee

Recall that in the definition \eqref{eq-Vir-opr-ref} of $\cL_m'$ ($m\geq -1$),
the operators on the right-hand sides all creates two new half-edges
(see \S \ref{sec-def-S&J}),
thus one has $\deg (\cM)=1$.
Then \eqref{eq-deg-Z} tells that:
\ben
\cM (\cZ_{(k)}) = (k+1)\cdot \cZ_{(k+1)},
\qquad
\forall k\geq 0,
\een
and then:
\ben
\cZ_{(k)} =
\frac{1}{k!} \cM^k (\cZ_{(0)})
=\frac{1}{k!} \cM^k (1),
\qquad \forall k\geq 0,
\een
where $1$ is the empty graph.
Therefore:
\ben
\cZ = \sum_{k\geq 0} \cZ_{(k)} =\sum_{k\geq 0} \frac{1}{k!} \cM^k (1)
= e^{\cM} (1).
\een
\end{proof}

\begin{Remark}
We call \eqref{eq-abs-cut&join} the abstract cut-and-join type representation
for the abstract partition function $\cZ$.
In the next section we will see that in the realization by the Hermitian one-matrix models,
\eqref{eq-abs-cut&join} gives a cut-and-join type representation for the
partition function $Z_N^{\text{Herm}}$,
and this leads to the fact that $Z_N^{\text{Herm}}$ is a tau-function of
the KP hierarchy.

\end{Remark}

\section{Realization by Hermitian One-Matrix Models at Finite $N$}
\label{sec-realization-Virasoro}

In this section we consider the realization of the abstract QFT by
the Hermitian one-matrix models.
We show that the fat Virasoro operators, fat Virasoro constraints,
and the Virasoro commutation relations (see \S \ref{sec-fatVirasoro})
all follow from the formalism developed in last section.
Moreover,
we regard the partition function $Z_N^{\text{Herm}}|_{g_s=1}$ as a vector
in the bosonic Fock space $\Lambda$,
and use the realization of the abstract cut-and-join type representation
to show that $Z_N^{\text{Herm}}|_{g_s=1}$ is a tau-function of the KP hierarchy.

\subsection{Realization by the fat Virasoro constraints at finite $N$}
\label{sec-realization-fatVirasoro}

Now let us consider the realizations of the abstract Virasoro constraints.
Since this formalism is inspired by the fat Virasoro constraints of
the Hermitian one-matrix models,
let us first consider how to realize the fat Virasoro constraints.

Recall that in \S \ref{sec-realization-herm-1mm} we have chosen
the Feynman rule to be \eqref{eq-1mm-FR2}
such that the abstract partition function $\cZ$ and abstract free energy $\cF$
are realized by $Z_N^{\text{Herm}}$ and $F_N^{\text{Herm}}$ respectively.
Now let us modify this Feynman rule in the following way
to include hollow vertices:
\be
\label{eq-FR-1mm-absvirasoro}
\begin{split}
&w_v:=g_{\val(v)}, \quad\text{for a solid vertex $v$;}\\
&w_v:=1, \qquad\quad\text{for a hollow vertex $v$;}\\
&w_e:=1, \qquad\quad\text{for an internal edge $e$;}\\
&w_f:=t, \qquad\quad\text{for a face $f$,}
\end{split}
\ee
where $\val(v)$ is the valence of the vertex $v$,
and $t=Ng_s$ is the 't Hooft coupling constant.
In particular,
if $\Gamma$ is a fat graph without hollow vertices,
then the weight of $\Gamma$ is:
\ben
w_\Gamma=t^{|F(\Gamma)|}\cdot\prod_{v\in V(\Gamma)}g_{\val(v)},
\een
therefore the realizations of $\cF$ and $\cZ$ are
still $F_N^{\text{Herm}}$ and $Z_N^{\text{Herm}}$ respectively.

Now let us consider the realization of the abstract operators
defined in previous subsections.
It is not hard to see that:

\begin{Lemma}
\label{lem-realization-fatoperator}
For the Feynman rule \eqref{eq-FR-1mm-absvirasoro},
we have:
\begin{itemize}
\item[1)]
The operator $\pd_n$ is realized by the partial derivative $\frac{\pd}{\pd g_n}$;

\item[2)]
The operator $\cS_{m,n}$ is realized by $mg_{n+1}\frac{\pd}{\pd g_m}$;

\item[3)]
The operator $\cJ_{m,n}$ is realized by $mn\frac{\pd}{\pd g_m}\frac{\pd}{\pd g_n}$
for $m,n>0$;
and the operator $\cJ_{m,0}(-\sqcup \Gamma_{dot})$ is realized by $mt\frac{\pd}{\pd g_m}$;

\item[4)]
The operator $\gamma_{-1}$ is realized by multiplying by $tg_1$,
and the operator $\gamma_{0}$ is realized by multiplying by $t^2$.
\end{itemize}
\end{Lemma}

Since the procedure $K^{ho}$ does not change the weight of a fat graph,
as a corollary of this lemma we have:

\begin{Theorem}
Given the Feynman rule \eqref{eq-FR-1mm-absvirasoro},
the abstract Virasoro operator $\cL_m$ is realized by
the fat Virasoro operator $L_{m,t}^{\text{Herm}}$
(see \eqref{eq-1mm-fatvir-opr-1}, \eqref{eq-1mm-fatvir-opr-2})
of the Hermitian one-matrix models for every $m\geq -1$.
Moreover, the abstract Virasoro constraints in Theorem \ref{thm-abstract-vir1}
is realized by the fat Virasoro constraints:
\ben
L_{m,t}^{\text{Herm}} Z_N^{\text{Herm}} =0, \qquad\forall m\geq -1,
\een
and Theorem \ref{thm-fat-abs-virasorocomm}
is realized by the Virasoro commutation relations:
\ben
[L_{m,t}^{\text{Herm}},L_{n,t}^{\text{Herm}}]=(m-n)L_{m+n,t}^{\text{Herm}},
\qquad \forall m,n\geq -1.
\een

\end{Theorem}

\subsection{Realization on the bosonic Fock space}

The partition function $Z_N^{\text{Herm}}$ of the Hermitian one-matrix models is
known to be a $\tau$-function of the KP hierarchy,
with respect to the variables $T_n:=\frac{g_n}{n}$
(see \cite{sty}).
Thus it is natural to understand $Z_N^{\text{Herm}}$
as a vector in the bosonic Fock space.
In this subsection let us discuss the realization of
the abstract Virasoro constraints on the bosonic Fock space,
i.e.,
on the space of all symmetric functions.

Now let us recall the construction of the bosonic Fock space.
Denote by $\Lambda$ the space of all symmetric functions,
then $\Lambda$ can be spanned by $\{p_{\lambda}\}_{\lambda}$,
where
\ben
p_{\lambda}:=p_{\lambda_1}\cdots p_{\lambda_n}
\een
for a partition $\lambda=(\lambda_1,\cdots,\lambda_n)$,
and $p_k$ is the Newton symmetric function of degree $k$.
It is well-known that $\Lambda$ admits a natural structure of a bosonic Fock space.
Let $n\in\bZ$,
and define the operators $\alpha_n :\Lambda\to\Lambda$ (the `bosons') by:
\be
\label{eq-boson}
\alpha_n (f):=\left\{
\begin{aligned}
&p_{-n}\cdot f, & \quad  n<0;\\
&0, & \quad  n=0; \\
&n\frac{\pd}{\pd p_n}f,  & \quad  n>0,
\end{aligned}
\right.
\ee
then $\{\alpha_n\}_{n\in \bZ}$ span a Heisenberg algebra,
i.e.,
they satisfy the following Heisenberg commutation relations:
\be
[\alpha_m,\alpha_n]=m\delta_{m+n,0},\quad
\forall m,n\in\bZ.
\ee
The operators $\alpha_n$ with $n<0$ are called the creators,
and  $\alpha_n$ with $n>0$ are called the annihilators.
The vector $1\in \Lambda$ is called the bosonic vacuum vector,
and it is annihilated by the annihilators $\alpha_n$ ($n>0$).
It is also clear that:
\ben
p_\lambda=\alpha_{-\lambda}1,
\een
where $\alpha_{-\lambda}:=\alpha_{-\lambda_1}\cdots\alpha_{-\lambda_n}$
for a partition $\lambda=(\lambda_1,\cdots,\lambda_n)$.

The space $\Lambda$ carries a natural Hermitian metric defined by:
\be
\langle
p_\lambda,p_\mu\rangle
=z_\mu\delta_{\lambda,\mu},
\ee
where
\ben
z_{\mu}:=\prod_k k^{m_k}\cdot m_k!,
\een
and $m_k$ is the number of $k$ appearing in the partition $\mu$.
The operators $\alpha_n$ and $\alpha_{-n}$ are dual to each other
with respect to this metric for every $n\not=0$.

Now let us consider the realization of the abstract Virasoro constraints
on the bosonic Fock space $\Lambda$.
Take $g_s=1$ and $g_n:=p_n$ in the realization
introduced in \S \ref{sec-realization-fatVirasoro}.
The Feynman rule is:
\be
\label{eq-FR-bosonicfock}
\begin{split}
&w_v:=p_{\val(v)}, \quad\text{for a solid vertex $v$;}\\
&w_v:=1, \qquad\quad\text{for a hollow vertex $v$;}\\
&w_e:=1, \qquad\quad\text{for an internal edge $e$;}\\
&w_f:=N, \quad\quad\text{for a face $f$.}
\end{split}
\ee
then the realization
$Z^B:=Z_N^{\text{Herm}} |_{g_s=1}$ of the abstract partition function
$ \cZ $ is an element in $\Lambda$.
Now consider the realizations of the operators.
Notice that for every partition $\lambda=(\lambda_1,\cdots,\lambda_n)$,
the symmetric function $p_\lambda$ can be represented as
the weight of some fat graph:
one simply take $n$ vertices of valences $\lambda_1,\cdots,\lambda_n$ respectively,
and glue them together arbitrarily to obtain a fat graph
(notice that if $|\lambda|$ is odd,
one need to add an additional hollow vertex of an odd valence).
Then similar to Lemma \ref{lem-realization-fatoperator}, we have:
\begin{itemize}
\item[1)]
The operator $\pd_n$ is realized by $\frac{1}{n}\alpha_n$;

\item[2)]
The operator $\cS_{m,n}$ is realized by $\alpha_{-(n+1)}\alpha_m$;

\item[3)]
The operator $\cJ_{m,n}$ is realized by $\alpha_m\alpha_n$;

\item[4)]
The operator $\gamma_{-1}$ is realized by $N\alpha_{-1}$,
and the operator $\gamma_{0}$ is realized by multiplying by $N^2$.
\end{itemize}
Thus the realizations of Theorem \ref{thm-abstract-vir1} and
Theorem \ref{thm-fat-abs-virasorocomm} gives us
the following Virasoro constraints on the bosonic Fock space $\Lambda$:

\begin{Theorem}
Define a family of operators $\{L_m\}_{m\geq -1}$ on $\Lambda$ by:
\begin{equation}
\begin{split}
&L_{-1}:=-\alpha_1+\sum_{n\geq 1}\alpha_{-(n+1)}\alpha_n+N\alpha_{-1};\\
&L_0:=-\alpha_2+\sum_{n\geq 1}\alpha_{-n}\alpha_n+N^2;\\
&L_1:=-\alpha_3+\sum_{n\geq 1}\alpha_{-n}\alpha_{n+1}+2N\alpha_1;\\
&L_m:=-\alpha_{m+2}+\sum_{n\geq 1}\alpha_{-n}\alpha_{n+m}
+\sum_{n=1}^{m-1}\alpha_n\alpha_{m-n}+2N\alpha_m,
\quad m\geq 2.
\end{split}
\end{equation}
Then these operators satisfy the Virasoro commutation relations:
\be
[L_m,L_n]=(m-n)L_{m+n},
\qquad \forall m,n\geq -1.
\ee
Moreover,
they annihilate the partition function $Z_N^{\text{Herm}} |_{g_s=1}$:
\be
L_m (Z_N^{\text{Herm}} |_{g_s=1}) = 0, \qquad \forall m\geq -1.
\ee
\end{Theorem}

\begin{Remark}

The realizations of operators on fat graphs on the bosonic Fock space suggests us to regard the space spanned by fat graphs
as a refinement of the bosonic Fock space,
and this will lead to some connections to representation theory and integrable hierarchies.
We wish to discuss this topic in future investigations.

\end{Remark}

\subsection{Realization of the abstract cut-and-join type representation}

In this subsection let us consider the realization of the abstract cut-and-join type representation
\eqref{eq-abs-cut&join} for $\cZ$.
For simplicity we take $g_s=1$ in this subsection,
and the realization of $\cZ$ is $Z_N^{\text{Herm}}|_{g_s=1}$ given the Feynman rule \eqref{eq-FR-1mm-absvirasoro}.

First consider the realization of the operator $\cM$.
Recall that the linear operator $\cI$ turns a hollow vertex into a solid one.
Since the operator $\cL_m'$ ($m\geq -1$) defined by \eqref{eq-Vir-opr-ref}
creates a hollow vertex of valence $m+2$,
the composition
\ben
\cI \circ \cL_m' : \cV^0 \to \cV^0
\een
is realized by $g_{m+2} \cdot L_m'$ under the Feynman rule \eqref{eq-FR-1mm-absvirasoro},
where $L_m'$ is the following realization of $\cL_m'$:

\be
\begin{split}
&L_{-1}'=
\sum_{n\geq 1}ng_{n+1}\frac{\pd}{\pd g_n}
+tg_1,\\
&L_0'=
\sum_{n\geq 1}ng_{n}\frac{\pd}{\pd g_n}
+t^2,\\
&L_1'=
\sum_{n\geq 1}(n+1)g_{n}\frac{\pd}{\pd g_{n+1}}
+2t\frac{\pd}{\pd g_1},\\
&L_m'=\sum_{k\geq 1}(k+m)g_k\frac{\pd}{\pd g_{k+m}}
+\sum_{k=1}^{m-1}k(m-k)\frac{\pd}{\pd g_k}
\frac{\pd}{\pd g_{m-k}}+2tm\frac{\pd}{\pd g_m},
\quad m\geq 2.
\end{split}
\ee
Therefore the realization of the operator $\cM:\cV^0 \to \cV^0$ (see \eqref{eq-def-cM}) is:
\be
M= \frac{1}{2} \sum_{m\geq -1} g_{m+2}\cdot L_m',
\ee
and Theorem \ref{thn-abs-cut&join} gives the following cut-and-join representation of
$Z_N^{\text{Herm}}|_{g_s=1}$:
\begin{Theorem}
We have:
\be
Z_N^{\text{Herm}}|_{g_s=1} = e^M (1).
\ee
\end{Theorem}

\subsection{The element in $\widehat{\mathfrak{gl}(\infty)}$ corresponding to
the KP tau-function $Z_N^{\text{Herm}}$}

In this subsection we use the above cut-and-join type representation to
show that $Z_N^{\text{Herm}}|_{g_s=1}$ is a tau-function of the KP hierarchy.
See \cite{djm, sa} for an introduction to Kyoto School's approach to
the KP hierarchy.
See \cite{kz} for a proof of the fact that the dessin partition function
is a KP tau-function using this method.

First let us take $g_n= p_n$ for every $n\geq 1$,
and regard $Z_N^{\text{Herm}}|_{g_s=1}$ as an element in the
bosonic Fock space $\Lambda$.
It is clear that the above realization $M$ of the operator $\cM$
can be rewritten in terms of the bosonic operators \eqref{eq-boson}:
\be
M= \half\sum_{i+j+k=-2} :\alpha_i \alpha_j \alpha_k:
+\frac{t}{2} \sum_{i+j=-2} :\alpha_i \alpha_j: + \frac{t^2}{2} \alpha_{-2},
\ee
where $:\alpha_{i_1}\alpha_{i_2}\cdots \alpha_{i_n}:$ is the normal ordering
of the bosons $\alpha_{i_1},\cdots, \alpha_{i_n}$ defined by:
\ben
:\alpha_{i_1}\alpha_{i_2}\cdots \alpha_{i_n}:
= \alpha_{i_{\sigma(1)}}\alpha_{i_{\sigma(2)}}\cdots \alpha_{i_{\sigma(n)}},
\een
where $\sigma\in S_n$ is a permutation such that $\sigma(1)\leq \sigma(2)\leq\cdots\leq \sigma(n)$.
It is known that the operators
\ben
\sum_{i+j+k=-2} :\alpha_i \alpha_j \alpha_k:,
\qquad
\sum_{i+j=-2} :\alpha_i \alpha_j:,
\qquad
\alpha_{-2}
\een
are all elements in the infinite-dimensional Lie algebra $\widehat{\mathfrak{gl}(\infty)}$
(see eg. Kazarian \cite[\S 4]{ka}),
thus $M\in \widehat{\mathfrak{gl}(\infty)}$,
and then:
\be
e^M \in \widehat{GL(\infty)}.
\ee

In Sato's theory of the KP hierarchy,
the space of all (formal power series) tau-functions is a $\widehat{GL(\infty)}$-orbit
of the trivial solution $\tau=1$.
Therefore,
the partition function $Z_N^{\text{Herm}}|_{g_s=1} = e^M (1)$ is also a tau-function
of the KP hierarchy.
This is a well-known result in the literature proved by other methods,
see eg. \cite{sty}.

\section{Quantum Deformation Theory of the Spectral Curve for the Abstract QFT for Fat Graphs}
\label{sec-QDT}

In this section,
we explain how to construct a spectral curve for the abstract QFT for fat graphs
from the abstract Virasoro constraints,
and discuss the quantum deformation theory of this spectral curve.
The notion of the `special deformation of the spectral curve' was introduced
by the second author in \cite{zhou9} in the study of topological 2D gravity.
In that case,
the spectral curve is the Airy curve,
and the special deformation of the Airy curve can be constructed
using the one-point functions of genus zero of the Witten-Kontsevich tau-function.
The Virasoro constraints for topological 2D gravity are encoded in
the quantum deformation theory of the Airy curve.
In \cite{zhou2}, the second author
generalized this idea to the formalism of emergent geometry.
Given a Gromov-Witten type theory (A-theory),
the spectral curve of this theory together with its special deformation
(as a B-theory)
should emerge from the Virasoro constraints.
The emergence of spectral curves for the Hermitian one-matrix models
(with respect to both fat and thin genus expansion)
have been presented in \cite{zhou10}.
The emergence of the Eynard-Orantin topological recursion
is discussed in \cite{zhou8,zhou11}.
In this work we generalize some of these results to the setting of abstract QFT of fat graphs.
As an application,
we show that the quantum deformation theory of the spectral curve
for the Hermitian one-matrix models is a realization of this theory.

\subsection{A reformulation of the abstract Virasoro constraints of genus zero}

In this subsection,
let us reformulate the abstract Virasoro constraints
of genus zero in a more compact way.
This reformulation will be useful when we discuss the special deformation
in next subsection.

Let $k$ be a non-negative integer,
and denote by $\mathfrak{Fat}^{ho-k}$ the set of fat graphs with $k$ hollow vertices
(and an arbitrary number of solid vertices).
Let $\cV^{ho-k}$ be the linear space:
\be
\cV^{ho-k}:=\prod_{\Gamma\in\mathfrak{Fat}^{ho-k}}\bQ\Gamma.
\ee
Now let us introduce a family of linear operators
\ben
\fs_k:\cV^{ho-1}\to\cV^{ho-1},
\qquad  k\geq 0,
\een
in the following way.

\begin{Definition}
Let $\Gamma$ be a fat graph with a hollow vertex of valence $n$,
then:
\begin{itemize}
\item[1)]
If $0\leq k\leq n$,
define
\ben
\fs_k(\Gamma):=\frac{1}{n}\sum_{h\in H(v)} \cS(\Gamma,h),
\een
where $v$ is the hollow vertex in $\Gamma$,
$H(v)$ is the set of half-edges attached to $v$,
and the definition of $\cS(\Gamma,h)$ is given in Definition \ref{def-vertexsplit-opr}.

\item[2)]
If $k>n$,
define $\fs_k(\Gamma):=0$.

\end{itemize}
\end{Definition}

In other words,
the operator $\fs_k$ is a vertex-splitting operator for hollow vertices--
it splits a hollow vertex of valence $n$
into a new hollow vertex of valence $n-k+1$ and a new solid vertex of valence $k+1$,
and creates a new edge connecting them.
By definition we have the following easy observation:

\begin{Lemma}
\label{lem-specialdefo1}
For every $0\leq k\leq n$,
we have:
\be
\fs_k\circ \pd_n=\frac{1}{n}\cS_{n,k}
\ee
on $\cV^{ho-0}$.
\end{Lemma}

Moreover,
we need to introduce another linear operator $\fj:\cV^{ho-2}\to\cV^{ho-1}$.

\begin{Definition}
Let $\Gamma$ be a fat graph with two hollow vertices $v_1$ and $v_2$,
then we define:
\ben
\fj(\Gamma):=\frac{1}{\val(v_1)\cdot \val(v_2)}
\sum_{\substack{h_1\in H(v_1)\\h_2\in H(v_2)}}\cJ(\Gamma,h_1,h_2),
\een
where $\val(v_i)$ is the valence of the vertex $v_i$,
and the definition of $\cJ(\Gamma,h_1,h_2)$ is given in Definition \ref{def-abs-opr-J}.
\end{Definition}

It is not hard to see that:
\begin{Lemma}
\label{lem-specialdefo2}
For every $m,n\geq 1$,
we have:
\be
\fj\circ(\pd_m\pd_n)=\frac{1}{m\cdot n}\cJ_{m,n}
\ee
on $\cV^{ho-0}$.
\end{Lemma}

\begin{Remark}
As we will see in \S \ref{sec-QDT-Herm},
the operators $\fs_n$ and $\pd_n$ are analogues of the
bosonic creators (acting as multiplying by a coupling constant)
and the bosonic annihilators (acting as taking partial derivative w.r.t. a coupling constant)
respectively.
Here we generalize these bosonic operators on the bosonic Fock space to operators on graphs.

\end{Remark}

Now let us reformulate the abstract Virasoro constraints of genus zero
using the above operators.
Recall that $\cZ=\exp\cF$ and $\cF=\sum_{g\geq 0} g_s^{2g-2}\cF_g$,
then expanding the abstract Virasoro constraints \eqref{thm-abstract-vir1}
and analyzing the coefficients of $g_s^{-2}$,
we get the following abstract Virasoro constraints for $\cF_0$:
\be
\label{eq-absVir-genus0}
\begin{split}
&-\pd_1\cF_0+\sum_{n\geq 1}\cS_{n,n}\cF_0+\Gamma_{-1}=0,\\
&-2\pd_2\cF_0+\sum_{n\geq 1}\cS_{n,n-1}\cF_0+\Gamma_{0}=0,\\
&-3\pd_{3}\cF_0+\sum_{n\geq 1}\cS_{n+1,n-1}\cF_0
+2\cJ_{1,0}(\cF_0\cdot\Gamma_{dot})=0,\\
&-(m+2)\pd_{m+2}\cF_0+\sum_{n\geq 1}\cS_{n+m,n-1}\cF_0
+\sum_{n=1}^{m-1}n(m-n)\cdot\fj\big(\pd_n\cF_0\cdot\pd_{m-n}\cF_0\big)\\
&\quad +2\cJ_{m,0}\big(\cF_0\cdot\Gamma_{dot}\big)=0,
\quad m\geq 1,\\
\end{split}
\ee
see Definition \ref{def-abs-Viropr}
for the expressions of $\Gamma_{-1}$, $\Gamma_0$ and $\Gamma_{dot}$.

Let $z$ be a formal variable,
and denote:
\be
\begin{split}
&\fs:=\frac{1}{\sqrt{2}}\sum_{n=0}^\infty
(\fs_{n}-\delta_{n,1})\cdot z^{n},\\
&\ff_0:=\sqrt{2}\bigg(
\pd_0(\Gamma_{dot})\cdot z^{-1}+
\sum_{n=1}^\infty n\big(\pd_n\cF_0\big)\cdot z^{-n-1}
\bigg),
\end{split}
\ee
where the convention $\pd_0(\Gamma_{dot})$ is the graph consisting of
a single hollow vertex of valence zero.
Then we have:

\begin{Proposition}
\label{prop-specialdefo}
The abstract Virasoro constraints \eqref{eq-absVir-genus0} for $\cF_0$
is equivalent to:
\be
2\fs(\ff_0)+\fj(\ff_0\cdot\ff_0)=-\pd_0(\Gamma_{dot}).
\ee
\end{Proposition}

\begin{proof}
By Lemma \ref{lem-specialdefo1} and Lemma \ref{lem-specialdefo2} we have:
\begin{equation*}
\begin{split}
\fs(\ff_0)
=&\fs_0(\pd_0\Gamma_{dot})z^{-1}-\pd_0\Gamma_{dot}
+\sum_{k=0}^\infty\sum_{n=k}^\infty n\big(\fs_k\pd_n\cF_0\big)z^{k-n-1}
-\sum_{n=1}^{\infty} n\big(\pd_n\cF_0\big)z^{-n-1}
\\
=&\Gamma_{-1}\cdot z^{-1}-\pd_0\Gamma_{dot}
+\sum_{k=0}^\infty\sum_{n=k}^\infty \big(\cS_{n,k}\cF_0\big)z^{k-n-1}
-\sum_{n=1}^{\infty} n\big(\pd_n\cF_0\big)z^{-n},
\end{split}
\end{equation*}
and
\begin{equation*}
\begin{split}
\fj(\ff_0\cdot\ff_0)=&
2\fj\bigg(\big(\pd_0\Gamma_{dot}\big)^2\bigg)z^{-2}
+2\sum_{m,n\geq 1}mn\cdot\fj\big(\pd_m\cF_0\cdot\pd_n\cF_0\big)z^{-m-n-2}\\
&+4\sum_{n=1}^\infty n\cdot \fj\bigg(
\big(\pd_0\Gamma_{dot}\big)\cdot \pd_n\cF_0
\bigg)z^{-n-2}
\\
=&2\Gamma_0\cdot z^{-2}+
2\sum_{m,n\geq 1}mn\cdot\fj\big(\pd_m\cF_0\cdot\pd_n\cF_0\big)z^{-m-n-2}\\
&+4\sum_{n=1}^\infty \cJ_{n,0}(\Gamma_{dot}\cdot\cF_0)z^{-n-2},\\
\end{split}
\end{equation*}
therefore:
\begin{equation*}
\begin{split}
&2\fs(\ff_0)+\fj(\ff_0\cdot\ff_0)+\pd_0(\Gamma_{dot})\\
=&
2\bigg(-\pd_1\cF_0+\sum_{n\geq 1}\cS_{n,n}\cF_0+\Gamma_{-1}\bigg)z^{-1}
+2\bigg(
-2\pd_2\cF_0+\sum_{n\geq 1}\cS_{n,n-1}\cF_0+\Gamma_{0}
\bigg)z^{-2}\\
&+2\sum_{m=1}^\infty\bigg(
-(m+2)\pd_{m+2}\cF_0+\sum_{n\geq 1}\cS_{n+m,n-1}\cF_0\\
&+\sum_{n=1}^{m-1}n(m-n)\cdot\fj\big(\pd_n\cF_0\cdot\pd_{m-n}\cF_0\big)
+2\cJ_{m,0}\big(\cF_0\cdot\Gamma_{dot}\big)
\bigg)z^{-m-2}.
\end{split}
\end{equation*}
Then the conclusion is clear.
\end{proof}

\subsection{Emergence of the spectral curve and its special deformation}

In this subsection,
we construct the spectral curve and its special deformation
for the abstract QFT for fat graphs.

Now let us construct the spectral curve and its special deformation
for our abstract QFT for fat graphs.
Let $\{u_n\}_{n\geq 1}$ be a family of formal variables,
and denote:
\be
\cU:=\bQ\times\prod_{n\geq 1}\bQ\cdot u_n.
\ee
In what follows,
we will define the special deformation to be a series in a formal variable $z$,
whose coefficients are in $\cU\times\cV^{ho-1}$.

\begin{Definition}
The special deformation of the spectral curve of the abstract QFT for fat graphs
is defined to be the following:
\be
\label{eq-specialdefo-absQFT}
\begin{split}
y=&\frac{1}{\sqrt{2}}\sum_{n=0}^\infty
(u_{n}-\delta_{n,1})\cdot z^{n}+\ff_0\\
=&\frac{1}{\sqrt{2}}\sum_{n=0}^\infty
(u_{n}-\delta_{n,1})\cdot z^n+
\sqrt{2}\bigg(
\pd_0(\Gamma_{dot})\cdot z^{-1}+
\sum_{n=1}^\infty n\big(\pd_n\cF_0\big)\cdot z^{-n-1}\bigg),
\end{split}
\ee
where $y,z$ are two formal variables.

\end{Definition}

This special deformation is a family of curves on the $(y,z)$-plane,
parametrized by elements in the linear space $\cU\times\cV^{ho-1}$
(i.e., parametrized by formal variables $u_n$ and fat graphs with one hollow vertex).

Now let us explain what is our `spectral curve'.
We need to restrict \eqref{eq-specialdefo-absQFT} to
a smaller space of deformation parameters.
We consider the following subspace of $\cU\times \cV^{ho}$:
\ben
\bQ \times \prod_{k\geq 0}\prod_{\Gamma\in \mathfrak{Ho}_k} \bQ \Gamma,
\een
where $\mathfrak{Ho}_k$ is the set of all fat graph of genus zero,
consisting of a single hollow vertex
(and no solid vertices),
and $k$ loops attached to it.
In other words,
we take $u_n=0$ for each $n$ and
set every graph with at least one solid vertex to be zero
in \eqref{eq-specialdefo-absQFT}.
In this way,
we obtain the following family of curves on the $(y,z)$-plane:
\be
\label{eq-abstract-spectralcurve}
y=-\frac{z}{\sqrt{2}}
+\sqrt{2} \pd_0(\Gamma_{dot}) z^{-1}
+\sqrt{2}\cdot \sum_{k\geq 1}
\bigg(
\sum_{\Gamma\in \mathfrak{Ho}_k}\frac{2k}{|\Aut(\Gamma)|}\Gamma
\bigg) z^{-2k-1}.
\ee
We refer this curve as the spectral curve of our abstract QFT
for fat graphs.

The Virasoro constraints of genus zero is encoded in the special deformation
of the spectral curve in the following way.
Define a `multiplication' $\Phi$ of formal variables $u_n$
and graphs in $\mathfrak{Fat}^{ho-1}$ to be a linear map:
\be
\Phi:
\Sym^2\big(\cU\times \cV^{ho-1}\big)
\to
\cU\times
\bigg(\prod_{i\leq j}\bQ u_i u_j\bigg)
\times
\cV^{ho-1}\times
\bigg(
\prod_{\substack{n\geq 1\\\Gamma\in \mathfrak{Fat}^{ho-1}}}
\bQ\cdot u_n\Gamma
\bigg)
\ee
by requiring:
\begin{itemize}
\item[1)]
$\Phi(c,-)$ is multiplying by $c$ for $c\in \bQ$;
\item[2)]
$\Phi(u_i,u_j)=u_iu_j=u_ju_i$;
\item[3)]
$\Phi(\Gamma,\Gamma')=\fj(\Gamma\cdot \Gamma')$
for $\Gamma,\Gamma'\in \mathfrak{Fat}^{ho-1}$;
\item[4)]
Let $\Gamma\in\mathfrak{Fat}^{ho-1}$ be a graph with
hollow vertex of valence $n$,
then we require $\Phi(u_i,\Gamma)=u_i\Gamma$ if $i>n$;
\item[5)]
Let $\Gamma\in\mathfrak{Fat}^{ho-1}$ be a graph with
hollow vertex of valence $n$,
then we require $\Phi(u_i,\Gamma)=\fs_i(\Gamma)$ if $0<i\leq n$.
\end{itemize}

\begin{Remark}
This map $\Phi$ is well-defined,
since it is not hard to see that given any $u_i$
(or $u_iu_j$, $\Gamma$, $u_i\Gamma$),
there are only finitely many ways to obtain such an element
from a pair $(c,u_i)$, $(c,\Gamma)$,
$(u_i,u_j)$,
$(u_i,\Gamma)$,
or $(\Gamma,\Gamma')$.
\end{Remark}

Then our main result in this subsection is:

\begin{Theorem}
\label{thm-abstract-vira-genus0}
The abstract Virasoro constraints \eqref{eq-absVir-genus0}
for $\cF_0$ are equivalent to:
\be
\bigg(\Phi(y,y)\bigg)_-=0,
\ee
where for a formal series $\sum\limits_{n\in \bZ}a_n z^n$,
we denote
$\bigg(\sum\limits_{n\in \bZ}a_n z^n\bigg)_-
:=\sum\limits_{n<0}a_n z^n$.
\end{Theorem}

\begin{proof}
We have:
\begin{equation*}
\begin{split}
\Phi(y,y)=\bigg(
\frac{1}{\sqrt{2}}\sum_{n=0}^\infty
(u_{n}-\delta_{n,1}) z^{n}
\bigg)^2+\fj(\ff_0\cdot\ff_0)
+2\Phi\bigg(
\frac{1}{\sqrt{2}}\sum_{n=0}^\infty
(u_{n}-\delta_{n,1})z^{n},
\ff_0\bigg),
\end{split}
\end{equation*}
therefore
\ben
\bigg(
\Phi(y,y)
\bigg)_-=
\bigg(\fj(\ff_0\cdot\ff_0)+2\fs(\ff_0)\bigg)_-,
\een
then the conclusion follows from Proposition \ref{prop-specialdefo}.

\end{proof}

\subsection{Heisenberg algebra of operators on fat graphs}

What we have done in the previous subsection
actually indicates a natural structure of Heisenberg algebra.
Let us explain this in the present subsection.

Instead of considering fat graphs with coefficients in $\bQ$,
in this subsection we will consider graphs
whose coefficients are formal power series in $u_0,u_1,\cdots$,
i.e.,
we work on the space
\be
\cV_u:=
\prod_{\Gamma\in \mathfrak{Fat}^{ho}}
\bQ[[u_0,u_1,\cdots]]\cdot\Gamma.
\ee

Recall that in \S \ref{sec-hollow-ord} we have defined a family of operators
$\{\pd_n\}_{n\geq 1}$ acting on fat graphs.
Now let us extend the action of these operators
to the space $\cV_u$ in the following way:
\begin{itemize}
\item[1)]
For every $i\geq1$ and $j\geq 0$, we require:
\ben
\pd_i(u_j)=\delta_{i,j+1};
\een
\item[2)]
For every sequence $i_1,\cdots,i_n$ of non-negative integers
and a graph $\Gamma\in \mathfrak{Fat}^{ho}$,
we require the following Leibniz rule:
\begin{equation*}
\begin{split}
\pd_j(u_{i_1} u_{i_2}\cdots u_{i_n}\cdot\Gamma)=&
\sum_{k=1}^n u_{i_1} \cdots u_{i_{k-1}}
\pd_j(u_{i_k})
u_{i_{k+1}} \cdots u_{i_n}\cdot\Gamma\\
&+u_{i_1} u_{i_2}\cdots u_{i_n}\cdot
\pd_j(\Gamma).
\end{split}
\end{equation*}
\end{itemize}

We understand $u_j$ as the operator of multiplying by $u_j$
on the space $\cV_u$,
then we can easily see:

\begin{Lemma}
\label{lem-heisenberg}
The operators $\{\pd_i\}_{i\geq 1}$ and $\{u_j\}_{j\geq 0}$ satisfy
the following commutation relations:
\ben
[\pd_i,\pd_j]=0,
\qquad
[u_i,u_j]=0,
\qquad
[\pd_i,u_j]=\delta_{i,j+1}\cdot\Id,
\een
i.e.,
$\{\pd_i\}$ and $\{u_{j}\}$ generates a Heisenberg algebra
of operators on $\cV_u$.
\end{Lemma}

\subsection{Quantization of the special deformation}

In this subsection we discuss the quantum deformation theory of the
spectral curve.

For every integer $n>0$, let us denote:
\be
\beta_{-n}:=\frac{g_s^{-1}}{\sqrt{2}}(u_{n-1}-\delta_{n,2}),\qquad\quad
\beta_n:=n\sqrt{2}\cdot g_s \pd_n,
\ee
and let $\beta_0$ be the operator on $\cV_u$ defined by:
\be
\beta_0(u_{i_1}\cdots u_{i_n}\cdot\Gamma)=
\sqrt{2} g_s^{-1}\cdot
u_{i_1}\cdots u_{i_n}\cdot\bigg(
\pd_0\big(\Gamma_{dot}\big)\cdot\Gamma\bigg)
\ee
where $\pd_0(\Gamma_{dot})$ is the graph
consisting of a single hollow vertex of valence zero.
Then by Lemma \ref{lem-heisenberg}
we have the standard Heisenberg commutation relation:
\be
\label{eq-standard-Heisenberg}
[\beta_m,\beta_n]=m\delta_{m+n,0}.
\ee

Denote by $\hat{y}(z)$ the following series:
\be
\hat{y}(z):=\sum_{n\in \bZ}\beta_n\cdot z^{-n-1}.
\ee
Recall that the normal product $:\beta_{i_1}\cdots\beta_{i_n}:$
of bosons is defined by:
\ben
:\beta_{i_1}\cdots\beta_{i_n}:=\beta_{j_1}\cdots\beta_{j_n},
\een
where $(j_1,\cdots,j_n)$ is a permutation of $(i_i,\cdots,i_n)$
such that $j_1\leq j_2\leq\cdots\leq j_n$.
Then we have:
\be
:\half\hat{y}(z)^2:=\half\sum_{n\in \bZ}\sum_{k+l=n}:\beta_k\beta_l:z^{-n-2}
=\sum_{n\in \bZ}\wcL_n \cdot z^{-n-2},
\ee
where the operators $\wcL_n$ are given by:
\ben
\wcL_n=\half\sum_{k+l=n}:\beta_k\beta_l:z^{-n-2}.
\een
Then one can compute:
\be
\begin{split}
&\wcL_{-1}=g_s^{-2}u_0\gamma
-\pd_1
+\sum_{k=2}^\infty (k-1) u_{k-1}\pd_{k-1},
\\
&\wcL_0=g_s^{-2}\gamma^2
-2\pd_2
+\sum_{k=1}^\infty k u_{k-1}\pd _k,
\end{split}
\ee
and for $n\geq 1$,
\be
\wcL_n=2n\gamma\pd_n
-(n+2)\pd_{n+2}
+\sum_{k=1}^\infty (k+n)u_{k-1}\pd_{n+k}
+g_s^2\sum_{k=1}^{n-1}k(n-k)\pd_k\pd_{n-k},
\ee
where $\gamma$ is the operator of multiplying by $\pd_0(\Gamma_{dot})$.

Now let us consider the above operators $\{\wcL_n\}_{n\geq -1}$
as operators on fat graphs without hollow vertices.
Define a linear operator $\widehat\Phi$ by the following rules:
\ben
&&\widehat\Phi (u_0\gamma)=\fs_0\circ\gamma,
\qquad
\widehat\Phi (\gamma^2)=\fj\circ\gamma^2,
\qquad
\widehat\Phi (\gamma\pd_n)=\fj\circ(\gamma\pd_n),\\
&&\widehat\Phi ( \pd_n )=\pd_n,
\qquad
\widehat\Phi (\pd_m \pd_n )=\fj\circ(\pd_m\pd_n),
\qquad
\widehat\Phi (u_k \pd_n )=\fs_k\circ \pd_n.
\een
Then we have:

\begin{Proposition}
\label{prop-quantum-defo}
For every $n\geq -1$,
we have:
\be
\widehat\Phi (\wcL_n)=\cL_n
\ee
as operators on graphs without hollow vertices,
where $\{\cL_n\}_{n\geq -1}$ are the abstract Virasoro operators
(see Definition \ref{def-abs-Viropr}).
\end{Proposition}
\begin{proof}
The conclusion follows easily from
Lemma \ref{lem-specialdefo1} and Lemma \ref{lem-specialdefo2}.
\end{proof}

Using the standard Heisenberg commutation relation \eqref{eq-standard-Heisenberg}
one can derive the following OPE:
\be
\hat{y}(z)\hat{y} (w)=
:\hat{y}(z)\hat{y} (w):+\frac{1}{(w-z)^2},
\ee
thus we have:
\ben
\hat{y}(z+\epsilon)\hat{y} (z)=
:\hat{y}(z+\epsilon)\hat{y} (z):+\frac{1}{\epsilon^2}.
\een
Following \cite{zhou9},
let us define the regularized product
$\hat{y}(z)^{\odot 2}=\hat{y}(z) \odot \hat{y}(z)$ by:
\be
\hat{y}(z)^{\odot 2}:=\lim_{\epsilon\to 0}
\bigg(\hat{y}(z+\epsilon)\hat{y} (z)
-\frac{1}{\epsilon^2}\bigg),
\ee
then we have:
\begin{Theorem}
\label{thm-quantumdefo-abs}
The abstract Virasoro constraints $\cL_{n} (\cZ) =0$ ($n\geq -1$)
of the abstract QFT for fat graphs are equivalent to:
\be
\widehat\Phi\bigg(\big(\hat{y}(z)^{\odot 2}\big)_-\bigg) \cZ =0.
\ee
\end{Theorem}

\begin{proof}
By definition $\hat{y}(z)^{\odot 2}$ is:
\begin{equation*}
\begin{split}
\lim_{\epsilon\to 0}\bigg(
\hat{y}(z+\epsilon)\hat{y}(z)-\frac{1}{\epsilon^2}
\bigg)
=&\lim_{\epsilon \to 0}\big(
:\hat{y}(z+\epsilon) \hat{y}(z):
\big)\\
=&:\hat{y}(z)\hat{y}(z):\\
=&2\sum_{n\in \bZ} \wcL_n z^{-n-2}.
\end{split}
\end{equation*}
Thus the conclusion follows from Proposition \ref{prop-quantum-defo}.
\end{proof}

\subsection{Application: Realization by the quantum deformation theory
for Hermitian one-matrix models}

\label{sec-QDT-Herm}

In this subsection we consider the realization of the above quantum deformation theory
by the quantum deformation theory for
the spectral curve of Hermitian one-matrix models.

First let us consider the realization of the
special deformation \eqref{eq-specialdefo-absQFT} of the spectral curve.
We have already seen in previous sections that
the Hermitian one-matrix models provide us the most natural realization
of the abstract QFT for fat graphs,
if we choose the Feynman rule to be \eqref{eq-FR-1mm-absvirasoro}.
In this case
the weight of the graph $\pd_0(\Gamma_{dot})$
is simply $w_{(\pd_0\Gamma_{dot})}=t$,
and $\Phi(u_n,-)$ is actually multiplying by $g_{n+1}$.
Thus such a realization gives us a linear map
\ben
\mathcal{U}\times \cV^{ho-1}\to
\bC[[t,g_1,g_2,g_3,\cdots]]
\een
by mapping $u_n$ to $g_{n+1}$,
and $\Gamma$ to $w_{\Gamma}$.

Now recall that the special deformation \eqref{eq-specialdefo-absQFT}
is parametrized by elements in the space $\cU\times \cV^{ho-1}$,
thus applying the above linear map provides us
the following family of plane curves
parametrized by coupling constants $\{g_n\}_{n\geq 1}$
and the 't Hooft coupling constant $t=Ng_s$:
\be
y^{\text{Herm}}:=\frac{1}{\sqrt{2}}\sum_{n=0}^\infty
(g_{n+1}-\delta_{n,1})\cdot z^n+
\sqrt{2}\bigg(
t\cdot z^{-1}+
\sum_{n=1}^\infty n\frac{\pd F_{0}^{\text{Herm}}}{\pd g_n}\cdot z^{-n-1}\bigg),
\ee
where $F_0^{\text{Herm}}$ is the free energy of genus zero for the Hermitian one-matrix models
with respect to the fat genus expansion.
This family is exactly the fat special deformation of the spectral curve
for the Hermitian one-matrix models
(\cite[(78)]{zhou10}).
By taking $t_n=0$ for every $n\geq 1$
(or equivalently,
by applying the Feynman rule $w_v=w_e=1$ and $w_f=t$
to \eqref{eq-abstract-spectralcurve}),
we recover the fat spectral curve of the Hermitian one-matrix-models
(\cite[(97)]{zhou10}):
\be
\label{eq-1mm-fatspeccurve}
\begin{split}
y^{\text{Herm}} =&-\frac{z}{\sqrt{2}}+\sqrt{2}z\cdot \bigg(\frac{t}{z^2}+
\sum_{n>0}C_n\cdot\frac{t^{n+1}}{z^{2n+2}}\bigg)\\
=&-\frac{1}{\sqrt{2}}\cdot\sqrt{z^2-4t},
\end{split}
\ee
where $C_n=\frac{1}{n+1}\binom{2n}{n}$ is the $n$-th Catalan number
(here we have used the expression of one-point correlators
for Hermitian one-matrix models,
see Example \ref{eg-1mm-correlator}).
The fat Virasoro constraints of genus zero for the Hermitian one-matrix models
are encoded in the special deformation in the following way
(see \cite[Theorem 3.1]{zhou10}):

\begin{Theorem}
The fat Virasoro constraints of genus zero for the Hermitian one-matrix models
are equivalent to:
\ben
\bigg((y^{\text{Herm}})^2\bigg)_-=0.
\een

\end{Theorem}

This theorem is a straightforward consequence of Theorem \ref{thm-abstract-vira-genus0}.

Now let us consider the realization of the quantum deformation theory.
Again use the Feynman rule \eqref{eq-FR-1mm-absvirasoro},
the bosons $\{\beta_{n}\}_{n\in \bZ}$
can be realized by some multiplications and partial derivatives:
\begin{itemize}
\item[1)]
The bosonic creator $\beta_{-n}$ ($n>0$) is realized by
$\tilde\beta_{-n}:=\frac{g_s^{-1}}{\sqrt{2}}(g_n-\delta_{n,2})$;
\item[2)]
The bosonic annihilator $\beta_{n}$ ($n>0$) is realized by
$\tilde\beta_{n}:=n\sqrt{2}\cdot g_s \frac{\pd}{\pd g_n}$;
\item[3)]
The operator $\beta_0$ is realized by $\tilde\beta_{0}:=\sqrt{2}g_s^{-1}t$.
\end{itemize}
The operators $\{\tilde\beta_n\}_{n\in\bZ}$ generates a Heisenberg algebra
acting on the space of formal power series $\bC[[g_s,g_s^{-1},t,g_1,g_2,\cdots]]$:
\be
[\tilde\beta_m,\tilde\beta_n]=m\delta_{m+n,0},
\ee
and it is clear that the realization of the operator $\widehat\Phi$ is simply
the composition of these bosons.
Now the bosonic field $\hat{y}(z)$ is realized by:
\be
\hat{y}^{\text{Herm}}(z):=\sum_{n\in \bZ}\tilde\beta_n\cdot z^{-n-1}.
\ee
Using this field $\hat{y}^{\text{Herm}}(z)$,
The fat Virasoro constraints \eqref{eq-fatVirasoro}
for the Hermitian one-matrix models can be
converted into the same form as Theorem \ref{thm-quantumdefo-abs}.
Define the regularized product for $\hat{y}^{\text{Herm}}(z)$ to be:
\be
\hat{y}^{\text{Herm}}(z)^{\odot 2}:=\lim_{\epsilon\to 0}
\bigg(\hat{y}^{\text{Herm}}(z+\epsilon)\hat{y}^{\text{Herm}} (z)
-\frac{1}{\epsilon^2}\bigg),
\ee
then Theorem \ref{thm-quantumdefo-abs} gives us:

\begin{Theorem}
The fat Virasoro constraints for the Hermitian one-matrix models
are equivalent to:
\be
\bigg(\big(\hat{y}^{\text{Herm}}(z)^{\odot 2}\big)_-\bigg) Z_N^{\text{Herm}} =0,
\ee
where $Z_N^{\text{Herm}}$ is the partition function of the Hermitian one-matrix models.
\end{Theorem}

\section{Abstract $N$-Point Functions and Quadratic Recursions}
\label{sec-abstract-eorec}

In this section,
we define the abstract $n$-point functions $\cW_{g,n}$
to be the `generating series' of fat graphs of genus $g$ with $n$ vertices,
and derive a quadratic recursion for $\cW_{g,n}$.
This recursion looks similar to the Eynard-Orantin topological recursion,
and we will show in next sections that in some examples
the realizations of this recursion are equivalent to the E-O topological recursions.
We give a brief review of the E-O topological recursion in the first subsection,
and derive the quadratic recursion for the abstract $n$-point functions $\cW_{g,n}$
in the remaining subsections.
Here we need to restore the labels $v_1,\cdots,v_n$ on vertices of fat graphs,
and modify the vertex-splitting operators such that they may indicate
the information of these labels.

\subsection{Preliminaries of the Eynard-Orantin topological recursion
and quantum spectral curves}
\label{sec-eo}

In the present subsection,
let us recall some preliminaries of the Eyanrd-Orantin topological recursion
\cite{ce, eo} and Gukov-Su{\l}kowski's construction of the quantum spectral curves.

The input data of the E-O topological recursion
consist of a spectral curve together a Bergmann kernel on it.
A spectral curve $\cC$ is a parametrized curve
\be
x=x(z),\qquad\qquad y=y(z)
\ee
where $x$ and $y$ are two meromorphic functions on $\cC$,
and a Bergmann kernel is a $2$-differential
$B(z_1,z_2)$ on this curve with the following property:
\be
B(z,z')\sim\bigg(
\frac{1}{(z-z')^2}+O(1)
\bigg)dzdz'.
\ee

Assume that $x$ has only nondegenerate critical points
$a_1,\cdots,a_n$,
i.e.,
\be
x'(a_i)=0,\qquad\qquad x''(a_i)\not=0,
\ee
for every $i=1,\cdots,n$.
Then near each branch point $a_i$,
there exists locally a nontrivial involution $\sigma$
such that $x(\sigma(z))=x(z)$.
One can introduce a new coordinate $\zeta_i$ near $a_i$
in the following way:
\be
x=x(a_i)+\zeta_i^2.
\ee
This local coordinate is called the `local Airy coordinate'.
Then near $a_i$ the involution $\sigma$ is given by:
\be
\zeta(\sigma(z))=-\zeta(z).
\ee

Eynard and Orantin \cite{eo} constructed a family of
multi-linear differentials $\omega_{g,n}$
($g\geq 0, n\geq 1$)on the spectral curve $\cC$ in the following way.
First define the initial data $\omega_{0,1}$ and $\omega_{0,2}$ to be:
\be
\begin{split}
&\omega_{0,1}(z):=
y(z)dx(z),\\
&\omega_{0,2}(z_1,z_2):=
B(z_1,z_2).
\end{split}
\ee
And when $2g-1+n>0$,
the differential $\omega_{g,n+1}$ is defined recursively by:
\be
\label{eq-eorec}
\begin{split}
\omega_{g,n+1}(z_0,z_1,\cdots,z_n):=&
\sum_{i=1}^n\Res_{z\to a_i}K(z_0,z)\bigg[
\omega_{g-1,n+2}(z,\sigma(z),z_1,\cdots,z_n)\\
&+\sum_{\substack{g_1+g_2=g\\I\sqcup J=[n]}}^s
\omega_{g_1,|I|+1}(z,z_I)
\cdot
\omega_{g_2,|J|+1}(\sigma(z),z_J)
\bigg],
\end{split}
\ee
where the recursion kernel $K(z_0,z)$ is defined
locally near each branch point $a_i$
as follows:
\be
K(z_0,z):=
\frac{\int_{\sigma(z)}^z B(z_0,z)}
{2\big(y(z)-y(\sigma(z))\big)dx(z)}.
\ee

In \cite{ey1},
Eynard showed that after the Laplace transformation,
the E-O invariants $\omega_{g,n}^{EO}$
for a spectral curve with one branch point can be related to
intersection numbers on $\Mbar_{g,n}$.
This result was generalized to the case of spectral curves
with several branch points by Eynard in \cite{ey2},
where the E-O invariants can be related to intersection numbers
on the moduli space $\Mbar_{g,n}^{\mathfrak{b}}$ of
`colored' Riemann surfaces.

Gukov and Su{\l}kowski have proposed a method to
construct the quantum spectral curves using
Eynard-Orantin topological recursion in their work \cite{gs}, .
Let $\bC\times\bC$ be a complex plane with coordinates $(u,v)$
equipped with the symplectic form $\omega=\frac{\sqrt{-1}}{\hbar}du\wedge dv$.
Let $A(u,v)$ be a polynomial in $u$ and $v$,
then the curve
\ben
\mathcal{C}:\quad A(u,v)=0
\een
is a Lagrangian submanifold of $(\bC\times\bC,\omega)$.
A quantization procedure is supposed to turn the coordinates
$u$ and $v$ into operators $\hat u$ and $\hat v$ respectively
which satisfy the commutation relation $[\hat v,\hat u]=\hbar$,
and the algebra of functions in $u,v$ into a noncommutative algebra of operators.
The quantization of the polynomial $A(u,v)$ is an operator
\be
\widehat A
=\widehat A_0 +\hbar \widehat A_1 +\hbar^2 \widehat A_2 +\cdots,
\ee
where $\widehat A_0$ is obtained from $A$ by replacing $u,v$ by $\hat u,\hat v$ respectively.
Inspired by the matrix model origin of the E-O topological recursion,
Gukov and Su{\l}kowski defined the following Baker-Akhiezer function
(see \cite[\S 2]{gs}):
\be
Z(z):=\exp\biggl(\sum_{n=0}^{\infty}\hbar^{n-1}S_n(z)\biggr),
\ee
where
\ben
&&S_0(z):=\int^z v(z)du(z),\\
&&S_1(z):=-\half \log \frac{du}{dz},\\
&&S_n(z):=\sum_{2g-1+k=n}\frac{1}{k!}\int^z\cdots\int^z
\omega_{g,k}(z_1,\cdots,z_k),
\quad n\geq 2,
\een
and $\omega_{g,k}(z_1,\cdots,z_k)$ are the E-O invariants defined by
the spectral curve $u=u(z),v=v(z)$
together with a choice of Bergman kernel $B(p,q)$.
Then they conjectured that the quantum spectral curve $\widehat A$ associated to the curve $A$
can be obtained by solving the Schr\"odinger equation:
\be
\label{eq-Schrodinger}
\widehat A Z(z) =0.
\ee

\subsection{Abstract $n$-point functions $\cW_{g,n}$ and abstract Bergmann kernel}
\label{sec-abstract-npt}

In this subsection let us define the abstract $n$-point functions $\cW_{g,n}$.

Before given the specific definition of $\cW_{g,n}$,
we should first emphasize that we are facing a different situation
with \S \ref{sec-fat-abs-virasoro}.
As pointed out in the beginning of \S \ref{sec-abstractpartition}
and the beginning of \S \ref{sec-fat-abs-virasoro},
when we talk about the abstract free energies $\cF$ and
the abstract partition function $ \cZ $,
we always forget the labels $v_1,v_2,\cdots,v_n$ on vertices of a fat graph.
But from now on,
we need to retain these labels whenever we talk about $\cW_{g,n}$.

\begin{Definition}
For every $g\geq 0$ and $n\geq 1$,
define the abstract $n$-point function $\cW_{g,n}$ for the abstract QFT to be
the following formal infinite summation of fat graphs with labels on vertices:
\be
\label{eq-abstract-npt}
\cW_{g,n}:=
\delta_{g,0}\delta_{n,1}\cF_0^{(0)}
+\sum_{\mu\in \bZ_{>0}^n}\mu_1\mu_2\cdots\mu_n\cF_g^{\mu},
\ee
where the convention $\cF_0^{(0)}$ stands for
a graphs consisting of one single vertex and no edges
(see Definition \ref{def-fat-abs-cor}).
The formal summation $\cW_{g,n}$ is an element in the following vector space:
\be
\prod_{\mu\in \bZ_{>0}^n}
\prod_{\Gamma\in\mathfrak{Fat}_{g}^{\mu,c}}
\bQ\Gamma
\ee
for $(g,n)\not= (0,1)$.
The special one $\cW_{0,1}$ belongs to:
\be
\bQ\cdot \cF_0^{(0)}\times
\prod_{\mu\in \bZ_{>0}}
\prod_{\Gamma\in\mathfrak{Fat}_{0}^{(\mu),c}}
\bQ\Gamma.
\ee

\end{Definition}

The special cases $(g,n)=(0,1)$ and $(0,2)$ are of particular importance.
In \S \ref{sec-QDT} have already seen that the coefficients of the one-point function $\cW_{0,1}$
of genus zero are given by the Catalan numbers,
from which a spectral curve emerges naturally.
The case $(g,n)=(0,2)$ is also interesting,
in some known examples of realizations,
the realization of the two-point function $\cW_{0,2}$ of genus zero plays the role
of the Bergmann kernel in the E-O topological recursion.

\begin{Definition}
We call the abstract two point function
\be
\cW_{0,2}:=
\sum_{\mu_1,\mu_2>0}\mu_1\mu_2\cF_0^{(\mu_1,\mu_2)}
=\sum_{\mu_1,\mu_2>0} \sum_{\Gamma\in \mathfrak{Fat}_g^{(\mu_1,\mu_2),c}}
\frac{\mu_1\mu_2}{|\Aut(\Gamma)|}\Gamma
\ee
of genus zero
the abstract Bergmann kernel.
\end{Definition}

Similar to the case in \S \ref{sec-abs-qrec},
in order to derive quadratic recursions for
the abstract $n$-point functions $\cW_{g,n}$,
we also need to introduce some conventions on relabellings of fat graphs.
Given a finite set of indices $I=\{i_1,i_2,\cdots,i_n\}\subset \bZ_{>0}$,
we define the relabelled abstract $n$-point function $\cW_{g,I}$ to be:
\be
\label{eq-abstractnpt-relabel}
\cW_{g,I}:=
\delta_{g,0}\delta_{n,1}\cF_{0,I}^{(0)}
+\sum_{\mu\in \bZ_{>0}^n}\mu_1\mu_2\cdots\mu_n\cF_{g,I}^{\mu},
\ee
where the relabelled correlators $\cF_{g,I}^{\mu}$ are obtained from $\cF_{g,n}^\mu$ by
replacing the labels $v_1,\cdots,v_n$ by $v_{i_1},\cdots,v_{i_n}$ respectively
(see \S \ref{sec-abs-qrec}).

Theorem \ref{thm-abstract-rec} provides a recursion formula
for $K_1 \cW_{g,n}$
by taking summation over all $\mu\in \bZ_{>0}^n$ in \eqref{eq-abstract-rec}.
In the rest of this subsection let us write down this recursion.
First let us introduce two new families of operators:
\ben
\sigma_k: \prod_{g,\mu}\bQ \Gamma
\to \prod_{g,\mu}\bQ \Gamma,
\qquad
\tau_k: \prod_{g,\mu}\bQ \Gamma
\to \prod_{g,\mu}\bQ \Gamma
\een
for $k\geq 1$ as follows.
Given a fat graph $\Gamma$,
we define:
\be
\label{eq-operator-sigma}
\sigma_k(\Gamma):=\val(v_k)\cdot \Gamma,
\qquad \forall 1\leq k\leq n,
\ee
where $\val(v_k)$ is the valence of the vertex labelled by $v_k$;
and we require the result to be zero if $k>n$.
Moreover, let us define:
\be
\begin{split}
\tau_k(\Gamma):=&\sum_{\substack{l+m=\val(v_k)+2\\l,m\geq 1}} l\cdot m\cdot \Gamma\\
=&\frac{1}{6}\big(\val(v_k)^3+6\val(v_k)^2+11\val(v_k)+6\big)\cdot \Gamma
\end{split}
\ee
for every $1\leq k\leq n$,
and require the result to be zero if $k>n$.
Then we have:

\begin{Theorem}
\label{thm-abstract-eotype}
The abstract $n$-point functions $\cW_{g,n}$
satisfies the following quadratic recursion relation:
\be
\label{eq-abstract-eotype}
\begin{split}
K_1(\cW_{g,n})=&\sum_{j=2}^n
\tau_1(\cW_{g,n-1})
+
(\sigma_1+\sigma_2+2)\bigg(
\cW_{g-1,n+1}\\
&\qquad\qquad+\sum_{\substack{g_1+g_2=g\\I\sqcup J=[n+1]\backslash\{1,2\}}}
\cW_{g_1,\{1\}\sqcup I}\cdot\cW_{g_2,\{2\}\sqcup J}
\bigg).\\
\end{split}
\ee

\end{Theorem}

\begin{proof}

Let us apply the edge-contraction operator $K_1$ on $\cW_{g,n}$,
then the quadratic recursion \eqref{eq-abstract-rec} gives us:
\begin{equation*}
\begin{split}
&K_1 (\cW_{g,n})\\
=&\sum_{\mu\in\bZ_{>0}^n}
\mu_1\cdots\mu_n\cdot K_1 \cF_g^\mu\\
=&\sum_{\mu}\mu_1\cdots\mu_n\bigg(
\sum_{j=2}^n(\mu_1+\mu_j-2)\cF_g^{(\mu_1+\mu_j-2,\mu_{[n]\backslash\{1,j\}})}
\bigg)\\
+&\sum_{\mu}\sum_{\substack{\alpha+\beta=\mu_1-2\\\alpha\geq 1,\beta\geq 1}}
\mu_1\cdots\mu_n\alpha\beta\bigg(
\cF_{g-1}^{(\alpha,\beta,\mu_{[n]\backslash\{1\}})}+
\sum_{\substack{g_1+g_2=g\\I\sqcup J=[n]\backslash\{1\}}}
\cF_{g_1,\{1\}\sqcup (I+1)}^{(\alpha,\mu_I)}\cF_{g_2,\{2\}\sqcup (J+1)}^{(\beta,\mu_J)}
\bigg)\\
+&\sum_\mu \mu_1\cdots\mu_n\bigg(
(\mu_1-2)\cdot
\cF_{0,\{1\}}^{(0)}\cF_{g,[n+1]\backslash\{1\}}^{(\mu_1-2,\mu_{[n]\backslash\{1\}})}
+ (\mu_1-2)\cdot
\cF_{0,\{2\}}^{(0)}\cF_{g,[n+1]\backslash\{2\}}^{(\mu_1-2,\mu_{[n]\backslash\{1\}})}\\
&\qquad\qquad\qquad
+\delta_{n,1}\delta_{g,0}\delta_{\mu_1,2}(\cF_0^{(0)})^2
\bigg).
\end{split}
\end{equation*}
Here we have:
\begin{equation*}
\begin{split}
&\sum_{\mu}\mu_1\cdots\mu_n\bigg(
\sum_{j=2}^n(\mu_1+\mu_j-2)\cF_g^{(\mu_1+\mu_j-2,\mu_{[n]\backslash\{1,j\}})}
\bigg)\\
=&\sum_{j=2}^n\sum_{\substack{\mu_0\\\mu_{[n]\backslash\{1,j\}}}}
\sum_{\mu_1+\mu_j-2=\mu_0}\bigg[
\mu_1\mu_j\cdot\bigg(
\mu_0\cdot\prod_{i\in [n]\backslash\{1,j\}}\mu_i\bigg)
\cF_g^{(\mu_0,\mu_{[n]\backslash\{1,j\}})}\bigg]\\
=&\sum_{j=2}^n
\tau_1(\cW_{g,n-1}),
\end{split}
\end{equation*}
moreover,
\begin{equation*}
\begin{split}
&\sum_{\mu}\sum_{\substack{\alpha+\beta=\mu_1-2\\\alpha\geq 1,\beta\geq 1}}
\mu_1\cdots\mu_n\alpha\beta\bigg(
\cF_{g-1}^{(\alpha,\beta,\mu_{[n]\backslash\{1\}})}+
\sum_{\substack{g_1+g_2=g\\I\sqcup J=[n]\backslash\{1\}}}
\cF_{g_1,\{1\}\sqcup (I+1)}^{(\alpha,\mu_I)}\cF_{g_2,\{2\}\sqcup (J+1)}^{(\beta,\mu_J)}
\bigg)\\
&+\sum_\mu \mu_1\cdots\mu_n\bigg(
(\mu_1-2)\cdot
\cF_{0,\{1\}}^{(0)}\cF_{g,[n+1]\backslash\{1\}}^{(\mu_1-2,\mu_{[n]\backslash\{1\}})}
+ (\mu_1-2)\cdot
\cF_{0,\{2\}}^{(0)}\cF_{g,[n+1]\backslash\{2\}}^{(\mu_1-2,\mu_{[n]\backslash\{1\}})}\\
&\qquad\qquad\qquad\quad
+\delta_{n,1}\delta_{g,0}\delta_{\mu_1,2}(\cF_0^{(0)})^2
\bigg)\\
&=\sum_{\substack{\alpha,\beta\geq 1\\\mu_2,\cdots,\mu_n}}
(\alpha+\beta+2)\cdot\alpha\beta\mu_2\cdots\mu_n\bigg(
\cF_{g-1}^{(\alpha,\beta,\mu_{[n]\backslash\{1\}})}
+\\
&\quad
\sum_{\substack{g_1+g_2=g\\I\sqcup J=[n]\backslash\{1\}}}
\cF_{g_1,\{1\}\sqcup (I+1)}^{(\alpha,\mu_I)}\cF_{g_2,\{2\}\sqcup (J+1)}^{(\beta,\mu_J)}
\bigg)
+\sum_\mu \mu_1\cdots\mu_n\bigg(
\delta_{n,1}\delta_{g,0}\delta_{\mu_1,2}(\cF_0^{(0)})^2+
\\
&\quad
(\mu_1-2)\cdot
\cF_{0,\{1\}}^{(0)}\cF_{g,[n+1]\backslash\{1\}}^{(\mu_1-2,\mu_{[n]\backslash\{1\}})}
+ (\mu_1-2)\cdot
\cF_{0,\{2\}}^{(0)}\cF_{g,[n+1]\backslash\{2\}}^{(\mu_1-2,\mu_{[n]\backslash\{1\}})}
\bigg)\\
&=(\sigma_1+\sigma_2+2)\bigg(
\cW_{g-1,n+1}
+\sum_{\substack{g_1+g_2=g\\I\sqcup J=[n+1]\backslash\{1,2\}}}
\cW_{g_1,\{1\}\sqcup I}\cW_{g_2,\{2\}\sqcup J}
\bigg).
\end{split}
\end{equation*}
Thus the conclusion holds.
\end{proof}

In \S \ref{sec-qrec-npt} we will find another quadratic recursion for $\cW_{g,n}$,
since the above recursion is not satisfactory for our purpose.
In the above recursion,
the summation on the right-hand side involves terms of type $(g,n)=(0,1)$,
and this differs from the Eynard-Orantin topological recursion.
Moreover,
what we really want is a quadratic recursion for $\cW_{g,n}$ itself,
but not for $K_1 \cW_{g,n}$.

In order to obtain such a good recursion,
we need to modify the vertex-splitting operators
to restore the labels on vertices,
and use these new vertex-splitting operators to reformulate the above recursion.
We will do this in the following subsections.

\subsection{Vertex-splitting operators with labels and a quadratic recursion}
\label{sec-vertexsplitting-label}

In this subsection,
let us introduce the vertex-splitting operators with labels,
and reformulate the recursion relation
\eqref{eq-abstract-eotype} using these operators.

We will denote by $\cS_{\{1;j\}}$ and $\cJ_{\{1,2\}}$
the new type of operators we are going to define.
They are similar to the operators $\cS_{k,l}$ and $\cJ_{k,l}$
discussed in \S \ref{sec-fat-abs-virasoro},
but notice that the indices $k,l$ in $\cS_{k,l}$ and $\cJ_{k,l}$
mean that we operate on all vertices of valences $k$ and $l$ respectively;
and the indices $1,2,j$ in $\cS_{\{1;j\}}$ and $\cJ_{\{1,2\}}$
mean that we operate on the vertices labelled by $v_1$, $v_2$, and $v_j$
respectively.
In this section we do not need hollow vertices anymore.

The operator $S_{\{1;j\}}$ is defined as follows.
Let $\Gamma$ be a fat graph with $n$ vertices labelled by $v_1,v_2,\cdots, v_n$
(where $n\geq j-1$),
and choose a subset $H\subset H(v_1)$ of half-edges incident at $v_1$,
such that the half-edges in $H$ are adjacent according to the cyclic order on $v_1$.
Then we split the vertex $v_1$ into two different vertices
with a new internal edge connecting them,
such that the half-edges in $H$ and $H(v_1)\backslash H$ are
attached to the two new vertices respectively,
and the cyclic orders are compatible with the original cyclic order at $v_1$.
Moreover,
we relabel $v_2,v_3,\cdots,\hat{v}_j,\cdots,v_{n+1}$
on the original vertices $v_2,v_3,\cdots,v_{n}$ respectively;
then we label $v_1$ on the vertex which $H$ attached to,
and label $v_j$ on the vertex which $H(v_1)\backslash H$ attached to.
Finally,
we take summation over all possible subsets $H$ and
all possible ways to split the vertex.
Notice that here we allow the subsect $H$ to be empty or $H(v_1)$.

\begin{Example}
We give an example of this new vertex-splitting operator
(here for simplicity we only show its action on the vertex $v_1$):
\begin{flalign*}
\begin{split}
&\begin{tikzpicture}[scale=1.125]
\node [align=center,align=center] at (0.2+0.2,0) {$\cS_{\{1;j\}}\bigg($};
\draw [fill] (2-0.3+0.15+0.1,0) circle [radius=0.06];
\node [align=center,align=center] at (2-0.3+0.15+0.1,-0.35) {$v_1$};
\draw [thick] (0+2-0.3+0.15+0.1,0) -- (-0.5+2-0.3+0.15+0.1,0.5);
\draw [thick] (0+2-0.3+0.15+0.1,0) -- (-0.5+2-0.3+0.15+0.1,-0.5);
\draw [thick] (1.7+0.15+0.1,0) -- (2.35+0.15+0.1,0);
\node [align=center,left] at (-0.5+2-0.25+0.3,0.5) {$h_1$};
\node [align=center,left] at (-0.5+2-0.25+0.3,-0.5) {$h_2$};
\node [align=center,right] at (2.3+0.2,0) {$h_3$};
\draw [->,>=stealth] (-0.2+1.7+0.15+0.1,-0.1) .. controls (0.3+1.7+0.15+0.1,-0.5) and (0.3+1.7+0.15+0.1,0.5) ..  (-0.2+1.7+0.15+0.1,0.1);
\node [align=center,align=center] at (3.35,0) {$\bigg)=$};
\draw [fill] (4,0) circle [radius=0.06];
\draw [fill] (4.6,0) circle [radius=0.06];
\draw [thick] (4,0) -- (4.6,0);
\draw [->,>=stealth] (4.4,-0.1) .. controls (4.9,-0.5) and (4.9,0.5) ..  (4.4,0.1);
\draw [thick] (5.2,0) -- (4.6,0);
\draw [thick] (5.2,0.5) -- (4.6,0);
\draw [thick] (5.2,-0.5) -- (4.6,0);
\node [align=center,align=center] at (4,-0.35) {$v_1$};
\node [align=center,align=center] at (4.6,-0.35) {$v_j$};
\node [align=center,align=center] at (5.4,-0.5) {$h_1$};
\node [align=center,align=center] at (5.4,0) {$h_2$};
\node [align=center,align=center] at (5.4,0.5) {$h_3$};
\node [align=center,align=center] at (5.95,0) {$+$};
\draw [fill] (6.4,0) circle [radius=0.06];
\draw [fill] (4.6+2.4,0) circle [radius=0.06];
\draw [thick] (4+2.4,0) -- (4.6+2.4,0);
\draw [->,>=stealth] (4.4+2.4,-0.1) .. controls (4.9+2.4,-0.5) and (4.9+2.4,0.5) ..  (4.4+2.4,0.1);
\draw [thick] (5.2+2.4,0) -- (4.6+2.4,0);
\draw [thick] (5.2+2.4,0.5) -- (4.6+2.4,0);
\draw [thick] (5.2+2.4,-0.5) -- (4.6+2.4,0);
\node [align=center,align=center] at (4+2.4,-0.35) {$v_1$};
\node [align=center,align=center] at (4.6+2.4,-0.35) {$v_j$};
\node [align=center,align=center] at (5.4+2.4,-0.5) {$h_2$};
\node [align=center,align=center] at (5.4+2.4,0) {$h_3$};
\node [align=center,align=center] at (5.4+2.4,0.5) {$h_1$};
\node [align=center,align=center] at (5.95+2.4,0) {$+$};
\draw [fill] (6.4+2.4,0) circle [radius=0.06];
\draw [fill] (4.6+2.4+2.4,0) circle [radius=0.06];
\draw [thick] (4+2.4+2.4,0) -- (4.6+2.4+2.4,0);
\draw [->,>=stealth] (4.4+2.4+2.4,-0.1) .. controls (4.9+2.4+2.4,-0.5) and (4.9+2.4+2.4,0.5) ..  (4.4+2.4+2.4,0.1);
\draw [thick] (5.2+2.4+2.4,0) -- (4.6+2.4+2.4,0);
\draw [thick] (5.2+2.4+2.4,0.5) -- (4.6+2.4+2.4,0);
\draw [thick] (5.2+2.4+2.4,-0.5) -- (4.6+2.4+2.4,0);
\node [align=center,align=center] at (4+2.4+2.4,-0.35) {$v_1$};
\node [align=center,align=center] at (4.6+2.4+2.4,-0.35) {$v_j$};
\node [align=center,align=center] at (5.4+2.4+2.4,-0.5) {$h_3$};
\node [align=center,align=center] at (5.4+2.4+2.4,0) {$h_1$};
\node [align=center,align=center] at (5.4+2.4+2.4,0.5) {$h_2$};
\end{tikzpicture}\\
&\qquad\qquad\begin{tikzpicture}[scale=1.125]
\node [align=center,align=center] at (-0.65,0) {$+$};
\draw [fill] (0.6,0) circle [radius=0.06];
\draw [fill] (1.1,0) circle [radius=0.06];
\draw [thick] (0.1,0) -- (1.1,0);
\draw [thick] (1.6,0.35) -- (1.1,0);
\draw [thick] (1.6,-0.35) -- (1.1,0);
\node [align=center,align=center] at (0.6,-0.35) {$v_1$};
\node [align=center,align=center] at (1.1,-0.35) {$v_j$};
\node [align=center,align=center] at (1.8,0.35) {$h_3$};
\node [align=center,align=center] at (1.8,-0.35) {$h_2$};
\node [align=center,align=center] at (-0.1,0) {$h_1$};
\draw [->,>=stealth] (0.9,-0.1) .. controls (1.4,-0.5) and (1.4,0.5) ..  (0.9,0.1);
\node [align=center,align=center] at (-0.65+3,0) {$+$};
\draw [fill] (0.6+3,0) circle [radius=0.06];
\draw [fill] (1.1+3,0) circle [radius=0.06];
\draw [thick] (0.1+3,0) -- (1.1+3,0);
\draw [thick] (1.6+3,0.35) -- (1.1+3,0);
\draw [thick] (1.6+3,-0.35) -- (1.1+3,0);
\node [align=center,align=center] at (0.6+3,-0.35) {$v_1$};
\node [align=center,align=center] at (1.1+3,-0.35) {$v_j$};
\node [align=center,align=center] at (1.8+3,0.35) {$h_1$};
\node [align=center,align=center] at (1.8+3,-0.35) {$h_3$};
\node [align=center,align=center] at (-0.1+3,0) {$h_2$};
\draw [->,>=stealth] (0.9+3,-0.1) .. controls (1.4+3,-0.5) and (1.4+3,0.5) ..  (0.9+3,0.1);
\node [align=center,align=center] at (-0.65+6,0) {$+$};
\draw [fill] (0.6+6,0) circle [radius=0.06];
\draw [fill] (1.1+6,0) circle [radius=0.06];
\draw [thick] (0.1+6,0) -- (1.1+6,0);
\draw [thick] (1.6+6,0.35) -- (1.1+6,0);
\draw [thick] (1.6+6,-0.35) -- (1.1+6,0);
\node [align=center,align=center] at (0.6+6,-0.35) {$v_1$};
\node [align=center,align=center] at (1.1+6,-0.35) {$v_j$};
\node [align=center,align=center] at (1.8+6,0.35) {$h_2$};
\node [align=center,align=center] at (1.8+6,-0.35) {$h_1$};
\node [align=center,align=center] at (-0.1+6,0) {$h_3$};
\draw [->,>=stealth] (0.9+6,-0.1) .. controls (1.4+6,-0.5) and (1.4+6,0.5) ..  (0.9+6,0.1);
\end{tikzpicture}\\
&\qquad\qquad\begin{tikzpicture}[scale=1.125]
\node [align=center,align=center] at (0.65-3,0) {$+$};
\draw [fill] (-0.6,0) circle [radius=0.06];
\draw [fill] (-1.1,0) circle [radius=0.06];
\draw [thick] (-0.1,0) -- (-1.1,0);
\draw [thick] (-1.6,0.35) -- (-1.1,0);
\draw [thick] (-1.6,-0.35) -- (-1.1,0);
\node [align=center,align=center] at (-0.6,-0.35) {$v_j$};
\node [align=center,align=center] at (-1.1,-0.35) {$v_1$};
\node [align=center,align=center] at (-1.8,0.35) {$h_2$};
\node [align=center,align=center] at (-1.8,-0.35) {$h_3$};
\node [align=center,align=center] at (0.1,0) {$h_1$};
\draw [->,>=stealth] (-1.3,-0.1) .. controls (-0.8,-0.5) and (-0.8,0.5) ..  (-1.3,0.1);
\node [align=center,align=center] at (0.65,0) {$+$};
\draw [fill] (-0.6+3,0) circle [radius=0.06];
\draw [fill] (-1.1+3,0) circle [radius=0.06];
\draw [thick] (-0.1+3,0) -- (-1.1+3,0);
\draw [thick] (-1.6+3,0.35) -- (-1.1+3,0);
\draw [thick] (-1.6+3,-0.35) -- (-1.1+3,0);
\node [align=center,align=center] at (-0.6+3,-0.35) {$v_j$};
\node [align=center,align=center] at (-1.1+3,-0.35) {$v_1$};
\node [align=center,align=center] at (-1.8+3,0.35) {$h_3$};
\node [align=center,align=center] at (-1.8+3,-0.35) {$h_1$};
\node [align=center,align=center] at (0.1+3,0) {$h_2$};
\draw [->,>=stealth] (-1.3+3,-0.1) .. controls (-0.8+3,-0.5) and (-0.8+3,0.5) ..  (-1.3+3,0.1);
\node [align=center,align=center] at (0.65+3,0) {$+$};
\draw [fill] (-0.6+3+3,0) circle [radius=0.06];
\draw [fill] (-1.1+3+3,0) circle [radius=0.06];
\draw [thick] (-0.1+3+3,0) -- (-1.1+3+3,0);
\draw [thick] (-1.6+3+3,0.35) -- (-1.1+3+3,0);
\draw [thick] (-1.6+3+3,-0.35) -- (-1.1+3+3,0);
\node [align=center,align=center] at (-0.6+3+3,-0.35) {$v_j$};
\node [align=center,align=center] at (-1.1+3+3,-0.35) {$v_1$};
\node [align=center,align=center] at (-1.8+3+3,0.35) {$h_1$};
\node [align=center,align=center] at (-1.8+3+3,-0.35) {$h_2$};
\node [align=center,align=center] at (0.1+3+3,0) {$h_3$};
\draw [->,>=stealth] (-1.3+3+3,-0.1) .. controls (-0.8+3+3,-0.5) and (-0.8+3+3,0.5) ..  (-1.3+3+3,0.1);
\end{tikzpicture}\\
&\qquad\qquad\begin{tikzpicture}[scale=1.125]
\node [align=center,align=center] at (2.65,0) {$+$};
\draw [fill] (4,0) circle [radius=0.06];
\draw [fill] (4.6,0) circle [radius=0.06];
\draw [thick] (4,0) -- (4.6,0);
\draw [->,>=stealth] (3.8,-0.1) .. controls (4.3,-0.5) and (4.3,0.5) ..  (3.8,0.1);
\draw [thick] (3.4,0) -- (4,0);
\draw [thick] (3.4,0.5) -- (4,0);
\draw [thick] (3.4,-0.5) -- (4,0);
\node [align=center,align=center] at (4,-0.35) {$v_1$};
\node [align=center,align=center] at (4.6,-0.35) {$v_j$};
\node [align=center,align=center] at (3.2,-0.5) {$h_3$};
\node [align=center,align=center] at (3.2,0) {$h_2$};
\node [align=center,align=center] at (3.2,0.5) {$h_1$};
\node [align=center,align=center] at (5.05,0) {$+$};
\draw [fill] (4+2.4,0) circle [radius=0.06];
\draw [fill] (4.6+2.4,0) circle [radius=0.06];
\draw [thick] (4+2.4,0) -- (4.6+2.4,0);
\draw [->,>=stealth] (3.8+2.4,-0.1) .. controls (4.3+2.4,-0.5) and (4.3+2.4,0.5) ..  (3.8+2.4,0.1);
\draw [thick] (3.4+2.4,0) -- (4+2.4,0);
\draw [thick] (3.4+2.4,0.5) -- (4+2.4,0);
\draw [thick] (3.4+2.4,-0.5) -- (4+2.4,0);
\node [align=center,align=center] at (4+2.4,-0.35) {$v_1$};
\node [align=center,align=center] at (4.6+2.4,-0.35) {$v_j$};
\node [align=center,align=center] at (3.2+2.4,-0.5) {$h_1$};
\node [align=center,align=center] at (3.2+2.4,0) {$h_3$};
\node [align=center,align=center] at (3.2+2.4,0.5) {$h_2$};
\node [align=center,align=center] at (5.05+2.4,0) {$+$};
\draw [fill] (4+2.4+2.4,0) circle [radius=0.06];
\draw [fill] (4.6+2.4+2.4,0) circle [radius=0.06];
\draw [thick] (4+2.4+2.4,0) -- (4.6+2.4+2.4,0);
\draw [->,>=stealth] (3.8+2.4+2.4,-0.1) .. controls (4.3+2.4+2.4,-0.5) and (4.3+2.4+2.4,0.5) ..  (3.8+2.4+2.4,0.1);
\draw [thick] (3.4+2.4+2.4,0) -- (4+2.4+2.4,0);
\draw [thick] (3.4+2.4+2.4,0.5) -- (4+2.4+2.4,0);
\draw [thick] (3.4+2.4+2.4,-0.5) -- (4+2.4+2.4,0);
\node [align=center,align=center] at (4+2.4+2.4,-0.35) {$v_1$};
\node [align=center,align=center] at (4.6+2.4+2.4,-0.35) {$v_j$};
\node [align=center,align=center] at (3.2+2.4+2.4,-0.5) {$h_2$};
\node [align=center,align=center] at (3.2+2.4+2.4,0) {$h_1$};
\node [align=center,align=center] at (3.2+2.4+2.4,0.5) {$h_3$};
\node [align=center,align=center] at (9.8,-0.3) {$.$};
\end{tikzpicture}\\
\end{split}&&
\end{flalign*}

\end{Example}

Now let us define another operator $\cJ_{\{1,2\}}$.
Given a fat graph $\Gamma$ (not necessarily connected)
with $n$ vertices labelled by $v_1,v_2,\cdots, v_n$ ($n\geq 2$),
we consider the graph obtained from $\Gamma$ by
merging the two vertices $v_1$ and $v_2$ into a new vertex
and adding an additional loop on it,
such that the cyclic order on this new vertex is as:
\ben
\text{half-edges on $v_1$}
\quad\to\quad
h
\quad\to\quad
\text{half-edges on $v_1$}
\quad\to\quad
h',
\een
where $h$ and $h'$ are the two half-edges on the new loop,
and half-edges on $v_1$ (resp. $v_2$)
are arranged according to the original cyclic orders
on $v_1$ (resp. $v_2$).
Moreover, we need to label $v_1$ on this new vertex,
and relabel $v_2,v_3,\cdots,v_{n-1}$
on the original vertices $v_3,v_4,\cdots, v_n$ respectively.
Then $\cJ_{\{1,2\}}(\Gamma)$ is defined to be the summation of
all possible resulting graph.

\begin{Example}
Here we give an example to make this definition clear
(here we only show the action on the vertices $v_1$ and $v_2$):
\begin{flalign*}
\begin{split}
&\begin{tikzpicture}[scale=1.02]
\node [align=center,align=center] at (0.15,0) {$\cJ_{\{1,2\}}\bigg($};
\draw [fill] (2-0.3+0.15,0) circle [radius=0.06];
\node [align=center,align=center] at (2-0.3+0.15,-0.35) {$v_1$};
\draw [thick] (0+2-0.3+0.15,0) -- (-0.5+2-0.3+0.15,0.5);
\draw [thick] (0+2-0.3+0.15,0) -- (-0.5+2-0.3+0.15,-0.5);
\draw [thick] (1.7+0.15,0) -- (2.35+0.15,0);
\node [align=center,left] at (-0.5+2-0.25+0.15,0.5) {$h_1$};
\node [align=center,left] at (-0.5+2-0.25+0.15,-0.5) {$h_2$};
\node [align=center,right] at (2.3+0.15,0) {$h_3$};
\draw [->,>=stealth] (-0.2+1.7+0.15,-0.1) .. controls (0.3+1.7+0.15,-0.5) and (0.3+1.7+0.15,0.5) ..  (-0.2+1.7+0.15,0.1);
\node [align=center,align=center] at (4.6,0) {$\bigg)=$};
\draw [fill] (3.6,0) circle [radius=0.06];
\node [align=center,align=center] at (3.6+0.35,0) {$v_2$};
\draw [thick] (3.6,-0.55) -- (3.6,0.55);
\node [align=center,right] at (3.6,0.5) {$h_4$};
\node [align=center,right] at (3.6,-0.5) {$h_5$};
\draw [->,>=stealth] (-0.2+1.7+0.15+1.75,-0.1) .. controls (0.3+1.7+0.15+1.75,-0.5) and (0.3+1.7+0.15+1.75,0.5) ..  (-0.2+1.7+0.15+1.75,0.1);
\draw [fill] (0+7-0.8,0) circle [radius=0.06];
\node [align=center,align=center] at (0+7-0.8,-0.5) {$v_1$};
\draw [thick] (7.4,0) arc (0:174.5:0.6);
\draw [thick] (7.4,0) arc (360:185.5:0.6);
\draw [thick] (-0.5+7-0.8,0) -- (0+7-0.07-0.8,0);
\draw [thick] (-0.5+7-0.8,0.6) -- (0+7-0.05-0.8,0.05);
\draw [thick] (-0.5+7-0.8,-0.6) -- (0+7-0.05-0.8,-0.05);
\draw [thick] (0.05+7-0.8,0.05) -- (0.4+7-0.8,0.3);
\draw [thick] (0.05+7-0.8,-0.05) -- (0.4+7-0.8,-0.3);
\node [align=center,left] at (-0.45+7-0.8,0) {$h_2$};
\node [align=center,left] at (-0.45+7-0.8,0.6) {$h_1$};
\node [align=center,left] at (-0.45+7-0.8,-0.6) {$h_3$};
\node [align=center,right] at (0.35+7-0.8,0.3) {$h_4$};
\node [align=center,right] at (0.35+7-0.8,-0.3) {$h_5$};
\draw [->,>=stealth] (-0.2+7-0.8,-0.1) .. controls (0.3+7-0.8,-0.5) and (0.3+7-0.8,0.5) ..  (-0.2+7-0.8,0.1);
\node [align=center,align=center] at (7.9,0) {$+$};
\draw [fill] (0+7-0.8+3,0) circle [radius=0.06];
\node [align=center,align=center] at (0+7-0.8+3,-0.5) {$v_1$};
\draw [thick] (7.4+3,0) arc (0:174.5:0.6);
\draw [thick] (7.4+3,0) arc (360:185.5:0.6);
\draw [thick] (-0.5+7-0.8+3,0) -- (0+7-0.07-0.8+3,0);
\draw [thick] (-0.5+7-0.8+3,0.6) -- (0+7-0.05-0.8+3,0.05);
\draw [thick] (-0.5+7-0.8+3,-0.6) -- (0+7-0.05-0.8+3,-0.05);
\draw [thick] (0.05+7-0.8+3,0.05) -- (0.4+7-0.8+3,0.3);
\draw [thick] (0.05+7-0.8+3,-0.05) -- (0.4+7-0.8+3,-0.3);
\node [align=center,left] at (-0.45+7-0.8+3,0) {$h_2$};
\node [align=center,left] at (-0.45+7-0.8+3,0.6) {$h_1$};
\node [align=center,left] at (-0.45+7-0.8+3,-0.6) {$h_3$};
\node [align=center,right] at (0.35+7-0.8+3,0.3) {$h_5$};
\node [align=center,right] at (0.35+7-0.8+3,-0.3) {$h_4$};
\draw [->,>=stealth] (-0.2+7-0.8+3,-0.1) .. controls (0.3+7-0.8+3,-0.5) and (0.3+7-0.8+3,0.5) ..  (-0.2+7-0.8+3,0.1);
\end{tikzpicture}\\
&\quad
\begin{tikzpicture}[scale=1.02]
\node [align=center,align=center] at (7.9-3,0) {$+$};
\draw [fill] (0+7-0.8,0) circle [radius=0.06];
\node [align=center,align=center] at (0+7-0.8,-0.5) {$v_1$};
\draw [thick] (7.4,0) arc (0:174.5:0.6);
\draw [thick] (7.4,0) arc (360:185.5:0.6);
\draw [thick] (-0.5+7-0.8,0) -- (0+7-0.07-0.8,0);
\draw [thick] (-0.5+7-0.8,0.6) -- (0+7-0.05-0.8,0.05);
\draw [thick] (-0.5+7-0.8,-0.6) -- (0+7-0.05-0.8,-0.05);
\draw [thick] (0.05+7-0.8,0.05) -- (0.4+7-0.8,0.3);
\draw [thick] (0.05+7-0.8,-0.05) -- (0.4+7-0.8,-0.3);
\node [align=center,left] at (-0.45+7-0.8,0) {$h_3$};
\node [align=center,left] at (-0.45+7-0.8,0.6) {$h_2$};
\node [align=center,left] at (-0.45+7-0.8,-0.6) {$h_1$};
\node [align=center,right] at (0.35+7-0.8,0.3) {$h_4$};
\node [align=center,right] at (0.35+7-0.8,-0.3) {$h_5$};
\draw [->,>=stealth] (-0.2+7-0.8,-0.1) .. controls (0.3+7-0.8,-0.5) and (0.3+7-0.8,0.5) ..  (-0.2+7-0.8,0.1);
\node [align=center,align=center] at (7.9,0) {$+$};
\draw [fill] (0+7-0.8+3,0) circle [radius=0.06];
\node [align=center,align=center] at (0+7-0.8+3,-0.5) {$v_1$};
\draw [thick] (7.4+3,0) arc (0:174.5:0.6);
\draw [thick] (7.4+3,0) arc (360:185.5:0.6);
\draw [thick] (-0.5+7-0.8+3,0) -- (0+7-0.07-0.8+3,0);
\draw [thick] (-0.5+7-0.8+3,0.6) -- (0+7-0.05-0.8+3,0.05);
\draw [thick] (-0.5+7-0.8+3,-0.6) -- (0+7-0.05-0.8+3,-0.05);
\draw [thick] (0.05+7-0.8+3,0.05) -- (0.4+7-0.8+3,0.3);
\draw [thick] (0.05+7-0.8+3,-0.05) -- (0.4+7-0.8+3,-0.3);
\node [align=center,left] at (-0.45+7-0.8+3,0) {$h_3$};
\node [align=center,left] at (-0.45+7-0.8+3,0.6) {$h_2$};
\node [align=center,left] at (-0.45+7-0.8+3,-0.6) {$h_1$};
\node [align=center,right] at (0.35+7-0.8+3,0.3) {$h_5$};
\node [align=center,right] at (0.35+7-0.8+3,-0.3) {$h_4$};
\draw [->,>=stealth] (-0.2+7-0.8+3,-0.1) .. controls (0.3+7-0.8+3,-0.5) and (0.3+7-0.8+3,0.5) ..  (-0.2+7-0.8+3,0.1);
\node [align=center,align=center] at (7.9-3+6,0) {$+$};
\draw [fill] (0+7-0.8+6,0) circle [radius=0.06];
\node [align=center,align=center] at (0+7-0.8+6,-0.5) {$v_1$};
\draw [thick] (7.4+6,0) arc (0:174.5:0.6);
\draw [thick] (7.4+6,0) arc (360:185.5:0.6);
\draw [thick] (-0.5+7-0.8+6,0) -- (0+7-0.07-0.8+6,0);
\draw [thick] (-0.5+7-0.8+6,0.6) -- (0+7-0.05-0.8+6,0.05);
\draw [thick] (-0.5+7-0.8+6,-0.6) -- (0+7-0.05-0.8+6,-0.05);
\draw [thick] (0.05+7-0.8+6,0.05) -- (0.4+7-0.8+6,0.3);
\draw [thick] (0.05+7-0.8+6,-0.05) -- (0.4+7-0.8+6,-0.3);
\node [align=center,left] at (-0.45+7-0.8+6,0) {$h_1$};
\node [align=center,left] at (-0.45+7-0.8+6,0.6) {$h_3$};
\node [align=center,left] at (-0.45+7-0.8+6,-0.6) {$h_2$};
\node [align=center,right] at (0.35+7-0.8+6,0.3) {$h_4$};
\node [align=center,right] at (0.35+7-0.8+6,-0.3) {$h_5$};
\draw [->,>=stealth] (-0.2+7-0.8+6,-0.1) .. controls (0.3+7-0.8+6,-0.5) and (0.3+7-0.8+6,0.5) ..  (-0.2+7-0.8+6,0.1);
\node [align=center,align=center] at (7.9+6,0) {$+$};
\draw [fill] (0+7-0.8+3+6,0) circle [radius=0.06];
\node [align=center,align=center] at (0+7-0.8+3+6,-0.5) {$v_1$};
\draw [thick] (7.4+3+6,0) arc (0:174.5:0.6);
\draw [thick] (7.4+3+6,0) arc (360:185.5:0.6);
\draw [thick] (-0.5+7-0.8+3+6,0) -- (0+7-0.07-0.8+3+6,0);
\draw [thick] (-0.5+7-0.8+3+6,0.6) -- (0+7-0.05-0.8+3+6,0.05);
\draw [thick] (-0.5+7-0.8+3+6,-0.6) -- (0+7-0.05-0.8+3+6,-0.05);
\draw [thick] (0.05+7-0.8+3+6,0.05) -- (0.4+7-0.8+3+6,0.3);
\draw [thick] (0.05+7-0.8+3+6,-0.05) -- (0.4+7-0.8+3+6,-0.3);
\node [align=center,left] at (-0.45+7-0.8+3+6,0) {$h_1$};
\node [align=center,left] at (-0.45+7-0.8+3+6,0.6) {$h_3$};
\node [align=center,left] at (-0.45+7-0.8+3+6,-0.6) {$h_2$};
\node [align=center,right] at (0.35+7-0.8+3+6,0.3) {$h_5$};
\node [align=center,right] at (0.35+7-0.8+3+6,-0.3) {$h_4$};
\draw [->,>=stealth] (-0.2+7-0.8+3+6,-0.1) .. controls (0.3+7-0.8+3+6,-0.5) and (0.3+7-0.8+3+6,0.5) ..  (-0.2+7-0.8+3+6,0.1);
\node [align=center,align=center] at (16.6,-0.35) {$.$};
\end{tikzpicture}\\
\end{split}&&
\end{flalign*}

\end{Example}

Now we denote:
\be
\label{eq-abstract-npt-tilde}
\begin{split}
&\tcW_{g,n}:=
\delta_{g,0}\delta_{n,1}\cF_0^{(0)}
+\sum_{\mu\in \bZ_{>0}^n}\cF_g^{\mu},
\\
&\tcW_{g,I}:=
\delta_{g,0}\delta_{n,1}\cF_{0,I}^{(0)}
+\sum_{\mu\in \bZ_{>0}^n}\cF_{g,I}^{\mu},
\end{split}
\ee
or equivalently,
\be
\label{eq-npt-hat-tilde}
\begin{split}
&\tcW_{g,n}:=
\delta_{g,0}\delta_{n,1}\cF_0^{(0)}+
\sigma_1^{-1}\cdots \sigma_n^{-1}\cW_{g,n},\\
&\tcW_{g,I}:=
\delta_{g,0}\delta_{n,1}\cF_{0,I}^{(0)}+
\sigma_{i_1}^{-1}\cdots \sigma_{i_n}^{-1}\cW_{g,I},
\end{split}
\ee
(see \eqref{eq-operator-sigma}).
Then we can describe our main result in this subsection:

\begin{Theorem}
\label{thm-abstract-eotype-2}
For $(g,n)\not=(0,1)$,
we have:
\be
\label{eq-abstract-eotype-2}
\begin{split}
\tcW_{g,n}=&
\sigma_1^{-1}\bigg[\sum_{j=2}^{n}\cS_{\{1;j\}}\tcW_{g,n-1}
+
\cJ_{\{1,2\}}\bigg(
\tcW_{g-1,n+1}\\
&\qquad\quad
+\sum_{\substack{g_1+g_2=g\\I\sqcup J=[n+1]\backslash\{1,2\}}}
\tcW_{g_1,\{1\}\sqcup I}\cdot\tcW_{g_2,\{2\}\sqcup J}\bigg)
\bigg],
\end{split}
\ee
or equivalently,
\be
\label{eq-abstract-eotype-3}
\begin{split}
\cW_{g,n}=&
\sigma_2\sigma_3\cdots\sigma_n
\bigg[\sum_{j=2}^{n}\cS_{\{1;j\}}\tcW_{g,n-1}
+
\cJ_{\{1,2\}}\bigg(
\tcW_{g-1,n+1}\\
&\qquad\qquad\quad
+\sum_{\substack{g_1+g_2=g\\I\sqcup J=[n+1]\backslash\{1,2\}}}
\tcW_{g_1,\{1\}\sqcup I}\cdot\tcW_{g_2,\{2\}\sqcup J}\bigg)
\bigg].
\end{split}
\ee

\end{Theorem}

We omit the proof of this theorem since it is not hard to see that
\eqref{eq-abstract-eotype-2} is actually equivalent to
the recursion relation \eqref{eq-abstract-eotype}.
Similar to the discussions in \S \ref{sec-fat-abs-virasoro},
we may simply regard the operators $\cS_{\{1;j\}}$ and $\cJ_{\{1,2\}}$
as the inverse procedures of the edge-contraction in $K_1$,
then one can easily construct the
correspondence of graphs in two sides of the above equations
using the same method that we have used in the proof of
Theorem \ref{thm-abstract-rec}.

\begin{Example}
Let us give two examples of the above recursion relation.
Since $\tcW_{g,n}$ are all infinite sums,
here we simply take two graphs and compare their coefficients
in two sides of \eqref{eq-abstract-eotype-2}.

The first example is $(g,n)=(0,4)$.
and let us consider the following graph:
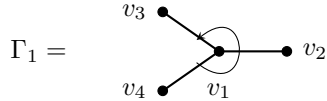
\ben
\begin{tikzpicture}[scale=1.5]
\node [align=center,align=center] at (-1.6,0) {$\Gamma_1=$};
\draw [fill] (0,0) circle [radius=0.0425];
\draw [fill] (0.6,0) circle [radius=0.0425];
\draw [fill] (-0.5,0.35) circle [radius=0.0425];
\draw [fill] (-0.5,-0.35) circle [radius=0.0425];
\node [align=center,align=center] at (0,-0.35) {$v_1$};
\node [align=center,align=center] at (0.85,0) {$v_2$};
\node [align=center,align=center] at (-0.75,0.35) {$v_3$};
\node [align=center,align=center] at (-0.75,-0.35) {$v_4$};
\draw [thick] (0,0) -- (0.6,0);
\draw [thick] (0,0) -- (-0.5,0.35);
\draw [thick] (0,0) -- (-0.5,-0.35);
\draw [->,>=stealth] (-0.2,-0.1) .. controls (0.3,-0.5) and (0.3,0.5) ..  (-0.2,0.1);
\node [align=center,align=center] at (1.15,-0.25) {$.$};
\end{tikzpicture}
\een
The coefficient of $\Gamma_1$ in the left-hand side of \eqref{eq-abstract-eotype-2}
is $\frac{1}{|\Aut(\Gamma_1)|}=1$.
Now let us compute its coefficient in the right-hand side.

Considering all possible edge-contractions of half-edges on $v_1$,
we see that $\Gamma_1$ can only be obtained from the following graph
in $\mathfrak{Fat}_0^{(2,1,1)}$ via the operators $\cS_{(1;j)}$:
\ben
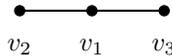
\begin{tikzpicture}[scale=1.6]
\draw [fill] (0,0) circle [radius=0.0425];
\draw [fill] (0.6,0) circle [radius=0.0425];
\draw [fill] (1.2,0) circle [radius=0.0425];
\draw [thick] (0,0) -- (1.2,0);
\node [align=center,align=center] at (0,-0.3) {$v_2$};
\node [align=center,align=center] at (0.6,-0.3) {$v_1$};
\node [align=center,align=center] at (1.2,-0.3) {$v_3$};
\node [align=center,align=center] at (1.55,-0.15) {$.$};
\end{tikzpicture}
\een
And all the possible ways to obtain $\Gamma_1$ from the above graph are:
\ben
&&\begin{tikzpicture}[scale=1.55]
\node [align=center,align=center] at (-3.95,0) {$\cS_{\{1;2\}}\bigg($};
\draw [fill] (0-3-0.2,0) circle [radius=0.045];
\draw [fill] (0.6-3-0.2,0) circle [radius=0.045];
\draw [fill] (1.2-3-0.2,0) circle [radius=0.045];
\draw [thick] (0-3-0.2,0) -- (1.2-3-0.2,0);
\node [align=center,align=center] at (0-3-0.2,-0.35) {$v_2$};
\node [align=center,align=center] at (0.6-3-0.2,-0.35) {$v_1$};
\node [align=center,align=center] at (1.2-3-0.2,-0.35) {$v_3$};
\node [align=center,align=center] at (-1.4,0) {$\bigg)=$};
\draw [fill] (0,0) circle [radius=0.045];
\draw [fill] (0.6,0) circle [radius=0.045];
\draw [fill] (-0.5,0.35) circle [radius=0.045];
\draw [fill] (-0.5,-0.35) circle [radius=0.045];
\node [align=center,align=center] at (0,-0.35) {$v_1$};
\node [align=center,align=center] at (0.85,0) {$v_2$};
\node [align=center,align=center] at (-0.75,0.35) {$v_3$};
\node [align=center,align=center] at (-0.75,-0.35) {$v_4$};
\draw [thick] (0,0) -- (0.6,0);
\draw [thick] (0,0) -- (-0.5,0.35);
\draw [thick] (0,0) -- (-0.5,-0.35);
\draw [->,>=stealth] (-0.2,-0.1) .. controls (0.3,-0.5) and (0.3,0.5) ..  (-0.2,0.1);
\node [align=center,align=center] at (2.1,0) {$+$ other terms;};
\end{tikzpicture}\\
&&\begin{tikzpicture}[scale=1.55]
\node [align=center,align=center] at (-3.95,0) {$\cS_{\{1;3\}}\bigg($};
\draw [fill] (0-3-0.2,0) circle [radius=0.045];
\draw [fill] (0.6-3-0.2,0) circle [radius=0.045];
\draw [fill] (1.2-3-0.2,0) circle [radius=0.045];
\draw [thick] (0-3-0.2,0) -- (1.2-3-0.2,0);
\node [align=center,align=center] at (0-3-0.2,-0.35) {$v_2$};
\node [align=center,align=center] at (0.6-3-0.2,-0.35) {$v_1$};
\node [align=center,align=center] at (1.2-3-0.2,-0.35) {$v_3$};
\node [align=center,align=center] at (-1.4,0) {$\bigg)=$};
\draw [fill] (0,0) circle [radius=0.045];
\draw [fill] (0.6,0) circle [radius=0.045];
\draw [fill] (-0.5,0.35) circle [radius=0.045];
\draw [fill] (-0.5,-0.35) circle [radius=0.045];
\node [align=center,align=center] at (0,-0.35) {$v_1$};
\node [align=center,align=center] at (0.85,0) {$v_2$};
\node [align=center,align=center] at (-0.75,0.35) {$v_3$};
\node [align=center,align=center] at (-0.75,-0.35) {$v_4$};
\draw [thick] (0,0) -- (0.6,0);
\draw [thick] (0,0) -- (-0.5,0.35);
\draw [thick] (0,0) -- (-0.5,-0.35);
\draw [->,>=stealth] (-0.2,-0.1) .. controls (0.3,-0.5) and (0.3,0.5) ..  (-0.2,0.1);
\node [align=center,align=center] at (2.1,0) {$+$ other terms;};
\end{tikzpicture}\\
&&\begin{tikzpicture}[scale=1.55]
\node [align=center,align=center] at (-3.95,0) {$\cS_{\{1;4\}}\bigg($};
\draw [fill] (0-3-0.2,0) circle [radius=0.045];
\draw [fill] (0.6-3-0.2,0) circle [radius=0.045];
\draw [fill] (1.2-3-0.2,0) circle [radius=0.045];
\draw [thick] (0-3-0.2,0) -- (1.2-3-0.2,0);
\node [align=center,align=center] at (0-3-0.2,-0.35) {$v_2$};
\node [align=center,align=center] at (0.6-3-0.2,-0.35) {$v_1$};
\node [align=center,align=center] at (1.2-3-0.2,-0.35) {$v_3$};
\node [align=center,align=center] at (-1.4,0) {$\bigg)=$};
\draw [fill] (0,0) circle [radius=0.045];
\draw [fill] (0.6,0) circle [radius=0.045];
\draw [fill] (-0.5,0.35) circle [radius=0.045];
\draw [fill] (-0.5,-0.35) circle [radius=0.045];
\node [align=center,align=center] at (0,-0.35) {$v_1$};
\node [align=center,align=center] at (0.85,0) {$v_2$};
\node [align=center,align=center] at (-0.75,0.35) {$v_3$};
\node [align=center,align=center] at (-0.75,-0.35) {$v_4$};
\draw [thick] (0,0) -- (0.6,0);
\draw [thick] (0,0) -- (-0.5,0.35);
\draw [thick] (0,0) -- (-0.5,-0.35);
\draw [->,>=stealth] (-0.2,-0.1) .. controls (0.3,-0.5) and (0.3,0.5) ..  (-0.2,0.1);
\node [align=center,align=center] at (2.1,0) {$+$ other terms,};
\end{tikzpicture}\\
\een
thus the coefficient of $\Gamma_1$ in the right-hand side
of \eqref{eq-abstract-eotype-2} is $\frac{1}{3}(1+1+1)=1$.

Let us see another example.
Take $(g,n)=(1,2)$,
and let $\Gamma_2$ be:
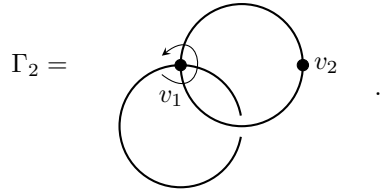
\ben
\begin{tikzpicture}[scale=1.25]
\node [align=center,align=center] at (-1.5,0) {$\Gamma_2=$};
\draw [fill] (0,0) circle [radius=0.06];
\draw [fill] (1.3,0) circle [radius=0.06];
\draw [->,>=stealth] (-0.2,-0.1) .. controls (0.3,-0.5) and (0.3,0.5) ..  (-0.2,0.1);
\draw [thick] (0,0) arc (180:540:0.65);
\draw [thick] (0,0) arc (90:10:0.65);
\draw [thick] (0,0) arc (90:350:0.65);
\node [align=center,align=center] at (-0.1,-0.35) {$v_1$};
\node [align=center,align=center] at (1.55,0) {$v_2$};
\node [align=center,align=center] at (2.1,-0.3) {$.$};
\end{tikzpicture}
\een
The coefficient of $\Gamma_2$ in the left-hand side of \eqref{eq-abstract-eotype-2}
is $\frac{1}{|\Aut(\Gamma_2)|}=\frac{1}{2}$.

Now $\Gamma_2$ can be obtained in two ways
from the right-hand side of \eqref{eq-abstract-eotype-2}:
\begin{itemize}
\item[1)]
$\Gamma$ can be obtained by applying $\cJ_{\{1,2\}}$:
\ben
\begin{tikzpicture}[scale=1.2]
\node [align=center,align=center] at (-1,0) {$\cJ_{\{1,2\}}\bigg($};
\draw [fill] (0,0) circle [radius=0.06];
\draw [fill] (0.6,0) circle [radius=0.06];
\draw [fill] (1.2,0) circle [radius=0.06];
\draw [thick] (0,0) -- (1.2,0);
\node [align=center,align=center] at (0,-0.35) {$v_1$};
\node [align=center,align=center] at (0.6,-0.35) {$v_3$};
\node [align=center,align=center] at (1.2,-0.35) {$v_2$};
\node [align=center,align=center] at (1.95,0) {$\bigg)=$};
\draw [fill] (0+3.3,0) circle [radius=0.06];
\draw [fill] (1.3+3.3,0) circle [radius=0.06];
\draw [->,>=stealth] (-0.2+3.3,-0.1) .. controls (0.3+3.3,-0.5) and (0.3+3.3,0.5) ..  (-0.2+3.3,0.1);
\draw [thick] (0+3.3,0) arc (180:540:0.65);
\draw [thick] (0+3.3,0) arc (90:10:0.65);
\draw [thick] (0+3.3,0) arc (90:350:0.65);
\node [align=center,align=center] at (-0.1+3.3,-0.35) {$v_1$};
\node [align=center,align=center] at (1.55+3.3,0) {$v_2$};
\node [align=center,align=center] at (6.45,0) {$+$ other terms.};
\end{tikzpicture}
\een

\item[2)]
$\Gamma_2$ can also be obtained by applying $\cS_{\{1;j\}}$:
\ben
\begin{tikzpicture}[scale=1.2]
\node [align=center,align=center] at (-1.7,0) {$\cS_{\{1;2\}}\bigg(\quad\frac{1}{4}\times$};
\draw [fill] (0,0) circle [radius=0.06];
\draw [->,>=stealth] (-0.2,-0.1) .. controls (0.3,-0.5) and (0.3,0.5) ..  (-0.2,0.1);
\draw [thick] (0,0) arc (180:540:0.65);
\draw [thick] (0,0) arc (90:10:0.65);
\draw [thick] (0,0) arc (90:350:0.65);
\node [align=center,align=center] at (-0.1,-0.35) {$v_1$};
\node [align=center,align=center] at (2.05,0) {$\bigg)=$};
\draw [fill] (0+3.4,0) circle [radius=0.06];
\draw [fill] (1.3+3.4,0) circle [radius=0.06];
\draw [->,>=stealth] (-0.2+3.4,-0.1) .. controls (0.3+3.4,-0.5) and (0.3+3.4,0.5) ..  (-0.2+3.4,0.1);
\draw [thick] (0+3.4,0) arc (180:540:0.65);
\draw [thick] (0+3.4,0) arc (90:10:0.65);
\draw [thick] (0+3.4,0) arc (90:350:0.65);
\node [align=center,align=center] at (-0.1+3.4,-0.35) {$v_1$};
\node [align=center,align=center] at (1.55+3.4,0) {$v_2$};
\node [align=center,align=center] at (6.55,0) {$+$ other terms.};
\end{tikzpicture}
\een

\end{itemize}
Thus the coefficient of $\Gamma_2$
in the right-hand side of \eqref{eq-abstract-eotype-2}
is $\frac{1}{4}(1+1)=\frac{1}{2}$.

\end{Example}

\subsection{A quadratic recursion for abstract $n$-point functions}

\label{sec-qrec-npt}

In this subsection,
we give a reformulation the recursion relation in Theorem \ref{thm-abstract-eotype-2}.

Inspired by the formulation Eynard-Orantin topological recursion,
we want to modify the recursion \ref{eq-abstract-eotype-2} in such a way
that all terms of type $(g,n)=(0,1)$ are excluded from the summation
in the right-hand side of \eqref{eq-abstract-eotype-2}.
Notice that there are two such terms:
\be
\sigma_1^{-1}\circ \cJ_{\{1,2\}}
\bigg(
\tcW_{0,\{1\}}\cdot\tcW_{g,[n+1]\backslash\{1\}}+
\tcW_{g,[n+1]\backslash\{2\}}\cdot\tcW_{0,\{2\}}
\bigg),
\ee
and it is not hard to see that:
\be
\label{eq-mult-(0,1)}
\cJ_{\{1,2\}}\bigg(\tcW_{0,\{1\}}\cdot\tcW_{g,[n+1]\backslash\{1\}}\bigg)=
\cJ_{\{1,2\}}\bigg(\tcW_{g,[n+1]\backslash\{2\}}\cdot\tcW_{0,\{2\}}\bigg),
\ee
These two terms can be understood as some operator acting on
the abstract $n$-point function $\tcW_{g,n}$.
In fact,
they can be obtained from $\tcW_{g,n}$ by the following procedures:
First we choose a graph appearing in the expression of $\tcW_{g,n}$
(i.e., let $\Gamma_1$ be a graph of genus $g$ with $n$ vertices),
and let $\Gamma_2$ be a graph of genus $0$ with one single vertex
(here we allow the special case that $\Gamma_2$ is a single vertex of valence zero).
We choose a half-edge $h_1$ attached to the vertex $v_1\in V(\Gamma_1)$,
and a half-edge $h_2$ attached to the vertex in $\Gamma_2$.
Then we add a new loop $e$ to the vertex $v_1\in V(\Gamma_1)$
such that the original half-edges attached to $v_1$ are placed on the same side;
and we merge the vertex in $\Gamma_2$ with this new $v_1$ such that
the half-edges of $\Gamma_2$ are placed on the other side of $e$.
Moreover,
we require the cyclic orders are compatible
and $h_1$, $h_2$ are right after the two half-edges of $e$ respectively.
For example:
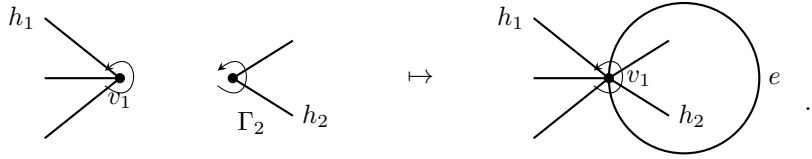
\begin{equation*}
\begin{tikzpicture}[scale=1]
\draw [fill] (0+11,0) circle [radius=0.06];
\draw [thick] (1+11,0) circle [radius=1];
\draw [thick] (-1+11,0) -- (0+11,0);
\draw [thick] (-1+11,0.8) -- (0+11,0);
\draw [thick] (-1+11,-0.8) -- (0+11,0);
\draw [thick] (0+11,0) -- (0.8+11,0.5);
\draw [thick] (0+11,0) -- (0.8+11,-0.5);
\node [align=center,left] at (-1+11,0.8) {$h_1$};
\node [align=center,right] at (0.8+11,-0.5) {$h_2$};
\node [align=center,right] at (2+11,0) {$e$};
\node [align=center,align=center] at (0.4+11,0) {$v_1$};
\draw [->,>=stealth] (-0.2+11,-0.1) .. controls (0.3+11,-0.5) and (0.3+11,0.5) ..  (-0.2+11,0.1);
\node [align=center,align=center] at (13.65,-0.4) {$.$};
\node [align=center,align=center] at (3+5+0.5,0) {$\mapsto$};
\draw [fill] (5-0.5,0) circle [radius=0.06];
\draw [fill] (6,0) circle [radius=0.06];
\node [align=center,align=center] at (5-0.5,-0.3) {$v_1$};
\node [align=center,align=center] at (6.25,-0.6) {$\Gamma_2$};
\draw [thick] (-1+5-0.5,0) -- (0+5-0.5,0);
\draw [thick] (-1+5-0.5,0.8) -- (0+5-0.5,0);
\draw [thick] (-1+5-0.5,-0.8) -- (0+5-0.5,0);
\draw [thick] (0+6,0) -- (0.8+6,0.5);
\draw [thick] (0+6,0) -- (0.8+6,-0.5);
\node [align=center,left] at (-1+5-0.5,0.8) {$h_1$};
\node [align=center,right] at (0.8+6,-0.5) {$h_2$};
\draw [->,>=stealth] (-0.2+5-0.5,-0.1) .. controls (0.3+5-0.5,-0.5) and (0.3+5-0.5,0.5) ..  (-0.2+5-0.5,0.1);
\draw [->,>=stealth] (-0.2+6,-0.1) .. controls (0.3+6,-0.5) and (0.3+6,0.5) ..  (-0.2+6,0.1);
\end{tikzpicture}
\end{equation*}
Finally,
we take summation over all possible $h_1,h_2$ and all possible $\Gamma_2$,
then this procedure describes \eqref{eq-mult-(0,1)} as an operator acting on
the abstract $n$-point function $\tcW_{g,n}$.
Let us denote this operator by $\cT$,
i.e.,
\be
\label{eq-operatorT}
\begin{split}
\cT(\tcW_{g,n}):=
&\cJ_{\{1,2\}}\bigg(\tcW_{0,\{1\}}\cdot\tcW_{g,[n+1]\backslash\{1\}}\bigg)\\
=&
\cJ_{\{1,2\}}\bigg(\tcW_{g,[n+1]\backslash\{2\}}\cdot\tcW_{0,\{2\}}\bigg),
\end{split}
\ee
then it is not hard to see that every graph appearing in the expression of $\cT(\tcW_{g,n})$
is still of genus $g$,
with $n$ vertices labelled by $v_1,v_2,\cdots,v_n$.

Then Theorem \ref{thm-abstract-eotype-2} gives us:

\begin{Theorem}
\label{thm-abstract-eotype-3}
For $(g,n)\not=(0,1)$,
we have:
\be
\label{eq-rec-abstractEO-2}
\begin{split}
(\sigma_1-2\cT)\tcW_{g,n}=&
\sum_{j=2}^{n}\cS_{\{1;j\}}\tcW_{g,n-1}
+
\cJ_{\{1,2\}}\bigg(
\tcW_{g-1,n+1}\\
&\qquad\quad
+\sum_{\substack{g_1+g_2=g\\ I\sqcup J=[n+1]\backslash\{1,2\}}}^s
\tcW_{g_1,\{1\}\sqcup I}\cdot\tcW_{g_2,\{2\}\sqcup J}\bigg),
\end{split}
\ee
where the notation $\sum\limits^s$ means that
we require $(g_1,I)\not=(0,\emptyset)$ and $(g_2,J)\not=(0,\emptyset)$
in this summation,
and $\sigma_1$ is defined by \eqref{eq-operator-sigma}.
Or equivalently,
by applying $\sigma_2\cdots\sigma_n$ to both sides we have:
\be
\label{eq-rec-abstractEO}
\begin{split}
&(1-2\cT\circ\sigma_1^{-1})\cW_{g,n}\\
=&
\sum_{j=2}^{n}\sigma_j\circ\cS_{\{1;j\}}\circ\sigma_1^{-1}\cW_{g,n-1}
+
\cJ_{\{1,2\}}\circ\sigma_1^{-1}\sigma_2^{-1}\bigg(
\cW_{g-1,n+1}\\
&\qquad\quad
+\sum_{\substack{g_1+g_2=g\\I\sqcup J=[n+1]\backslash\{1,2\}}}^s
\cW_{g_1,\{1\}\sqcup I}\cdot\cW_{g_2,\{2\}\sqcup J}\bigg).
\end{split}
\ee

\end{Theorem}

\begin{Remark}
As mentioned in the beginning of this section,
we want to relate this quadratic recursion for the abstract $n$-point functions $\cW_{g,n}$
to the Eynard-Orantin topological recursion.
In the following sections we'll see that the realization of the above recursion is
actually equivalent to the E-O topological recursion
in the example of Hermitian one-matrix models.
In general cases,
we formulate this idea as a conjecture in \S \ref{sec-conj}.

\end{Remark}

\subsection{Forgetting the labels on the abstract $n$-point functions}
\label{sec-qsc}

In this subsection we consider the abstract $n$-point functions
without the labels $v_1,\cdots,v_n$ on vertices
(similar to the case of the abstract free energy and the abstract partition function,
see \S \ref{sec-abstractpartition}).

Denote by $\cF_{g,n}$ the formal summation of fat graphs (without labels on vertices)
obtained by forgetting all the labels in
\be
\tcW_{g,n}=
\delta_{g,0}\delta_{n,1}\cF_0^{(0)}
+\sum_{\mu\in \bZ_{>0}^n}\cF_g^{\mu}.
\ee
In what follows let us reformulate the recursion \eqref{eq-rec-abstractEO-2} as
a quadratic recursion relation for $\cF_{g,n}$.
Recall that the previous subsections,
we have regarded the label $v_1$ as a special one
when performing the edge-contraction or vertex-splitting procedures.
Notice that the abstract $n$-point functions $\cW_{g,n}$ are invariant
under permutations of the labels $v_1,\cdots,v_n$,
thus the recursion \eqref{eq-rec-abstractEO-2} can also be reformulated in such a way
that the label $v_1$ is replaced by $v_j$ for $j=2,3,\cdots,n$.
Now we can take summation of all the $n$ recursion relations
(in which the special labels are $v_1,v_2,\cdots,v_n$ respectively)
and then forget all the labels,
and in this way we obtain the following:
\be
\label{eq-qrec-unlabel-npt}
\begin{split}
\sigma(\cF_{g,n}) - 2n\cdot\cJ(\cF_{0,1},\frac{\cF_{g,n}}{n})
=& n\cS(\cF_{g,n-1}) + n \bigg(\cJ(\frac{\cF_{g-1,n+1}}{n(n+1)})\\
&+
\sum_{\substack{g_1+g_2=g\\ i+j=n+1}}^s \binom{n-1}{i-1}\cdot
\cJ(\frac{\cF_{g_1,i}}{i},\frac{\cF_{g_2,j}}{j})
\bigg),
\end{split}
\ee
where $\sigma$ is the rescaling map
\ben
\sigma(\Gamma) = (\mu_1+\cdots+\mu_n) \Gamma,
\qquad
\Gamma\in \mathfrak{Fat}_{g}^{\mu,c},
\een
or equivalently,
$\sigma(\cF_{g,n})$ is obtained by forgetting all the labels in
$(\sigma_1+\cdots +\sigma_n)\tcW_{g,n}$.
And $\cJ(\Gamma_1,\Gamma_2)$ is defined to be the bilinear map such that:
\ben
\cJ(\Gamma_1,\Gamma_2):=
\sum_{\substack{h_1\in H(\Gamma_1)\\ h_2\in H(\Gamma_2)}}
\cJ' (\Gamma_1\cdot\Gamma_2, h_1,h_2),
\qquad
\Gamma_i\in \mathfrak{Fat}_{g_i}^{\mu_i,c},
\een
where $\cJ' (\Gamma_1\cdot\Gamma_2, h_1,h_2)$
is obtained from $\cJ (\Gamma_1\cdot\Gamma_2, h_1,h_2)$
(see Definition \ref{def-abs-opr-J}) by changing the hollow vertex to a solid one,
and $\Gamma_1\cdot\Gamma_2$ is the disjoint union of $\Gamma_1$ and $\Gamma_2$.
The linear operator $\cS$ is defined by:
\ben
\cS(\Gamma) := \sum_{k\geq 0}\sum_{l=0}^k
\cS_{k,l}' (\Gamma),
\een
where $\cS_{k,l}' (\Gamma)$ is obtained from $\cS_{k,l} (\Gamma)$
(see Definition \ref{def-vertexsplit-opr}) by changing the hollow vertex to a solid one.
And the linear map $\cJ(\Gamma)$ is defined by:
\ben
\cJ(\Gamma):= \sum_{v_1\in V(\Gamma)}
\sum_{v_2\in V(\Gamma)\backslash\{v_1\}}
\sum_{h_i \in H(v_i)} \cJ' (\Gamma, h_1, h_2),
\qquad
\Gamma\in \mathfrak{Fat}_{g}^{\mu,c}.
\een
Notice that the summation in the right-hand side of \eqref{eq-qrec-unlabel-npt}
do not contain the terms with $(g_1,i)=(0,1)$ or $(g_2,j)=(0,1)$,
thus one can rewrite \eqref{eq-qrec-unlabel-npt} as follows:
\begin{Lemma}
We have:
\be
\label{eq-qrec-unlabel-npt-2}
\begin{split}
\sigma(\cF_{g,n})
=& n\cS(\cF_{g,n-1}) + \frac{1}{n+1}\cJ(\cF_{g-1,n+1})\\
&+
\sum_{\substack{g_1+g_2=g\\ i+j=n+1}}
\frac{n}{ij}\binom{n-1}{i-1}\cdot
\cJ(\cF_{g_1,i},\cF_{g_2,j}).
\end{split}
\ee
\end{Lemma}

Now for every $m\geq 0$, define:
\be
\cF_{(m)} := \sum_{2g-1+n = m} \frac{1}{n!}\cF_{g,n},
\ee
then:
\begin{equation*}
\sigma(\cF_{(m)})
=\sum_{2g-1+n=m}\bigg(
\frac{\cS(\cF_{g,n-1})}{(n-1)!} + \frac{\cJ(\cF_{g-1,n+1})}{(n+1)!}+
\sum_{\substack{g_1+g_2=g\\ i+j=n+1}}
\frac{\cJ(\cF_{g_1,i},\cF_{g_2,j})}{i!\cdot j!}
\bigg),
\end{equation*}
or equivalently:
\begin{Theorem}
\label{thm-abstract-qsc}
For every $m\geq 0$ we have:
\be
\label{eq-abstract-qsc}
\sigma(\cF_{(m)})
=\cS (\cF_{(m-1)}) + \cJ (\cF_{(m-1)}) + \sum_{\substack{i+j=m\\i,j\geq 0}} \cJ(\cF_{(i)},\cF_{(j)}),
\ee
where we use the convention $\cF_{(-1)}:=0$.
\end{Theorem}

\begin{Remark}
In the next section,
we will consider the realization of the abstract $n$-point functions
by the $n$-point functions of the Hermitian one-matrix models,
and in this case the above recursion relation will give a Schr\"{o}dinger type equation.
Using a result by the second author in \cite{zhou11},
in \S \ref{sec-conj} we will see that it is the quantum spectral curve
for the Hermitian one-matrix models.
\end{Remark}

\section{Realization by $N$-Point Functions of Hermitian One-Matrix Models}
\label{sec-1mm-npt}

In this section we consider the realization of the abstract $n$-point functions
by the $n$-point functions of the Hermitian one-matrix models.
We show that the quadratic recursion \eqref{eq-rec-abstractEO} is realized by
a quadratic differential equation derived by the second author in \cite{zhou11}
using Virasoro constraints.
Moreover,

\subsection{Realization of the abstract $n$-point functions}
\label{sec-realization-1mm-EO}

In this subsection,
we construct the realization of the abstract $n$-point functions $\cW_{g,n}$
by the Hermitian one-matrix models.

Let $x:=(x_1,x_2,\cdots)$ be a family of formal variables,
and $\mu:=(\mu_1,\mu_2\cdots, \mu_n)$.
Given a fat graph $\Gamma\in \mathfrak{Fat}_g^{\mu}$,
we assign the following Feynman rule:
\be
\label{eq-FR-1mm-EOtype}
w_\Gamma:=x_1^{-(\mu_1+1)}x_2^{-(\mu_2+1)}\cdots x_n^{-(\mu_n+1)}
\cdot t^{|F(\Gamma)|},
\ee
where $t=Hg_s$ is the 't Hooft coupling constant
(see \S \ref{sec-realization-herm-1mm}).
Moreover,
given a set of indices $I=\{i_1,i_2,\cdots,i_n\}\subset\bZ_{>0}$
and $\Gamma\in\mathfrak{Fat}_{g,I}^{\mu,c}$ (see \S \ref{sec-abs-qrec}),
we assign:
\be
w_\Gamma:=x_{i_1}^{-(\mu_1+1)}x_{i_2}^{-(\mu_2+1)}\cdots x_{i_n}^{-(\mu_n+1)}
\cdot t^{|F(\Gamma)|}.
\ee
Then the abstract $n$-point functions $\cW_{g,n}$ are realized by
the following $n$-point functions of the Hermitian one-matrix model:
\be
\label{eq-herm-npt}
W_{g,n}^{\text{Herm}}(x_1,\cdots,x_n):=
\delta_{g,0}\delta_{n,1}\cdot t x_1^{-1}+
\sum_{\mu\in \bZ_{>0}^n}
\langle p_{\mu_1}\cdots p_{\mu_n}
\rangle_g^c
\cdot x_1^{-(\mu_1+1)}\cdots x_n^{-(\mu_n+1)},
\ee
where $\langle p_{\mu_1} \cdots p_{\mu_n} \rangle_g^c$
are fat correlators of the Hermitian one-matrix models.
And similarly,
the abstract $n$-point functions with shifted indices $\cW_{g,I}$
(defined by \eqref{eq-abstractnpt-relabel}) are realized by:
\begin{equation*}
W_{g,n}^{\text{Herm}}(x_{i_1},\cdots,x_{i_n})=
\delta_{g,0}\delta_{n,1}\cdot t x_{i_1}^{-1}+
\sum_{\mu\in \bZ_{>0}^n}
\langle p_{\mu_1}\cdots p_{\mu_n}
\rangle_g^c
\cdot x_{i_1}^{-(\mu_1+1)}\cdots x_{i_n}^{-(\mu_n+1)}.
\end{equation*}

\begin{Example}
$W_{0,1}^{\text{Herm}}$ is given by the generating series of Catalan numbers:
\begin{equation*}
\begin{split}
W_{0,1}^{\text{Herm}}(x_1)=&\frac{t}{x_1}+
\sum_{n>0}C_n\cdot\frac{t^{n+1}}{x_1^{2n+1}}\\
=&\half\bigg(
x_1-\sqrt{x_1^2-4t}
\bigg)\\
=&\frac{t}{x_1}+
\frac{t^{2}}{x_1^{3}}+
\frac{2t^{3}}{x_1^{5}}+
\frac{5t^{4}}{x_1^{7}}+
\frac{14t^{5}}{x_1^{9}}+
\cdots.
\end{split}
\end{equation*}
\end{Example}

\subsection{Realization of the quadratic recursion}
\label{sec-realizationrec-npt}

Now let us consider the realization of the quadratic recursion \eqref{eq-rec-abstractEO}
in the case of Hermitian one-matrix models.
In this way we recover a quadratic recursion proved in \cite{zhou11}.

We construct the realization of \eqref{eq-rec-abstractEO} term by term.
First, the term
\ben
\cT\circ\sigma_1^{-1}\cW_{g,n}=
\sigma_2^{-1}
\cJ_{\{1,2\}}\bigg(\tcW_{0,1}\cdot\cW_{g,[n+1]\backslash\{1\}}
\bigg)
\een
is realized by:
\be
x_1^{-1}\cdot
W_{0,1}^{\text{Herm}}(x_1)\cdot W_{g,n}^{\text{Herm}}(x_1,x_2,\cdots,x_n),
\ee
since we have
\ben
\frac{1}{x_1^{(\mu_1+\mu_2+2)+1}}
=\frac{1}{x_1}\cdot
\frac{1}{x_1^{\mu_1+1}}\cdot\frac{1}{x_1^{\mu_2+1}},
\een
and the operator $\cJ_{\{1,2\}}$ produces a factor $\mu_1\cdot\mu_2$
(the number of ways to adding a loop)
which kills $\sigma_2^{-1}$ and changes $\tcW_{0,1}$ to $\cW_{0,1}$.
Similarly,
the term
\ben
\cJ_{\{1,2\}}\circ\sigma_1^{-1}\sigma_2^{-1}\bigg(
\cW_{g-1,n+1}
+\sum_{\substack{g_1+g_2=g\\I\sqcup J=[n+1]\backslash\{1,2\}}}^s
\cW_{g_1,\{1\}\sqcup I}\cdot\cW_{g_2,\{2\}\sqcup J}\bigg)
\een
is realized by:
\be
\begin{split}
&E_{x_1,u,v}^{\text{Herm}} W_{g-1,n+1}^{\text{Herm}}(u,v,x_2,\cdots,x_n)\\
+&
\sum_{\substack{g_1+g_2=g\\I\sqcup J=[n]\backslash\{1\}}}^s
x_1^{-1}\cdot
W_{g_1,|I|+1}^{\text{Herm}}(x_1,x_I)W_{g_2,|J|+1}^{\text{Herm}}(x_1,x_J),
\end{split}
\ee
where $x_I:=(x_{i_1},\cdots,x_{i_k})$ for $I=\{i_1,\cdots,i_k\}$,
and the operator $E_{a,b,c}^{\text{Herm}}$ is defined by
(see \cite[\S 3.5]{zhou8}):
\be
E_{a,b,c}^{\text{Herm}} f(b,c)=
a^{-1}\cdot
\big(\lim_{b\to c}f(b,c)\big)|_{c=a}.
\ee
Finally,
we need to construct the realizations of the terms
$\sigma_j\circ\cS_{\{1;j\}}\circ\sigma_1^{-1}\cW_{g,n-1}$
(for $j=2,\cdots,n$).
Similar to the computations in \cite[\S 3.5]{zhou8},
we need to consider the following transformation:
\begin{equation*}
\begin{split}
\frac{1}{x_1^{m+1}}\quad\mapsto\quad&
\sum_{\substack{k+l=m\\k,l\geq 0}}(l+1)\cdot
\frac{1}{x_1^{(k+1)+1}}\cdot\frac{1}{x_j^{(l+1)+1}}\\
=&\frac{1}{(x_1-x_j)^2}\bigg[\bigg(
\frac{1}{x_1^{m+1}}\cdot\frac{1}{x_1}
+\frac{1}{x_j^{m+1}}\bigg(
\frac{m+1}{x_j}-\frac{m+2}{x_1}
\bigg)
\bigg].
\end{split}
\end{equation*}
Define an operator $D_{a,b}^{\text{Herm}}$ by:
\be
D_{a,b}^{\text{Herm}} f(a):=\frac{1}{(a-b)^2}\bigg(
\frac{1}{a}f(a)
-\frac{\pd}{\pd b}f(b)
+\frac{b^2}{a}\cdot\frac{\pd}{\pd b}\big(\frac{1}{b}f(b)\big)
\bigg),
\ee
then $\sigma_j\circ\cS_{\{1;j\}}\circ\sigma_1^{-1}\cW_{g,n-1}$ is realized by:
\be
D_{x_1,x_j}^{\text{Herm}} W_{g,n-1}^{\text{Herm}} (x_1,\cdots,\hat{x}_j,\cdots,x_n),
\ee
where $\hat{x}_j$ means deleting the term $x_j$.

Therefore,
the recursion \eqref{eq-rec-abstractEO} is realized by:
\be
\begin{split}
&\big(1-2x_1^{-1}\cdot W_{0,1}^{\text{Herm}}(x_1)\big)W_{g,n}^{\text{Herm}}(x_1,\cdots,x_n)\\
=&\sum_{j=2}^n D_{x_1,x_j}^{\text{Herm}} W_{g,n-1}^{\text{Herm}}(x_1,\cdots,\hat{x}_j,\cdots,x_n)
+E_{x_1,u,v}W_{g-1,n+1}^{\text{Herm}}(u,v,x_2,\cdots,x_n)\\
&+
\sum_{\substack{g_1+g_2=g\\I\sqcup J=[n]\backslash\{1\}}}^s
x_1^{-1}\cdot
W_{g_1,|I|+1}^{\text{Herm}}(x_1,x_I)W_{g_2,|J|+1}^{\text{Herm}}(x_1,x_J),
\end{split}
\ee
for every $(g,n)\not=(0,1)$.
Now denote:
\be
\tD_{a,b}^{\text{Herm}}:=\frac{D_{a,b}^{\text{Herm}}}{1-2a^{-1}\cdot W_{0,1}^{\text{Herm}}(a)},
\qquad
\tE_{a,b,c}^{\text{Herm}}:=\frac{E_{a,b,c}^{\text{Herm}}}{1-2a^{-1}\cdot W_{0,1}^{\text{Herm}}(a)},
\ee
then we recover one of the main results in \cite{zhou11}:

\begin{Theorem}
[\cite{zhou11}]
\label{thm-EOtype-rec-1mm}

We have the following quadratic recursion for
the $n$-point functions $W_{g,n}^{\text{Herm}}(x_1,\cdots,x_n)$
of the Hermitian one-matrix models:
\be
\label{eq-EOtype-rec-1mm}
\begin{split}
&W_{g,n}^{\text{Herm}}(x_1,x_2,\cdots,x_n)\\
=&\sum_{j=2}^n \tD_{x_1,x_j}^{\text{Herm}} W_{g,n-1}^{\text{Herm}}(x_1,\cdots,\hat{x}_j,\cdots,x_n)
+\tE_{x_1,u,v}^{\text{Herm}} W_{g-1,n+1}^{\text{Herm}}(u,v,x_2,\cdots,x_n)\\
&+
\sum_{\substack{g_1+g_2=g\\I\sqcup J=[n]\backslash\{1\}}}^s
\tE_{x_1,u,v}^{\text{Herm}}\bigg(
W_{g_1,|I|+1}^{\text{Herm}}(u,x_I)\cdot W_{g_2,|J|+1}^{\text{Herm}}(v,x_J)\bigg),
\end{split}
\ee
if $(g,n)\not=(0,1)$.
\end{Theorem}

\begin{Remark}
In \cite{zhou11},
the above quadratic recursion is derived from the fat Virasoro constraints.
Now in this work,
we have seen that the this recursion and the fat Virasoro constraints
are both consequences of the edge-contraction recursion relation for fat graphs.

\end{Remark}

\begin{Example}
Now using the above recursion
one can compute $W_{0,2}^{\text{Herm}}$:
\begin{equation*}
\begin{split}
W_{0,2}^{\text{Herm}} (x_1,x_2)=
&\frac{1}{1-2x_1^{-1}\cdot W_{0,1}^{\text{Herm}} (x_1)} \cdot\bigg(
D_{x_1,x_2}^{\text{Herm}} W_{0,1}^{\text{Herm}} (x_1)
\bigg)\\
=&
\frac{1}{2(x_1-x_2)^2}\bigg(
-1+\frac{x_1x_2-4t}{\sqrt{x_1^2-4t}\cdot\sqrt{x_2^2-4t}}
\bigg)
\\
=&\frac{t}{x_1^2 x_2^2}+
3t^2\bigg(\frac{1}{x_1^2x_2^4}+\frac{1}{x_1^4x_2^2}\bigg)+
\frac{2t^2}{x_1^3x_2^3}+
10t^3\bigg(\frac{1}{x_1^2x_2^6}+\frac{1}{x_1^6x_2^2}\bigg)\\
&+
8t^3\bigg(\frac{1}{x_1^3x_2^5}+\frac{1}{x_1^5x_2^3}\bigg)+
\frac{12t^3}{x_1^4x_2^4}
+\cdots.
\end{split}
\end{equation*}
This is the realization of the abstract Bergmann kernel $\cW_{0,2}$.
\end{Example}

\subsection{Realization of $\cF_{g,n}$ and $\cF_{(m)}$}
\label{sec-realization-qsc}

In this subsection we consider the realization of $\cF_{g,n}$, $\cF_{(m)}$ and their recursion
introduced in \S \ref{sec-qsc}.

Recall that $\cF_{g,n}$ is obtained by forgetting the labels $v_1,\cdots,v_n$ on vertices
on $\tcW_{g,n}$.
Thus it is natural to introduce a new formal variable $x$
and modify the Feynman rule in \S \ref{sec-realization-1mm-EO} by taking $x_1=\cdots=x_n=x$.
Here the Feynman rule we choose is:
\be
\label{eq-Feynman-qsc}
w_\Gamma := x^{-\mu_1}x^{-\mu_2}\cdots x^{-\mu_n} \cdot t^{|F(\Gamma)|}
=x^{-|\mu|}\cdot t^{|F(\Gamma)|}
\ee
where $t$ is the 't Hooft coupling constant as before.
Then $\cF_{g,n}$ is realized by:
\be
\label{eq-hermnpt-average}
\begin{split}
S_{g,n}^{\text{Herm}} (x)
:=&
\sum_{\mu\in\bZ_{>0}^n} \sum_{\Gamma\in \mathfrak{Fat}_g^{\mu,c}}
\frac{x^{-|\mu|} t^{|F(\Gamma)|}}{|\Aut(\Gamma)|}\\
=&
\sum_{\mu\in\bZ_{>0}^n} \langle
\frac{p_{\mu_1}}{\mu_1}\cdots \frac{p_{\mu_n}}{\mu_n}\rangle_g^c
\cdot x^{-|\mu|} t^{|F(\Gamma)|}
\end{split}
\ee
for $(g,n)\not=(0,1)$.
And for the special case $(g,n)=(0,1)$,
we define:
\be
S_{0,1}^{\text{Herm}} (x)
:= t\log x+
\sum_{\mu_1>0} \sum_{\Gamma\in \mathfrak{Fat}_0^{(\mu),c}}
\frac{x^{-|\mu|} t^{|F(\Gamma)|}}{|\Aut(\Gamma)|}.
\ee
It is clear that $\frac{d}{dx}S_{0,1}^{\text{Herm}} (x) = -W_{0,1}^{\text{Herm}} (x)$.
Moreover,
comparing \eqref{eq-hermnpt-average} with \eqref{eq-herm-npt},
one easily finds that for $(g,n)\not= (0,1)$,
the function $S_{g,n}^{\text{Herm}}(x)$ is the following averaging
of the $n$-point function $W_{g,n}^{\text{Herm}}(x_1,\cdots,x_n)$:
\be
S_{g,n}^{\text{Herm}}=
\int_x^{+\infty}\cdots \int_x^{+\infty}
W_{g,n}^{\text{Herm}} (x_1,\cdots,x_n) dx_1\cdots dx_n.
\ee
The realization of $\cF_{(m)}$ is:
\be
\begin{split}
S_m^{\text{Herm}}=&\sum_{2g-1+n=m} S_{g,n}^{\text{Herm}}\\
=&\delta_{m,0}t \log x+ \sum_{2g-1+n=m}
\sum_{\mu\in\bZ_{>0}^n} \langle
\frac{p_{\mu_1}}{\mu_1}\cdots \frac{p_{\mu_n}}{\mu_n}\rangle_g^c
\cdot x^{-|\mu|} t^{|F(\Gamma)|}.
\end{split}
\ee

Now let us consider the realization of the recursion \eqref{eq-abstract-qsc}
with respect to the Feynman rule \eqref{eq-Feynman-qsc}.
Given a connected fat graph $\Gamma$ of type $(g,\mu)$,
the operator $\sigma$ is realized by acts by:
\ben
w_{\sigma(\Gamma)} = w_{|\mu|\cdot \Gamma}
=|\mu|\cdot x^{-|\mu|} t^{|F(\Gamma)|}
= -x\frac{d}{dx} w_{\Gamma},
\een
And $\cS(\Gamma)$ is realized by:
\ben
w_{\cS(\Gamma)} = \sum_{v\in V(\Gamma)} \val(v) (\val(v)+1)
x^{-2}\cdot w_\Gamma,
\een
where the factor $\val(v)(\val(v)+1)$ is the number of ways to choose a subset $H$ of
the $H(v)$ such that $H$ consists of adjacent half-edges according to the cyclic order,
and $x^{-2}$ comes from the two new half-edges.
Similarly,
for every connected graph $\Gamma$,
the realization of $\cJ(\Gamma)$ is given by:
\ben
w_{\cJ(\Gamma)} = \sum_{v_1\in V(\Gamma)}\sum_{v_2\in V(\Gamma)\backslash\{v_1\}}
\val(v_1) \val(v_2) x^{-2} \cdot w_\Gamma.
\een
Therefore,
for a connected fat graph $\Gamma$ of type $(g,\mu)$ we have:
\ben
w_{\cS(\Gamma)}+w_{\cJ(\Gamma)}
= |\mu|(|\mu|+1) \cdot x^{-|\mu|}-2 t^{|F(\Gamma)|}
= \frac{d^2}{dx^2} w_\Gamma.
\een
And given two connected graphs $\Gamma_1$ and $\Gamma_2$
of types $(g_1,\mu)$ and $(g_2,\nu)$ respectively,
the realization of $\cJ(\Gamma_1,\Gamma_2)$ is:
\ben
w_{\cJ(\Gamma_1,\Gamma_2)} &= &
(\sum_{v_1\in V(\Gamma_1)}\val(v_1))
(\sum_{v_2\in V(\Gamma_1)}\val(v_2)) x^{-2}
\cdot w_{\Gamma_1}w_{\Gamma_2}\\
&=&
|\mu|\cdot |\nu|x^{-2}
\cdot w_{\Gamma_1}w_{\Gamma_2}\\
&=&\frac{d w_{\Gamma_1}}{dx} \frac{d w_{\Gamma_2}}{dx}.
\een
Thus the realization of the recursion \eqref{eq-abstract-qsc} gives the following:
\be
-x\frac{d}{dx} S_m^{\text{Herm}}
= \frac{d^2}{dx^2} S_{m-1}^{\text{Herm}}
+\sum_{j=0}^m \frac{d S_j^{\text{Herm}}}{dx} \frac{d S_{m-j}^{\text{Herm}}}{dx},
\qquad
\forall m\geq 1,
\ee
and for $m=0$:
\be
-x\frac{d}{dx} S_0^{\text{Herm}} = t + (\frac{d}{dx}S_0^{\text{Herm}})^2.
\ee
Or equivalently,
\be
\label{eq-Herm-qsc}
\bigg( \hbar^2 \frac{d^2}{dx^2} + \hbar x\frac{d}{dx} +t \bigg)
\exp\bigg( \sum_{m\geq 0} \hbar^{m-1} S_m^{\text{Herm}}(x) \bigg) =0,
\ee
where $\hbar$ is a formal variable.
Furthermore,
let us shift $S_0^{\text{Herm}}(x)$ by $\frac{x^2}{4}$ and denote:
\be
\widetilde S_m^{\text{Herm}}(x) :=
\frac{x^2}{4}\cdot\delta_{n,0} + S_m^{\text{Herm}}(x),
\ee
then the above recursion becomes:
\be
0= \frac{d^2}{dx^2} \widetilde S_{m-1}^{\text{Herm}}
+\sum_{j=0}^m \frac{d \widetilde S_j^{\text{Herm}}}{dx} \frac{d \widetilde S_{m-j}^{\text{Herm}}}{dx},
\qquad
\forall m\geq 1,
\ee
and
\be
(\frac{d}{dx}  \widetilde S^{\text{Herm}}_0 )^2
-\frac{1}{4}x^2 +t =0.
\ee
Or equivalently,
\begin{Theorem}
We have:
\be
\label{eq-Herm-qsc-2}
\bigg( \hbar^2 \frac{d^2}{dx^2} - \frac{1}{4}x^2 +t \bigg)
\exp\bigg( \sum_{m\geq 0} \hbar^{m-1} \widetilde S_m^{\text{Herm}}(x) \bigg) =0,
\ee
\end{Theorem}

\begin{Remark}
Taking a classical limit
\ben
\hbar\frac{d}{dx}\mapsto \frac{y}{\sqrt{2}},
\qquad\qquad
x\mapsto x
\een
in the quadratic operator above,
one obtains the following curve on the $(x,y)$-plane:
\ben
2y^2= x^2- 4t,
\een
which is the fat spectral curve \eqref{eq-1mm-fatspeccurve} of the Hermitian one-matrix models.
Using a result of the second author in \cite{zhou11},
the quadratic recursion \eqref{eq-EOtype-rec-1mm} is equivalent to the Eynard-Orantin
topological recursion on this fat spectral curve,
and then the equation \eqref{eq-Herm-qsc-2} becomes the quantum spectral curve
in the sense of Gukov-Su{\l}kowski \cite{gs}.
We will briefly explain this in \S \ref{sec-conj} as an example of our conjecture.

\end{Remark}

\subsection{A special case: Enumeration of fat graphs}

As pointed out in \S \ref{sec-realization-herm-1mm},
the Hermitian one-matrix model is a refinement of the enumeration of fat graphs
where the degree of the 't Hooft coupling constant $t=Ng_s$ encodes
the number of faces in a graph.
The the enumeration problem in \cite[\S 3-\S 4]{dmss}
and \cite[\S 3-\S 4]{ms} is the special case $t=1$.

In this case,
consider the following Feynman rule for $\Gamma$ of type $(g,\mu)$:
\be
\label{eq-FR-enumfat}
\Gamma\quad\mapsto\quad
w_\Gamma :=x_1^{-\mu_1}\cdots x_n^{-\mu_n},
\ee
and then $\tcW_{g,n}$ is realized by the following functions $F_{g,n}^C$:
\be
F_{g,n}^C(x_1,\cdots,x_n):=
\delta_{g,0}\delta_{n,1}\log x_1 +
\sum_{\mu\in \bZ_{>0}^n}D_{g,n}(\mu)\cdot
x_1^{-\mu_1}\cdots x_n^{-\mu_n}.
\ee
Using exactly the same method used in \S \ref{sec-realizationrec-npt},
we may obtain the following realization of the recursion \eqref{eq-rec-abstractEO-2}
(here we omit the details):
\begin{equation*}
\begin{split}
&\bigg(-x_1-2\frac{\pd}{\pd x_1}F_{0,1}^C (x_1)\bigg)\cdot
\frac{\pd}{\pd x_1}F_{g,n}^C (x_1,\cdots,x_n)\\
=&
\sum_{j=2}^n
\frac{1}{x_1-x_j}\bigg(
-\frac{\pd}{\pd x_j}F_{g,n-1}^C (x_2,\cdots,x_n)
+\frac{\pd}{\pd x_1}F_{g,n-1}^C (x_1,\cdots,\hat{x}_j,\cdots,x_n)
\bigg)\\
+&
E_{x_1,u,v}^C\bigg(
F_{g-1,n+1}^C(u,v,x_{[n]\backslash\{1\}})
+\sum_{\substack{g_1+g_2=g\\I\sqcup J=[n]\backslash\{1\}}}^s
F_{g_1,|I|+1}^C(u,x_I) F_{g_2,|J|+1}^C(v,x_J)
\bigg),
\end{split}
\end{equation*}
where $\frac{\pd}{\pd x_1}F_{0,1}^C(x_1)=-z_1$,
and $E_{a,b,c}^C$ is the operator:
\be
E_{a,b,c}^C f(b,c):=
\big(\lim_{b\to c}
\frac{\pd^2 f(b,c)}{\pd b \pd c}
\big)\big|_{c=a}.
\ee
After the following change of variables
(see \cite[\S 4]{dmss}):
\be
x_j=\frac{t_j+1}{t_j-1}+\frac{t_j-1}{t_j+1},
\ee
one recovers the recursion in \cite[Theorem 4.1]{ms}:
\be
\label{eq-qrec-mulase}
\begin{split}
&\frac{\pd}{\pd t_1}F_{g,n}^C(x_1,\cdots,x_n)\\
=&
\frac{(t_1^2-1)^2 (t_1^2-t_j) (t_j+1)}{16t_1^2(t_j^2-t_1^2)}
\frac{\pd}{\pd t_1} F_{g,n-1}^C(x_{[n]\backslash\{j\}})
-\frac{(t_j^2-1)^3}{16t_j (t_j^2-t_1^2)} \frac{\pd}{\pd t_j}
F_{g,n-1}^C(x_{[n]\backslash\{1\}})\\
&-\frac{(t_1^2-1)^3}{32t_1^2}\bigg(
\frac{\pd^2}{\pd t \pd t'}
F_{g-1,n+1}^C (x(t),x(t'),x_2,\cdots,x_n)
\bigg)\bigg|_{t=t'=t_1}\\
&-\frac{(t_1^2-1)^3}{32t_1^2}\cdot
\sum_{\substack{g_1+g_2=g\\I\sqcup J=[n]\backslash\{1\}}}^{stable}
\frac{\pd}{\pd t_1}F_{g_1,|I|+1}^C(x_1,x_I)
\frac{\pd}{\pd t_1}F_{g_2,|J|+1}^C(x_1,x_J),
\end{split}
\ee
where $\sum\limits^{stable}$ means excluding all terms involving
$F_{0,1}^C$ and $F_{0,2}^C$ in the summation.

Furthermore,
taking
\be
S_m^C (x) := \sum_{2g-1+n=m} \frac{1}{n!} F_{g,n}^C (x,\cdots,x)
=S_m^{\text{Herm}}(x)|_{t=1},
\qquad
m\geq 0,
\ee
then \eqref{eq-Herm-qsc} gives (\cite[Theorem 4.5]{ms}):
\be
\label{eq-ms-qsc}
\bigg( \hbar^2 \frac{d^2}{dx^2} + \hbar x\frac{d}{dx} +1 \bigg)
\exp\bigg( \hbar^{m-1} S_m^C (x) \bigg) =0.
\ee

\begin{Remark}
In this enumeration problem,
Dumitrescu et al. \cite{dmss, ms} have shown that the quadratic recursion \eqref{eq-qrec-mulase}
is equivalent to the E-O topological recursion on the Catalan curve:
\ben
x=z+\frac{1}{z},
\een
and then the above equation becomes the quantum spectral curve for the Catalan curve \cite{ms}.
In \cite{kls, cls},
Cutimanco et al. have modified the E-O topological recursion of Dumitrescu et al. and obtained the
E-O recursion and quantum spectral curve for the harmonic oscillator curve:
\ben
y^2 = x^2 - c^2,
\een
and related this problem to some other physical works and combinatorial problems.

Notice that the harmonic oscillator curve is a rescaling of the fat spectral curve
\eqref{eq-1mm-fatspeccurve} of the Hermitian one-matrix models
if one takes $c^2= t$ to be the 't Hooft coupling constant.
This gives an interpretation of the E-O invariants of the above curve
by the Hermitian one-matrix models and the corresponding refined enumeration problem of fat graphs.
\end{Remark}

\section{Conjectures Towards Eynard-Orantin Topological Recursion
and Quantum Spectral Curves}

\label{sec-conj}

As mentioned in the beginning of \S \ref{sec-abstract-eorec},
we want to relate the quadratic recursion relation in Theorem \ref{thm-abstract-eotype-3}
to the Eynard-Orantin topological recursion
and the quantum spectral curves.
In this section we propose some conjectures based on this idea.
These conjectures hold in the examples of Hermitian one-matrix models and
the enumeration of fat graphs.

\subsection{The main conjectures}
\label{sec-conj-mainconj}

Our first conjecture is the following:
\begin{Conjecture}
\label{conj-1}
Assign a suitable Feynman rule to fat graphs.
Let $\cC$ be the realization of the spectral curve \eqref{eq-abstract-spectralcurve}
for the abstract QFT,
and let $B(x_1,x_2)$ be the realization of the abstract Bergmann kernel $\cW_{0,2}$.
Then the Eynard-Orantin invariants associated to the spectral curve $\cC$ and
Bergmann kernel $B(x_1,x_2)$ are the realizations of
the abstract $n$ point functions $\cW_{g,n}$.
\end{Conjecture}

Recall that we have already shown that the abstract $n$-point functions $\cW_{g,n}$
satisfy the quadratic recursion relation in Theorem \ref{thm-abstract-eotype-3}.
Therefore,
a sufficient condition for the above conjecture to hold is that
the realization of \eqref{eq-rec-abstractEO-2} or \eqref{eq-rec-abstractEO}
is equivalent to the Eynard-Orantin topological recursion \eqref{eq-eorec} on the spectral curve $\cC$.
Thus it is natural to expect that:
\begin{Conjecture}
\label{conj-2}
The abstract $n$-point functions can be modified as multi-linear differentials
on the spectral curve \eqref{eq-abstract-spectralcurve} of the abstract QFT,
such that the quadratic recursion \eqref{eq-rec-abstractEO} can be reformulated in terms of
taking (formal) residues on this spectral curve.
Moreover,
the realization of this recursion is equivalent to the E-O topological recursion.
\end{Conjecture}

Furthermore,
when forgetting the labels (i.e., taking `averaging') on the vertices,
we may expect that:
\begin{Conjecture}
\label{conj-3}
The realization of the equation \eqref{eq-abstract-qsc}
is equivalent to the quantum spectral curve of the realization of
the curve \eqref{eq-abstract-spectralcurve}.
\end{Conjecture}

\subsection{Example: The Hermitian one-matrix models}

Conjecture \ref{conj-1} and \ref{conj-3} are known to be true
in this case of the Hermitian one-matrix models,
see \cite{zhou11}.

In fact,
perform a rescaling to \eqref{eq-1mm-fatspeccurve} and
rewrite the fat spectral curve \eqref{eq-1mm-fatspeccurve} in the following way:
\be
\label{eq-fatspec}
x^2-4y^2=4t,
\ee
where $t=Ng_s$ is the 't Hooft coupling constant.
This is an algebraic curve on the $(x,y)$-plane with two branch points.
Define a family of multi-linear differentials $\omega_{g,n}^{\text{Herm}} (p_1,\cdots,p_n)$ on
this curve by:
\be
\omega_{g,n}^{\text{Herm}} (p_1,\cdots,p_n):=W_{g,n}'(y_1,\cdots,y_n)dx_1\cdots dx_n,
\ee
where $W_{g,n}'$ are given by:
\be
\begin{split}
&W_{0,1}'(y_1):=-\frac{x_1}{2}+W_{0,1}^{\text{Herm}}(y_1),\\
&W_{0,2}'(y_1,y_2):=\frac{1}{(x_1-x_2)^2}+W_{0,2}^{\text{Herm}} (y_1,y_2),
\end{split}
\ee
and
$W_{g,n}'(p_1,\cdots,p_n):=W_{g,n}^{\text{Herm}} (y_1,\cdots,y_n)$
for $(g,n)\not= (0,1)$ or $(0,2)$.
Here $W_{g,n}^{\text{Herm}}$ are the realizations of $\cW_{g,n}$ defined in
\S \ref{sec-realization-1mm-EO}.
Now choose the Bergmann kernel on the spectral curve \eqref{eq-fatspec} to be:
\be
\label{eq-1mm-berg}
B^{\text{Herm}}(p_1,p_2):=
\omega_{0,2}^{\text{Herm}}(p_1,p_2)=
\frac{1}{2(x_1-x_2)^2}
\bigg(
1+\frac{x_1x_2-4t}{4y_1y_2}
\bigg)dx_1dx_2.
\ee
In \cite{zhou11} the second author proved that
the quadratic recursion \eqref{eq-EOtype-rec-1mm} is equivalent to the
following E-O topological recursion on the above spectral curve,
and then:
\begin{Theorem}
[\cite{zhou11}]
We have:
\be
\omega_{g,n}^{\text{Herm}}(p_1,\cdots,p_n)=
\omega_{g,n}^{EO}(p_1,\cdots,p_n)
\ee
for every pair $(g,n)\not=(0,1)$,
where $\omega_{g,n}^{EO}$ are the E-O invariants associated to
the spectral curve \eqref{eq-fatspec} and Bergmann kernel \eqref{eq-1mm-berg}.
\end{Theorem}

As a corollary,
the equation \eqref{eq-Herm-qsc-2} becomes the quantum spectral curve
of the fat spectral curve \eqref{eq-fatspec}.

\subsection{Example: Enumeration of fat graphs}

Similar results hold for the enumeration problem of fat graphs£¬
due to Dumitrescu et al. \cite{dmss, ms}.
This case can also be seen by taking $t=1$ in the Hermitian one-matrix models.

Denote by:
\be
D_{g,n}(\mu):=\sum_{\Gamma\in\mathfrak{Fat}_g^{\mu, c}}
\frac{1}{|\Aut(\Gamma)|}
\ee
the weighted number of fat graphs of a given type,
and let:
\be
F_{g,n}^C(x_1,\cdots,x_n):=
\delta_{g,0}\delta_{n,1}+
\sum_{\mu\in \bZ_{>0}^n}D_{g,n}(\mu)\cdot
x_1^{-\mu_1}\cdots x_n^{-\mu_n},
\ee
then the recursion \eqref{eq-qrec-mulase} is equivalent to the E-O topological recursion
on the Catalan curve (see \cite{dmss, ms}):
\ben
x=z+\frac{1}{z},
\een
which can be obtained from \eqref{eq-fatspec} by taking $t=1$ and shift $y$ to $y-\frac{x}{2}$.
Then the equation \eqref{eq-ms-qsc} is the quantum spectral curve of the
Catalan curve, see \cite[\S 4]{ms}.
In other words,
Conjecture \ref{conj-1} and \ref{conj-3} are both true in this case.

\vspace{.2in}
{\bf Acknowledgements}.
The second author is partly supported by NSFC grant 11661131005 and 11890662.

\end{document}